\documentclass[a4paper,11pt]{article}
\usepackage{jheppub} 
\UseRawInputEncoding
\usepackage{amsthm}
\newtheorem{theorem}{Theorem}
\newtheorem{lemma}[theorem]{Lemma}
\newtheorem{corollary}[theorem]{Corollary}
\theoremstyle{definition}
\newtheorem{definition}{Definition}[section]

\usepackage{enumerate}
\usepackage{lipsum}
\usepackage[utf8]{inputenc}

\usepackage{amssymb,amsmath,amsbsy,mathtools}
\usepackage{braket}
\usepackage{graphicx}
\usepackage{xcolor}

\usepackage{natbib}
\usepackage{float}

\newcommand{\bO}{\mathbb{O}}
\newcommand{\mX}{\mathcal{X}}
\newcommand{\ba}{{\bf a}}
\newcommand{\bb}{{\bf b}}
\newcommand{\mA}{\mathcal{A}}
\newcommand{\mB}{\mathcal{B}}

\newcommand{\mE}{\mathcal{E}}
\newcommand{\mF}{\mathcal{F}}
\newcommand{\mH}{\mathcal{H}}

\newcommand{\mY}{\mathcal{Y}}

\newcommand{\mK}{\mathcal{K}}

\newcommand{\mM}{\mathcal{M}}

\newcommand{\mO}{\mathcal{O}}

\newcommand{\mS}{\mathcal{S}}
\newcommand{\mT}{\mathcal{T}}
\newcommand{\lb}{\left(}
\newcommand{\ep}{\epsilon}
\newcommand{\rb}{\right)}

\newcommand{\p}{\partial}

\newcommand{\nn}{\nonumber}
\newcommand{\ee}{\end{equation}}

\title{Information loss, mixing and emergent type III$_1$ factors}
\author{Keiichiro Furuya, Nima Lashkari, Mudassir Moosa, Shoy Ouseph}
\affiliation{Department of Physics and Astronomy, Purdue University, West Lafayette, IN 47907, USA}
\emailAdd{kfuruya@purdue.edu}
\emailAdd{nima@purdue.edu}
\emailAdd{mudassir@purdue.edu}
\emailAdd{souseph@purdue.edu}

\abstract{
A manifestation of the black hole information loss problem is that the two-point function of probe operators in a large Anti-de Sitter black hole decays in time, whereas, on the boundary CFT, it is expected to be an almost periodic function of time. We point out that the decay of the two-point function (clustering in time) holds important clues to the nature of observable algebras, states, and dynamics in quantum gravity. 

We call operators that cluster in time ``mixing" and explore the necessary and sufficient conditions for mixing. The information loss problem is a special case of the statement that in type I algebras, there exists no mixing operators.
We prove that, in a thermofield double state (KMS state), if mixing operators form an algebra (close under multiplication), the resulting algebra must be a von Neumann type III$_1$ factor. In other words, the physically intuitive requirement that all nonconserved operators should exponentially mix is so strong that it fixes the observable algebra to be an exotic algebra called a type III$_1$  factor. More generally, for an arbitrary out-of-equilibrium state of a general quantum system (von Neumann algebra), we show that if the set of operators that mix under modular flow forms an algebra, it is a type III$_1$ von Neumann factor. 

In a theory of Generalized Free Fields (GFF), we show that if the two-point function clusters in time, all operators are mixing, and the algebra is a type III$_1$ factor. For example, in $\mathcal{N}=4$ SYM, above the Hawking-Page phase transition, clustering of the single trace operators implies that the algebra is a type III$_1$ factor, settling a recent conjecture of Leutheusser and Liu.
We explicitly construct the C$^*$-algebra and von Neumann subalgebras of GFF associated with time bands and, more generally, open sets of the bulk spacetime using the HKLL reconstruction map.}

\begin{document}
                              
\maketitle
\flushbottom

\section{Introduction} 

Consider a burning piece of coal described by a pure state $\ket{\Psi}$ of a microcanonical ensemble of average energy $E_0$. In an interacting system, perturbations to the system thermalize and we expect that under unitary time-evolution all connected correlators between early and late-time perturbations decay away exponentially fast in time (exponentially mixing):
\begin{eqnarray}\label{purestateevol}
    \braket{\Psi|a\: b(t)|\Psi}_{conn}\sim e^{-\alpha t E_0}
\end{eqnarray}
for some positive constant $\alpha$.
This is equivalent to saying that the correlator of one-point subtracted observables ${\bf a}:=a-\braket{\Psi|a|\Psi}$ should decay away. Similarly, in the canonical ensemble, in a thermalizing system, the connected thermal two-point functions are expected to decay exponentially fast at the late time:
\begin{eqnarray}\label{2ptthermall}
    \braket{{\bf a}\: {\bf b}(t)}_\beta\sim e^{-\alpha t/\beta}\ .
\end{eqnarray}
We call the decay of the connected correlator {\it clustering in time}.
For more complicated operators, the decay of connected correlators in time need not be exponential. If finite-temperature correlators decay away in time, we call the corresponding operators ``a" or ``b" {\it mixing}, and if the decay is exponentially fast we call them {\it exponentially mixing}.\footnote{In a spatially extended quantum system, the term thermalization time can refer to two distinct time scales: diffusion time (sometimes called pre-thermalization or local thermalization time), and global thermalization time (intimately tied to the scrambling time). In this work, we are only concerned with the time dependence of the observables in the strict thermodynamic limit. Hence, diffusion time is the only scale in the problem.} More generally, for arbitrary out-of-equilibrium state $\psi$, we define {\it $\psi$-mixing} operators and {\it exponentially $\psi$-mixing} operators as those that mix and mix exponentially fast, respectively, under the modular flow of the state $\psi$. 

The thermodynamic information loss problem is the statement that if the Hamiltonian has a discrete spectrum, due to Poincar\'e recurrences, there exist {\it no} pair of observables $a$ and $b$ whose correlator clusters in time. In particular, this implies that there cannot exist any $\psi$-mixing operators for any $\psi$. Stated in more mathematical terms, when the spectrum of the modular Hamiltonian is discrete, the correlator in (\ref{purestateevol}) for any pair of operators is an almost periodic function of time and has no long time limit. For instance, consider a 2d CFT and the two-point function of Virasoro primaries $\mO$ of dimension $\Delta$ at finite temperature in infinite volume:
\begin{eqnarray}\label{strongmixing}
    \braket{\mO(t)\mO(0)}_\beta=\lb\frac{\pi}{\beta\sinh(\pi t/\beta)}\rb^{2\Delta_\mO}\ .
\end{eqnarray}
Since the correlator decays exponentially fast, $\mO$ is an exponentially mixing operator.
Intuitively, one can think of type I algebras as those in which the spectrum of the modular Hamiltonian (logarithm of the density matrix) of any state is discrete.\footnote{To help the readers unfamiliar with the basics of operator algebras and the classification of von Neumann algebras we provide a quick review of the definitions and key intuitions in appendix \ref{app:op}.}
Therefore, (\ref{strongmixing}) already implies that the algebra of observables of a 2d CFT at finite temperature and infinite volume cannot be type I. This is as opposed to the same correlator on a spatial circle of circumference $L$ at zero temperature:
\begin{eqnarray}
  \braket{\mO(t)\mO(0)}_L=\lb\frac{\pi}{L\sin(\pi t/L)}\rb^{2\Delta_\mO}\ .
\end{eqnarray}
In this case, the correlator is periodic in time with periodicity $L$ and the algebra is type I.

The resolution to the information loss problem posed above is that in the thermodynamic limit (a large number of particles $N$, infinite volume, or vanishing $\hbar$), the entropy grows to infinity and the spectrum of the Hamiltonian in the high energies becomes continuous, allowing the correlators of certain distinguished {\it simple operators} to cluster in time. In the strict infinite-entropy limit, the Hilbert space is inseparable. There are many inequivalent superselection sectors, and the physical states of each sector correspond to finite energy excitations on top of the ``vacuum" of that sector. Each sector is a separable Hilbert space $\mH_\psi$ labeled by the vacuum $\psi$ that we represent in the Hilbert space as $\ket{1}_\psi$. For instance, choosing a (KMS) state of inverse temperature $\beta$ as the ``vacuum" the corresponding vector $\ket{1}_\beta$ is the thermofield double. It is an old (but less known) result in mathematical physics that the finite-temperature observable algebras generated by local operators in the strict thermodynamic limit of infinite volume are generically type III von Neumann algebras \cite{hugenholtz1967factor}. Intuitively, one can think of type III algebras as those that admit no density matrices. To state this result more precisely, we first notice that, for a generic Hamiltonian, we expect the energy levels  to be incommensurable.\footnote{In an infinite-dimensional system with no normalizable energy eigenvectors this condition is replaced by the assumption that (modular) time-evolution is a strongly continuous unitary flow.} Then, it follows from von Neumann's ergodic theorem that the time-averaged thermal correlators cluster 
\begin{eqnarray}\label{weakmixing}
    \lim_{T\to \infty}\frac{1}{T}\int_0^T dt \braket{\ba(t)\bb(0)}_{\beta}=0\ .
\end{eqnarray}
The work of \cite{hugenholtz1967factor} establishes that this occurs only if either the algebra is type III or the temperature is infinite.\footnote{For a summary of the proof see Theorem 1 of \cite{longo1979notes}.} In appendix \ref{app:typeiii}, we give an alternative argument based on the intuition that in the thermodynamic limit of infinite entropy (infinite volume or infinite $N$) the Hamiltonian is not observable, either because its fluctuations are infinite, or its action on the thermofield double $\ket{1}_\beta$ leaves the Hilbert space of the sector $\mH_\beta$. In Lemma \ref{KMStypeIII}, we show that this occurs only in type III algebras.\footnote{In light of the work in \cite{witten2022gravity}, this already implies that the GFF algebra in the KMS state above the Hawking-Page phase transition is type III.} Operators that satisfy (\ref{weakmixing}) are often called {\it ergodic} as opposed to mixing (or {\it strong mixing}) operators whose definition involves the $t\to\infty$ limit without any time averaging. In this work, we focus on strong mixing operators, which we simply refer to as mixing. 

In a thermal (KMS) state, the operators that commute with the Hamiltonian (conserved charges) play a distinguished dynamical role because they cannot decay in time. If they are observables (having finite fluctuations and their domain and range are contained within the Hilbert space), they form a subalgebra of observables that is the algebra of the zero modes (i.e. operators that do not evolve in time). The intuition is that in a truly interacting theory, any operator that is not conserved should be mixing (clusters in time).\footnote{This intuition underlies the emergence of hydrodynamics as the low-frequency effective theory of generic interacting theories.} In a sector defined with fixed charges, all operators should be mixing.
In this work, we argue that this is the smoking gun of a type III$_1$ observable algebra. Our key result can be summarized as the statement that if the set of mixing operators in the KMS state (more generally $\psi$-mixing operators) forms an algebra (close under multiplication), the resulting algebra is a von Neumann factor of type III$_1$ that contains no (modular) conserved charges. The set of (modular) conserved charges forms a subalgebra that, in mathematical terminology, is called the centralizer of the state $\psi$.\footnote{In the language of \cite{chandrasekaran2022large}, this is the algebra at infinity in time.} A type III$_1$ factor can be understood as an observable algebra where the modular flow of no state has any Poincar\'e recurrences. The intimate connection to mixing and the emergence of the hydrodynamic limit and the emergence of a local bulk in holography are precisely what make type III$_1$ von Neumann algebras physically unique and valuable to study. It is important to point out that type III$_1$ algebras can, and in general, have many states with non-trivial centralizers. 

As an example, once again consider a 2d CFT in infinite volume and finite temperature. We saw that all Virasoro primaries are mixing operators.
We will see in Lemma \ref{mixingop} that any countable sum of mixing operators is mixing. The operator product expansion (OPE) tells us that the multiplication of primary operators can be expanded as a formal sum over conformal primaries. Therefore, any two-point thermal function in this CFT can be expanded formally as a countable sum over exponentially decaying terms
\begin{eqnarray}
    \braket{a(t)b(0)}_\beta=\sum_\Delta c^{(\Delta)}_1 c_2^{(\Delta)} \lb\frac{\pi}{\beta\sinh(\pi t/\beta)}\rb^{2\Delta_\mO}\ .
\end{eqnarray}
The sum above is convergence if the time evolution from $a$ to $a(t)$ does not cross the lightcone of $b$. In the double thermofield, choosing $a(t)$ on the left boundary and $b$ on the right boundary, we can guarantee that $a(t)$ does not cross the lightcone of $b$. This suggests that in the KMS state, all operators are mixing, and then it follows from our Theorem \ref{clustering-thm} that the algebra of observables is type III$_1$ von Neumann factor.\footnote{Note that we are talking about the global algebra (the algebra of the entire spacetime) here. The von Neumann algebra of local causally complete regions in any QFT is type III$_1$ factor.}

Two-dimensional CFTs are special in that their thermal one-point functions at infinite volume vanish. In discussing mixing in a general theory, it is convenient to subtract the one-point functions from the local operators and define $\mathbb{O}=\mO-\braket{\mO}_\beta$. For instance, consider the four-point function of a free theory in higher than two dimensions. This correlator expressed in terms of the one-point removed operators makes the Wick contraction rule manifest (see Lemma \ref{isser}). The correlator
$\braket{\mathbb{O}(t_1)\mathbb{O}(t_2)\mathbb{O}(t_3)\mathbb{O}(t_4)}_\beta$ comprises three terms corresponding to different contractions $(12)(34)$, $(13)(24)$ and $(14)(23)$ which implies that the correlator $\braket{\mathbb{O}^2(t)\mathbb{O}^2(0)}_\beta$ clusters if the two-point function clusters. Therefore, it is easier to consider the operators $\mathbb{O}$ as the generator of the observable algebra.
This distinction is irrelevant in free field theory (generalized free fields) and in large $N$ theories because the thermal one-point functions vanish. 

From the point of view of quantum gravity as a unitary theory, an evaporating black hole is similar to a burning piece of coal. Consider a young black hole in an arbitrary state $\ket{\Psi}=\sum_\alpha c_\alpha \ket{E_\alpha}$ that belongs to a microcanonical ensemble of energy $E_0$. All correlators between early time probes $\mO$ acting during the collapse, and $\mO(t)$ after the black hole is formed decay exponentially in time:
\begin{eqnarray}
    \braket{\Psi|\mO(t)\: \mO(0)|\Psi}\sim e^{-\alpha t E_0} \ .
\end{eqnarray}
Maldacena pointed out that in AdS/CFT, a bulk gravitational calculation reveals that light conformal primaries are mixing which is in sharp contradiction with the expectation of a discrete spectrum of the Hamiltonian for the boundary CFT on a compact manifold. He argued that this is a manifestation of black hole information loss for an eternal (non-evaporating) black hole \cite{maldacena2003eternal}.


In Section \ref{sec2}, we review the analytic properties of the thermal two-point functions of the form \eqref{strongmixing}. We then review Maldacena's argument and observe that the von Neumann algebra of observables generated by the single-trace observables of $SU(N)$ $\mathcal{N}=4$ SYM in the limit of $N\to \infty$ in the thermofield double state above the Hawking-Page phase transition is not type I. This is implied by the existence of mixing operators. 
We end section \ref{sec2} by generalizing our discussion from the thermofield double (KMS state) and mixing under time-evolution (its modular flow) to a general out-of-equilibrium state and mixing under its corresponding modular flow.

In Section \ref{sec-info-loss}, we study mixing operators in a general quantum system. In particular, in Lemmas \ref{mixingop} and \ref{maximalmixingsystem}, we show that the set of all mixing observables forms an {\it operator system} (a $*$ closed subspace of observables that includes the identity operator) with the key property that it is preserved under time evolution. In Lemmas \ref{RB} and \ref{smooth}, we derive a necessary condition for a pair of observables to be mixing and the necessary and sufficient condition for them to be exponentially mixing. Then we focus on the strict thermodynamic or the large $N$ limit. In this limit, we expect the spectrum of Hamiltonian to become continuous at high energies, and in Lemma \ref{simpleop}, we explicitly construct mixing operators. We also draw parallels between these mixing operators and the {\it simple operators} of the eigenstate thermalization hypothesis (ETH). In this limit, we also expect that the operator system of mixing observables forms an algebra. In Lemma \ref{mixingvN}, we show that such an algebra can always be completed to a von Neumann algebra of mixing operators. Then, our key results in Theorem \ref{clustering-thm} and Corollary \ref{KSMtypeiii} show that the resulting von Neumann algebra must be a type III$_1$ factor in a state with a trivial centralizer. 

In section \ref{GFFsec}, we focus on the algebra Generalized Free Field (GFF). In a theory of GFF, the two-point function of the fundamental field fixes all the higher-point functions by Wick's theorem (see Lemma \ref{isser}). In Theorem \ref{GFFtypeiii1}, we show that if the fundamental field is a mixing operator, all operators in the von Neumann algebra of observables are mixing, and hence, in this case, the algebra of GFF is a type III$_1$ factor with a trivial centralizer. 

Our result in Theorem \ref{GFFtypeiii1} applies to the GFF algebra of single-trace observables of $SU(N)$ $\mathcal{N}=4$ SYM in the limit of $N\to \infty$ in the double thermofield state above the Hawking-Page phase. An insight into the nature of this GFF algebra comes from exact quantum error correction in holography. Complementary recovery implies that the map that projects the left (right) boundary algebra to the bulk algebra outside of the left (right) AdS black hole is a normal conditional expectation \cite{furuya2022real,faulkner2020holographic}. However, normal conditional expectations exist only between von Neumann algebras of the same type \cite{peterson2013notes,furuya2022real}. Since the bulk algebra outside of the horizon of either black hole is type III$_1$, the boundary GFF algebra is either type III$_1$ or exact quantum error correction does not hold in the $N\to \infty$ limit. Recently, Leutheusser and Liu conjectured that the GFF algebra above is type III$_1$ \cite{leutheusser2021causal,leutheusser2021emergent}. More specifically, they conjectured that in a GFF, the spectral density of the fundamental field is smooth and supported everywhere in real frequencies if and only if the resulting algebra is type III$_1$ (see Section 2.3 of \cite{leutheusser2021emergent}). We settle the Leutheusser-Liu conjecture and generalize to arbitrary states (in or out-of-equilibrium) of GFF in Theorem \ref{GFFtypeiii1} and Lemma \ref{measurableresult}. We find that if the spectral density is a Lebesgue-measurable function of real frequencies, then the algebra is type III$_1$ factor. The measurability condition is significantly weaker than the smoothness condition conjectured in \cite{leutheusser2021emergent}.\footnote{In fact, this condition is so weak that one has to use the uncountable axiom of choice to construct examples of non-measurable functions. A real non-Lebesgue-measurable function, by definition, sends some non-measurable set of $\mathbb{R}$ to a measurable set of $\mathbb{R}$. For instance, the characteristic function of a nonmeasurable set $X$ is a measurable function because it sends a nonmeasurable set $X$ to the measurable set $\{1\}$. See Vitali sets as elementary examples of nonmeasurable subsets of real numbers. From a physics standpoint, the spectral density is either a countable sum of Dirac delta functions (a tempered distribution) in type I algebras, or it is expected to be a measurable function of real frequencies.} Furthermore, we do not need an assumption on the support of spectral density in the frequency space.\footnote{See the discussion below Lemma \ref{measurableresult} for examples of type III$_1$ GFF algebras with spectral density supported on a finite frequency interval.} The key idea of our proof is that if all operators are $\psi$-mixing we have a type III$_1$ factor, however, the converse is incorrect. Type III$_1$ algebras can, and in general do have many states with nontrivial centralizers (nonmixing observables). That is why, as opposed to the conjecture by Leutheusser and Liu, we do not attempt to prove an if-and-only-if statement here. However, it is plausible that requiring the type III$_1$ algebra to be hyperfinite requires a stronger smoothness assumption on the spectral density similar to the Leutheusser-Liu conjecture. We postpone studying this problem for upcoming work.

Furthermore, Leutheusser and Liu proposed that there are von Neumann algebras associated with time intervals (time band algebras) in any theory of GFF \cite{leutheusser2021emergent}). We end the section \ref{GFFsec} by rigorously constructing the time-band algebras of GFF. We observe an ambiguity (choice) in the definition of these algebras that leads to many inequivalent C$^*$ and von Neumann algebras associated with time bands. For the choice suggested by holography, we construct the boundary C$^*$-algebras associated with arbitrary open sets in the bulk spacetime.

We conclude with a summary and discussion in Section \ref{sec:conclude} and include various appendices including Appendix \ref{app:notation} where we summarize the notation used in this work.

\section{Discrete energy gaps (type I algebras) forbid mixing}\label{sec2}

In this section, we first review the analytic properties of thermal correlators. The discussion is presented in a form that makes the generalization to modular correlators in an arbitrary out-of-equilibrium state straightforward. The starting point is the observation that if the energy gaps are discrete for any pair of operators, the thermal correlators are ($B^2$-Besicovitch) almost periodic functions and cannot decay in time. In AdS/CFT, the thermal correlator corresponding to light primary operators can be computed in the bulk using gravity. If the boundary CFT lives on a compact manifold, the energy spectrum is discrete, and therefore both the {\it energy differences} $E_m-E_n$ and the {\it energy gaps $E_{m+1}-E_m$} are discrete.\footnote{For a generic interacting infinite dimensional system, it is expected that the energies are incommensurable. Therefore, the energy differences are discrete but dense in $\mathbb{R}$. This is tied to the ergodicity and weak mixing in (\ref{weakmixing}).}
Hence, the thermal correlator is an almost periodic function of time which is in sharp contrast with the gravity answer that decays exponentially fast and clusters in time. In \cite{maldacena2003eternal}, Maldacena argued that the mismatch of the gravity answer with the boundary expectation is a manifestation of the black hole information loss problem for eternal black holes in Anti-de Sitter (AdS) space. Here, we generalize Maldacena's observation to the statement that in an arbitrary state $\psi$ (in or out of equilibrium) when the modular Hamiltonian $\hat{K}_\psi$ has a discrete spectrum, we have a type I von Neumann algebra. This forbids the existence of any modular mixing operator.

The resolution of this seeming paradox is tied to the large $N$ limit. In AdS/CFT the boundary theory has a large central charge $N^2$. At high energies $E\sim O(N^2)$ the energy gaps are expected to be $O(e^{-N^2})$\cite{festuccia2007arrow}. In the strict $N\to \infty$, even on a compact spatial manifold, the gaps close and the spectrum becomes continuous. In the order of limits, which $N\to \infty$ first and then $t\to \infty$, certain distinguished simple operators (gravity operators) cluster in time.

\subsection{Thermal correlators and clustering in time}

Given a Hamiltonian $H$ with energy eigenstate $\ket{E_{m}}$ and corresponding eigenvalues $E_m$, the thermofield double state (TFD) is defined on two copies of the Hilbert space, $\mathcal{H}_{L}\otimes\mathcal{H}_{R}$, as 
\begin{eqnarray}
    \ket{1}_\beta := Z^{-1/2} \sum_m e^{-\beta E_{m}/2}\ket{E_{m}}_L \ket{E_{m}}_R\ , \label{eq-tfd}
\end{eqnarray}
where $Z = \sum_{m}e^{-\beta E_{m}}$ is the thermal partition function and $\beta$ denotes inverse temperature. This state is invariant under the unitary evolution generated by the {\it modular Hamiltonian} $\hat{K} = H_{L} - H_{R}$, where in our notation $a_{L}=a\otimes 1\in \mA_L$ and $a_R=1\otimes a\in \mA_R$, and $\mA_L$ and $\mA_R$ are the left and right algebras, respectively. However, the evolution by $(H_L+H_R)/2$ is equivalent to the ``one-sided'' evolution generated by $H_{L}$ or by $H_{R}$:
\begin{eqnarray}
    e^{i (H_{L}+H_{R}) t/2}\ket{1}_\beta = e^{i H_L t}\ket{1}_\beta = Z^{-1/2} \sum_m e^{-E_{m}(\beta/2-it)}\ket{E_{m}}_L\ket{E_{m}}_R\ .
\end{eqnarray}

We are interested in the analytic properties of the following {\it left-right} (LR) time-evolved correlation functions 
\begin{align}\label{fabdef}
    f_{ab}(t) & := \mbox{}_{\beta}\bra{1} a_L^\dagger  b^{T}_R(t) \ket{1}_{\beta} =\mbox{}_{\beta}\bra{1} a_L^\dagger  e^{it\hat{K}} b^{T}_R \ket{1}_{\beta} \, , 
\end{align}
where $b^{T}$ is the transpose of $b$ in the energy eigenbasis, and we have defined the time-evolved operator
\begin{eqnarray}
    b(t):=e^{i t \hat{K}}b e^{-it \hat{K}}\ .
\end{eqnarray}
Expanded in energy eigenbasis we have
\begin{align}
    f_{ab}(t) =&\, \mbox{}_{\beta}\bra{1} e^{-itH_{L}} \, a^{\dagger} \otimes  b^{T} \, e^{itH_{L}}\ket{1}_{\beta} \, , \nonumber\\
    =&\, Z^{-1} \sum_{m,n}e^{ -\beta (E_m+E_n)/2}e^{i(E_m-E_n)t} a^*_{mn} b_{mn} \, ,\label{eq-fab-expression}
\end{align}
where $a_{mn} = \bra{E_{m}} a \ket{E_{n}}$, and similarly for $b$. The correlator $f_{ab}(t)$ is analytic inside the complex strip $\Im(t)\in [-\beta/2,\beta/2]$ with continuous boundary values\footnote{Note that time-evolution is a strongly continuous unitary flow.} corresponding to {\it left-left} (LL) correlators
\begin{eqnarray}\label{analext}
    &&f_{ab}(t-i\beta/2) = \mbox{}_{\beta}\bra{a} e^{it\hat{K}} \ket{b}_{\beta} \, , \nn\\
    &&f_{ab}(t+i\beta/2) = \mbox{}_{\beta}\bra{b^{\dagger}} e^{-it\hat{K}} \ket{a^{\dagger}}_{\beta} \ . 
\end{eqnarray}
We are using the notation $\ket{a}_{\beta} :=\, a_{L} \ket{1}_{\beta} $ and similarly for $\ket{b}_\beta$.
The two boundary values above are related to each other by the KMS condition \cite{kubo1957statistical,martin1959theory}.  

We say that operators $a$ and $b$ are mixing if the $f_{ab}(t)$ clusters in time. That is,
\begin{align}
    \lim_{t\to\infty} f_{ab}^{\text{conn}}(t) \, = \, 0 \, ,
\end{align}
where 
\begin{align}
    f_{ab}^{\text{conn}}(t) \, := \, f_{ab}(t) - {}_\beta\braket{1|a_{L}|1}_\beta {}_\beta\braket{1|b^{T}_{R}|1}_\beta 
\end{align}
is the connected part of the thermal two-point function. Therefore, in discussing mixing, it is convenient to introduce the following notation for the one-point subtracted operators 
\begin{eqnarray}
   \ba = a-{}_\beta\braket{1|a|1}_\beta\ ,
\end{eqnarray}
so that 
\begin{eqnarray}
&&f_{\ba\bb}(t)=f^{\text{conn}}_{ab}(t)={}_\beta\braket{1|a_L b_R^T(t)|1}_\beta - {}_\beta\braket{1|a_L|1}_\beta {}_\beta\braket{1|b_R^T|1}_\beta\ .
\end{eqnarray}
The correlator in the Fourier space is a tempered distribution 
    \begin{eqnarray}\label{fomega}
        \hat{f}_{ab}(\omega) \, = &&\, Z^{-1} \sum_{m,n} e^{-\beta(E_m+E_n)/2} \, a^*_{mn}b_{mn} \, \delta\left(\omega - E_{m} + E_{n}\right) \ .
    \end{eqnarray}
    We often consider the observables $a=b$ in (\ref{fomega}) so that the Fourier transform of $f_{aa}(t)$ is positive.\footnote{Functions with positive Fourier transform are called positive definite. For a discussion of positive definite 
functions and more special cases of $f_{ab}(t)$ see appendix \ref{app:special2pt}.} The relations in (\ref{analext}) show that the Fourier transform of the LL correlators are related to those of LR by multiplication by $e^{\pm \beta \omega/2}$. In frequency space, the KMS condition relates the time-ordered correlator to the commutator:
\begin{eqnarray}\label{KMSfreq}
    \braket{a(\omega)b}_\beta=\frac{1}{e^{\beta \omega}-1}\braket{[b,a(\omega)]}_\beta\ .
\end{eqnarray}

Alternatively, we can think of $\hat{f}_{ab}(\omega)$ 
 in terms of the spectral projections and measure of the modular Hamiltonian $\hat{K}$:
\begin{eqnarray}
    &&\hat{f}_{ab}(\omega)={}_\beta\braket{1|a_L dP_\omega b_R^T|1}_\beta\ . \label{fouriercont}
\end{eqnarray}
Since $dP_\omega$ is a positive operator we can write it as $dP_\omega=Q^\dagger Q$, then the Cauchy-Schwarz inequality implies
\begin{eqnarray}\label{Cauchyineq}
    |\hat{f}_{ab}(\omega)|=|{}_\beta\braket{1|(a_L Q^\dagger)(Q b_R^T)|1}_\beta|&&\leq |{}_\beta\braket{1|(a_L Q^\dagger)(Q a_L^\dagger)|1}_\beta|^{1/2}|{}_1\braket{1|((b_R^*) Q^\dagger)(Q b_R^T)|1}_\beta|^{1/2},\nn\\
    &&=|\hat{f}_{aa}(\omega)|^{1/2}|\hat{f}_{bb}(\omega)|^{1/2}
\end{eqnarray}
where in the last step we have used 
\begin{eqnarray}
    &&\braket{1|a_L dP_\omega a_L^\dagger|1}=e^{\beta \omega/2}\hat{f}_{aa}(\omega)\nn\\
  &&  \braket{1|b_R^* dP_\omega (b_R^T)|1}=e^{-\beta \omega/2}\hat{f}_{bb}(\omega)\ .
\end{eqnarray}


\begin{lemma}\label{almostperiod}
If the spectrum of the (modular) Hamiltonian is discrete then $f_{ab}(t)$ is a $B^2$-Besicovitch almost periodic function of time. 
\end{lemma}
\begin{proof}
A $B^2$-Besicovitch almost periodic function can be defined as the formal series 
\begin{align}
    f(t) = \sum_{k} \xi_{k} e^{i\lambda_{k}t} \, ,
\end{align}
where $\sum_{k} |\xi_{k}|^{2} < \infty$ and real $\lambda_{k}$ \cite{besicovitch1926generalized}. 
Recall from \eqref{eq-fab-expression} that $f_{ab}(t)$ can be written as 
\begin{align}
    f_{ab}(t) = \sum_{m}\sum_{n} \xi^{(ab)}_{mn} e^{i(E_m-E_n)t} \, ,
\end{align}
where
\begin{align}
    \xi^{(ab)}_{mn} \, = \, Z^{-1} e^{ -\beta (E_m+E_n)/2} a^*_{mn} b_{mn} \, .
\end{align}
Therefore, to show that $f_{ab}(t)$ is a $B^2$-Besicovitch almost periodic, we need to show that
\begin{align}
    \sum_{m}\sum_{n} |\xi^{(ab)}_{mn}|^{2} < \infty \, .
\end{align}
First, consider the case where $a = b$. In this case, $\xi^{(aa)}_{mn} > 0$, and hence,
\begin{align}
    \sum_{m}\sum_{n} |\xi^{(aa)}_{mn}|^{2} \leq \left( \sum_{m}\sum_{n} \xi^{(aa)}_{mn} \right)^{2} = \left( f_{aa}(t=0) \right)^{2} < \infty , \label{eq-almost-periodic-cond-1}
\end{align}
where we have assumed that the correlation function without modular flow is finite. Next, we consider the general case where $a \ne b$. In this case, 
\begin{align}
    \sum_{m}\sum_{n} |\xi^{(ab)}_{mn}|^{2} =  \sum_{m}\sum_{n} \xi^{(aa)}_{mn} \xi^{(bb)}_{mn}  \leq \, \sqrt{\sum_{m}\sum_{n} |\xi^{(aa)}_{mn}|^{2} } \, \sqrt{ \sum_{m}\sum_{n} |\xi^{(bb)}_{mn}|^{2} } < \infty \, ,
\end{align}
where we have used Cauchy-Schwarz inequality in the second step and \eqref{eq-almost-periodic-cond-1} in the third step. 
\end{proof}
An almost periodic function does not have a long time limit and cannot cluster in time. In other words, as long as the spectrum of energy gaps is discrete, the algebra is type I, and there exist no mixing operators.

\subsection{Information loss in eternal black holes}
Consider a pair of holographic CFTs on compact spaces (e.g. spheres) in the thermofield double state $\ket{1}_{\beta}$ of (\ref{eq-tfd}).
At any large but finite $N$, the energy spectrum $E_m-E_n$ and energy gaps are both discrete. Therefore, for any probe operators $\mO_L$ and $\mO_R$, the $LR$ thermal correlator $f_{\mO\mO}(t)$ is an almost-periodic function.\footnote{This argument can be repeated for $LL$ or $RR$ correlators as well.} However, this contradicts the gravity result.

For simplicity, let us focus on AdS$_{3}$/CFT$_{2}$. 
Choose as probe a Virasoro conformal primary $\mO$ of dimension $\Delta$ dual to a scalar field $\phi$ of mass $m^2=\Delta(2-\Delta)$ in the bulk. 
For fixed inverse temperature $\beta$, the boundary theory correlator corresponds to a quantum gravity path-integral. At large $N$ this path-integral is approximated by the saddle point approximation. Both the thermal AdS and BTZ black hole and, in fact, an $SL(2,\mathbb{Z})$ worth of saddle points contribute to the path-integral \cite{maldacena1999ads3,dijkgraaf2000black}. These saddles correspond to all geometries that fill in the interior of the Euclidean boundary torus. Above the Hawking-Page phase transition ($\beta<2\pi$) the BTZ saddle dominates, whereas for lower temperatures ($\beta>2\pi$) the thermal AdS saddle wins. Since we are interested in the $N\to \infty$ ($G_N\to 0)$ limit, we focus on the thermal AdS and the BTZ saddles.

In gravity, we compute the time-dependent thermal correlator $f_{\mO\mO}(t)$ using the two-point functions of $\phi$ in either the BTZ background or the thermal AdS depending on which saddle dominates. Setting the angular momentum to zero, we find the following LR-correlators on a spatial circle of circumference $2\pi$ at inverse temperature $\beta$
\cite{keski1999bulk,maldacena2003eternal}
\begin{align}
    \text{BTZ: } \quad f_{\mathcal{O}\mathcal{O}}(\theta,t) \, \sim& \, \sum_{n=-\infty}^\infty \lb \cosh\lb \frac{2\pi}{\beta}(\theta+2\pi n)\rb+\cosh\lb \frac{2\pi t}{\beta}\rb\rb^{-2\Delta} \, ,\label{thermal2pt}\\
        \text{Thermal AdS: } \quad f_{\mO\mO}(\theta,t) \, \sim& \, \sum_{m=-\infty}^\infty \lb \cos(\theta)+\cos\lb\frac{2\pi}{\beta}(t+im\beta )\rb\rb^{-2\Delta} \, \label{thermalads},
\end{align}
where $(\theta,t)$ are coordinates on $S^1\times \mathbb{R}_t$. In the BTZ answer, the $n=0$ term corresponds to the finite temperature two-point function when the boundary is $\mathbb{R}$, and the sum over $n$ comes from the method of images compactifying the spatial $\mathbb{R}$ to $S^1$ by identifying $\theta\sim \theta+2\pi$. 
In the thermal AdS answer, on the other hand, the $m=0$ corresponds to the zero temperature two-point function on a circle, and the sum over $m$ comes from the method of compactifying the time $\mathbb{R}_t$ to $S^1$ by identifying $t\sim t+i\beta$. 

The correlator $f_{\mO\mO}(t)$ in \eqref{thermal2pt} clusters in time because each term inside the sum decays to zero exponentially fast.\footnote{Note that $f_{ab}(t)$ is the same as the connected correlator because the one point functions in the BTZ background vanish.} The origin of this clustering correlator is the analyticity of the Fourier transform (see appendix \ref{thermalFourier})
\begin{eqnarray}
    \hat{f}_{\mO\mO}(k,\omega)\sim  \prod_{\pm} e^{-\omega_\pm/2}\Gamma(\Delta-i \omega_+)\Gamma(\Delta+i \omega_+) \, .
\end{eqnarray}
This is a smooth function of $\omega$ supported everywhere on the real $\omega$ axis \cite{leutheusser2021causal}. For fixed $n$, each term in \eqref{thermal2pt} corresponds to a thermal two-point function in the infinite volume where the spectrum of the CFT Hamiltonian is continuous, $\hat{f}_{\mO\mO}(\omega)$ is analytic, which implies that the correlator clusters. This is as opposed to the thermal AdS answer. For fixed $m$, each term in \eqref{thermalads} corresponds to a two-point function at a finite volume where the spectrum of the CFT Hamiltonian is discrete, therefore the correlator is an almost periodic function of time. In fact, one can observe that the thermal AdS answer is invariant under $t\to t+\beta$, and hence a periodic function of time.




In \cite{maldacena2003eternal}, Maldacena argued that the fact that $f_{\mO\mO}(t)$ computed from the gravity point of view clusters is a manifestation of the information loss problem for eternal black holes.
The way out of this dilemma is to realize that in the limit of $N\to \infty$, the spectrum of the Hamiltonian at large enough energies $E\sim O(N^2)$ becomes continuous, allowing for the existence of simple operators like $\mO$ that cluster.

In a large $N$ theory in arbitrary dimensions, single-trace conformal primary operators $\mO$ with dimension $\Delta =O(N^0)$ on the boundary are dual to supergravity fields in the bulk. Above the Hawking-Page phase transition, the gravity path-integral is dominated by large (semiclassically stable) black holes. In such backgrounds, the correlators $f_{\mO\mO}(t)$ cluster in time, and $\mO$ generate an operator system of mixing observables. For a detailed discussion of this correlator in AdS$_5$ black holes see \cite{festuccia2006excursions}. Festuccia and Liu have argued that clustering in time is a signature of an event horizon in dual geometry \cite{festuccia2007arrow}. For an LL-correlator the perturbation $\mO(0)$ falls into the black hole (thermalizes). The probes $\mO(t)$ at late times cannot recover the information content of $\mO(0)$. In the strict $N\to\infty$ limit, the one-point removed probes $\mO-\braket{\mO}_\beta$ generate an algebra of GFF \cite{witten2022gravity}. Below the Hawking-Page transition, there are no mixing operators and the algebra is type I. Above the Hawking-Page transition, every operator in the algebra is mixing, and by Theorem \ref{clustering-thm} that we prove in section \ref{sec-info-loss}, the observable algebra is a type III$_1$ factor.

At large but finite $N$, to recover the initial perturbation $\mO(0)$ we need to wait until the contribution of the other saddles becomes comparable to the dominant one. Since the non-dominant saddles are exponentially suppressed in $N^2$ we need to wait for a long time. This physics is intrinsically non-perturbative and lies beyond the $1/N$ perturbative correction to the correlator. 

The information loss puzzle is sharp in the limit where we send $N\to \infty$ first, and then $t\to \infty$. To clarify the importance of this limit, we briefly review the expected behavior of the thermal correlator in the hierarchy of time scales. In our large $N$ discussion we have the following hierarchy of time scales: Diffusion time $t_{eq}=O(N^0)$ (sometimes referred to as the local thermalization time), the scrambling time $t_{sc}=O(\log N)$, the Thouless time (the onset of the random matrix behavior) $t_{Thou}=O(\log N)$, and the classical Poincar\'e recurrence time $t_{Poin}=e^{O(N^2)}$.

\begin{enumerate}
    \item {\bf Diffusion time:} 
    This time-scale is $t_{eq}=O(N^0)$ and survives the $N\to \infty$ limit. This is the time scale we are primarily concerned with in this work. 
    Gravity suggests that single-trace operators are mixing. They generate an operator system of mixing observables $\mS$.
    At this time scale both the LL and LR two-point functions are expected to cluster in time evolution toward the future and the past:
    \begin{eqnarray}\label{clusterOO}
        \lim_{|t|\to \infty}f_{\mO\mO}(t)=0, \lim_{|t|\to \infty}f_{\mO\mO}(t+i\beta/2)=0\ .
    \end{eqnarray}
    The decay of the correlator above is exponential with exponents that are related to the black hole quasi-normal modes. 
    In the $N\to \infty$ limit, as we will see in section \ref{GFFsec} the multi-trace operators that we find by multiplying the one-point removed operators also cluster. 
    The equation (\ref{clusterOO}) implies that
    \begin{eqnarray}
        {}_\beta\braket{1|[\mO,\mO(t)]|1}_\beta=0\ .
    \end{eqnarray}
In GFF, the commutator above is central. Therefore, as we separate operators by a large time they become asymptotically Abelian with respect to each other. In other words, we have the {\it Norm Asymptotic Abelianness} property:
    \begin{eqnarray}\label{normAbel}
        \lim_{t\to \infty}\|[\mO,\mO(t)]\|_\infty=0\ .
    \end{eqnarray}
    As we see in section \ref{GFFsec}, this implies that all the higher point Out-of-Time-Ordered correlators (OTOCs) decay as well:
    \begin{eqnarray}
        \lim_{t\to \infty}{}_\beta\braket{1||[\mO,\mO(t)]|^{2k}|1}_\beta=0\ .
    \end{eqnarray}
    
    \item {\bf Scrambling time:} In holography, this time scale is $t_{sc}=O(\log N)$. In a large $N$ theory, the dominant piece of the commutator decays exponentially in $t$, however, there are always $1/N$ corrections. At time scale $t\sim \log N$ the commutator has decayed enough that we can no longer ignore the $1/N$ corrections. 
    At this stage, the OTOC starts growing \cite{maldacena2016bound}. 
    
    \item {\bf Thouless time:} The Thouless time is often defined to be the time scale for the perturbation to exponentially mix across the whole system. Alternatively, we can define it to be the onset of ramp (or the random matrix behavior) in the spectral form factor \cite{gharibyan2018onset}, or the time-scale at which the dynamics is well-approximated by a 2-design. In gravity, this time scale is expected to be $O(\log N)$. 
    
    \item {\bf Classical Poincar\'e recurrence time:} This is an exponentially long time $t_{Poin}=e^{O(N^2)}$ when the time evolution starts distinguishing individual black hole microstates. It corresponds to the plateau in the spectral form factor \cite{cotler2017black}.
\end{enumerate}

In lower dimensions, we have a better understanding of how the non-perturbative $e^{O(N^2)}$ effects restore the information at late times. In $0+1$-dimensional boundaries, we can study the thermal correlators of the SYK model in the conformal limit:
\begin{eqnarray}
    \braket{\psi_L(t)\psi_L(0)}_\beta=\frac{e^{-i\pi\Delta}}{\frac{\beta}{\pi} \sinh(\pi(t-i\ep)/\beta)}
\end{eqnarray}
which decays to zero exponentially fast. However, the theory is simple enough to study the corrections to this behavior \cite{PhysRevD.94.106002}. In two-dimensional boundaries, the thermal correlator on a compact space was studied using the {\it light-heavy} Virasoro block in the large $c$ limit, which corresponds to the correlation function of light-light operators in heavy black hole states \cite{fitzpatrick2016information}.

\subsection{Modular correlators in an arbitrary state}

To prepare for the main results in the next section, here, 
we generalize the definition of the thermal two-point function to the modular two-point function in a general state $\psi$ of an arbitrary von Neumann algebra. Consider an arbitrary density matrix $\psi$ in the separable Hilbert space $\mH_L$:
\begin{eqnarray}
    \psi=\sum_m p_m \ket{m}\bra{m}
\end{eqnarray}
where the probabilities $e^{-\beta E_m}/Z$ in the thermal state are replaced by a general probability distribution $p_m$. The analog of the thermofield double is the canonical purification
\begin{eqnarray}
    \ket{1}_\psi=\sum_m \sqrt{p_m}\ket{m}_L\ket{m}_R\ .
\end{eqnarray}
The modular evolution is given by the unitary flow $e^{-it\hat{K}} = \psi^{it}\otimes\psi^{-it} := \Delta_{\psi}^{it}$ and the modular  two-point function, analogous to \eqref{eq-fab-expression} is
\begin{eqnarray}
    &&f^\psi_{ab}={}_\psi\braket{1|a_L^\dagger e^{i t \hat{K}} b_R^T |1}_\psi=\sum_{mn}p_m^{1/2+it}p_n^{1/2-it}a_{mn}^*b_{mn}\ .
\end{eqnarray}

More generally, in an arbitrary von Neumann algebra $\mA$ and a faithful normal state $\psi$, the Hilbert space $\mH_L\otimes \mH_R$ is replaced by the GNS Hilbert space $\mH_\psi$ which is a faithful representation of the von Neumann algebra $a\to a\ket{1}_\psi$.
The Tomita-Takesaki modular theory defines for us a strongly continuous unitary flow called modular flow $a_\psi(t):=\alpha_\psi^t(a)=\Delta_\psi^{it}a \Delta_\psi^{-it}$. 
To simplify our notation we sometimes suppress the $\psi$ in $a_\psi(t)$.
For a quick review of von Neumann algebras, the definition of normal states, the construction of the GNS Hilbert space $\mH_\psi$, modular theory, and the classification of factors see appendix \ref{app:op}. 

We can consider $a_L\in \mA_L$ in a general von Neumann algebra $\mA_L$, and $b_R$ in the commutant $\mA_R=\mA'_L$. The time-evolution is the modular flow of the state $\Delta_\psi^{it}$, and the modular flowed two-point function \begin{eqnarray}
    f_{ab}^\psi(t):=
    {}_\psi\braket{1|a_L^\dagger\Delta_\psi^{1/2+it} b_L|1}_\psi \ .
\end{eqnarray}
Above we have defined $f_{ab}^\psi(t)$ using the analytic continuation of the LL modular correlators to avoid taking the transpose in the basis of the modular Hamiltonian. This correlator is analytic inside the strip, $\Im(t)\in (-1/2,1/2)$ with continuous boundary values satisfying the modular KMS condition:
\begin{eqnarray}
    f^\psi_{ab}(t-i/2)=f^\psi_{b^\dagger a^\dagger}(-t+i/2)\ .
\end{eqnarray}

In the Fourier space, we have an analog of (\ref{fouriercont}) where $dP_\psi(\omega)$ is the spectral operator-valued measure corresponding to the modular Hamiltonian: $\hat{K}_\psi=-\log \Delta_\psi$. Note that $f_{ab}(0)$ is the alternate inner product between vectors in $\mH_\psi$ \cite{furuya2022real}. Therefore, it follows from the Cauchy-Schwarz inequality that 
\begin{eqnarray}\label{cauchyalter}
    |f^\psi_{ab}(t)|\leq f^\psi_{aa}(0)^{1/2}f^\psi_{bb}(0)^{1/2}<\infty\ .
\end{eqnarray}
To simplify our notation, when unambiguous, we suppress the index $\psi$ in $f^\psi_{ab}$.

We can now generalize our discussion of clustering in time and mixing operators to an arbitrary normal state $\psi$ and $\psi$-mixing operators (those that cluster in modular time).
If the spectrum of the modular Hamiltonian (modular energy differences) of $\psi$ is discrete,\footnote{In generic infinite dimensional systems, the spectrum of energy differences are expected to be discrete but dense in $\mathbb{R}$. See appendix \ref{app:typeiii}.} we have type I von Neumann algebra. In this case, there exists no pair of mixing operators in type I algebras. This is a generalization of Maldacena's argument to out-of-equilibrium states. This has potential implications for pure states  dual to one-sided black holes. We will postpone this discussion to upcoming work.

For the modular flow of a general state $\psi$ the centralizer is the subalgebra of observables $z$ such that 
\begin{eqnarray} \label{eq-centralizer-invariant-flow}
    \alpha^t_\psi(z)=z\ .
\end{eqnarray}
In the physics literature, the operators in the centralizer are sometimes called the {\it modular zero modes} \cite{lashkari2021modular}. Of course, the modular zero modes depend on the state $\psi$, but the non-trivial fact proved by Connes is that the center of the centralizer (the maximal Abelian subalgebra of the centralizer) is the same in all states of the Hilbert space (the folium) (see appendix \ref{app:op}). Therefore, it is an algebraic invariant. It is this key result of Connes that allows us to deduce the algebra type from the modular flow of a single state. In a von Neumann algebra where all operators are $\psi$-mixing, the centralizer has to be trivial. Therefore, this algebraic invariant is trivial, which as we will see fixes the algebra type to be type III$_1$.

\section{Mixing operators}\label{sec-info-loss}

From the point of view of a unitary theory of quantum gravity, black hole evaporation is the same as a piece of coal burning. The decay of the thermal two-point function due to mixing raises a similar challenge in explaining thermalization in a closed quantum system. The thermal two-point function cannot decay unless we are in the strict thermodynamic limit of infinite entropy (for instance, an infinite number of particles $N\to\infty$ or in the classical limit of $\hbar\to 0$). The analog of the gravity answer for the burning piece of coal is a subspace of {\it mixing operators} whose correlators cluster in time. They form an operator system, i.e. a $*$-closed linear subspace of operators that includes identity.
In Lemma \ref{simpleop} below, we show that when the spectrum of the modular Hamiltonian is continuous in some range of energies there always exist mixing operators. We construct them explicitly and comment on the connections to the Eigenstate Thermalization Hypothesis (ETH). In holography, mixing operators are {\it simple operators}, in the sense that they correspond to excitations that propagate locally in the bulk \cite{engelhardt2019coarse}. In the thermalization of chaotic systems, mixing operators are {\it simple operators} as well in the sense that, in the strict infinite entropy limit, the operators that satisfy ETH are mixing (see Lemma \ref{simpleop}). In both holography and ETH, the mixing operators are a distinguished subspace of simple operators.

We start with the definition:
\begin{definition}
    We call a finite set of operators $\{a_1,a_2,\cdots \}$ ``future mixing" in state $\psi$ if for any $i,j$ the correlator $f^\psi_{ij}(t):=f^\psi_{a_i a_j}(t)$ clusters as $t\to \infty$, and ``past mixing" if they cluster as $t\to -\infty$. We call them ``mixing" if they cluster in both the future and the past.
    Similarly, if the decay to zero is exponentially fast, we call them ``future exponentially mixing", ``past exponentially mixing" and ``exponentially mixing", respectively.
\end{definition}
We often use the word ``mixing" for the case when the state is thermal (KMS), and use ``$\psi$-mixing" for the case where the state is some general $\psi$. 

\begin{definition}
    Consider a set of self-adjoint operators $\mX=\{\mO_1,\mO_2,\cdots\}$. We call the set of all operators $a=\alpha_0+\sum_i \alpha_i \mO_i$ with complex numbers $\alpha_i$ such that $f^\psi_{aa}(0)<\infty$ the {\it operator system generated by $\mX$} and denote it by $S_\psi^\mX$.\footnote{Note that the condition $f_{aa}^\psi(0)=\psi(a^\dagger a)<\infty$ means that the operator $a$ acting on the GNS Hilbert vacuum $\ket{1}_\psi$ creates a normalizable state. In other words, our operator system $S_\psi^\mX$ corresponds to a subspace of the GNS Hilbert space that is closed in the Hilbert space norm $L_{2,\psi}$ \cite{furuya2023monotonic}. 
    }
\end{definition}
Note that, in our definition, we can always add the identity operator to the generating set without changing $S_\psi^\mX$. Therefore, from now on, we will always assume that $1\in \mX$.
An operator system is a $*$-closed linear subspace of operators that includes the identity operator (see appendix \ref{app:op} for more information). 

To set the stage for the next two subsections, first in Lemma \ref{mixingop}, we show that $\psi$-mixing operators form an operator system. In Lemma \ref{maximalmixingsystem} we enumerate a few key properties of the maximal mixing operator system. Then, in Lemma \ref{RB}, we derive a sufficient condition for a set of operators to be mixing. In Lemma \ref{smooth}, we find the necessary and sufficient condition for operators to be exponentially mixing. 

\begin{lemma}[Mixing Operator Systems]\label{mixingop}
Consider $\mathcal{X}=\{\mO_1,\mO_2,\cdots\}$ a countable set of pairwise orthogonal future (or past) $\psi$-mixing bounded observables, i.e. self-adjoint operators, for a normal state $\psi$. The set $\mathcal{X}$ generates an operator system $\mS^\mX_\psi$ mixing in both the future and the past. 
\end{lemma}
\begin{proof} 
Since $\mO_i$ are self-adjoint we have
\begin{eqnarray}
    \lim_{t\to \infty}f^\psi_{ij}(-t)=\lim_{t\to \infty}f^\psi_{ji}(t)^*=0\ .
\end{eqnarray}
Therefore, the set $\mathcal{X}$ is future mixing if and only if it is past mixing. Next, we would like to prove that all operators in the operator system generated by $\mX$ are mixing.
Consider a pair of bounded operators $a=\sum_{i} \alpha_i \mO_i$ and $b=\sum_i \beta_i \mO_i$ with $\alpha_i$ and $\beta_i$ complex numbers such that the correlators $f_{aa}(0)<\infty$ and $f_{bb}(0)<\infty$.
We define the sequence of operators $\{a_n\}$ and $\{b_n\}$:
\begin{eqnarray}\label{countable}
a_n=\sum_{i=1}^n\alpha_i \mO_i,\qquad
b_n=\sum_{i=1}^n\beta_i \mO_i\ .
\end{eqnarray}
For finite $n$, the limit $t\to\infty$ of $f_{a_nb_n}(t)$ vanishes because we can bring the limit inside the sum, and each  term in summand vanishes. The idea is to prove that we can exchange the order of limits in $n$ and $t$ because the sequence of functions $f_{a_nb_n}(t)$ converges to $f_{ab}(t)$ uniformly.
To establish uniform convergence, we start with the inequalities
 \begin{eqnarray}
    |f_{ab}(t)-f_{a_nb_n}(t)|=&& |f_{(a-a_n)b}(t)+f_{a(b-b_n)}(t)-f_{(a-a_n)(b-b_n)}(t)|\nn\\
     \leq && |f_{(a-a_n)b}(t)|+|f_{a(b-b_n)}(t)|+|f_{(a-a_n)(b-b_n)}(t)|\nn\\
 \leq &&f_{(a-a_n)(a-a_n)}(0)^{1/2}f_{bb}(0)^{1/2}+f_{(b-b_n)(b-b_n)}(0)^{1/2}f_{aa}(0)^{1/2}\nn\\
 &&+f_{(a-a_n)(a-a_n)}(0)^{1/2}f_{(b-b_n)(b-b_n)}(0)^{1/2}\nn\\
 \leq && 2f_{(a-a_n)(a-a_n)}(0)^{1/2}f_{bb}(0)^{1/2}+f_{(b-b_n)(b-b_n)}(0)^{1/2}f_{aa}(0)^{1/2}\,
 \end{eqnarray}
 where we have used the Cauchy-Schwarz inequality in (\ref{cauchyalter}) and the fact that
 \begin{eqnarray}
   0\leq f_{(b-b_n)(b-b_n)}(0)\leq f_{bb}(0)\ .
 \end{eqnarray}
 Since the state is normal we have $\lim_{n\to\infty}f_{(a-a_n)(a-a_n)}(0)=0$. Hence, 
 \begin{eqnarray}
     \lim_{n\to \infty}\sup_t|f_{ab}(t)-f_{a_nb_n}(t)|\leq &&\lim_{n\to \infty}2f_{(a-a_n)(a-a_n)}(0)^{1/2}f_{bb}(0)^{1/2}+f_{(b-b_n)(b-b_n)}(0)^{1/2}f_{aa}(0)^{1/2}\nn\\
     = && 0\,
 \end{eqnarray}
 which implies that the sequence of functions $f_{a_nb_n}(t)$ converges to $f_{ab}(t)$ uniformly. 

 Therefore, we are allowed to exchange the order of limits $\lim_{n\to \infty}$ and $\lim_{t\to\infty}$, and as a result, we find that $f^\psi_{ab}(t)$ clusters in time
 \begin{eqnarray}
     \lim_{t\to \infty}f_{ab}(t)=\lim_{t\to \infty}\lim_{n\to \infty}f_{a_{n}b_{n}}(t)=\lim_{n\to \infty}\lim_{t\to \infty}f_{a_{n}b_{n}}(t)=0\ .
 \end{eqnarray}
It is clear that if $a$ is mixing $a^\dagger$ is also mixing, and the identity operator is trivially mixing. We find that every pair of operators in the operator system $\mS_\psi^\mX$ generated by $\mX$ is mixing.
\end{proof}

\begin{lemma}[Maximal Mixing Operator System]\label{maximalmixingsystem}
The operator system $\mS_\psi$ of all mixing operators in a state $\psi$  (the Maximal Mixing Operator System) has the following properties:
    \begin{enumerate}
        \item $\mS_\psi$ is closed under modular time-evolution: if $a\in \mS_\psi$ then the modular flowed operator $a_\psi(t)\in \mS_\psi$.
    \item The operator $a_\psi(z)$ with analytically continued time evolution to complex times $\Im(z)\in [0,1/2)$ is in $\mS_\psi$.
    \item $\mS_\psi$ cannot contain any operator that is a countable sum over eigenoperators of $\hat{K}$. In particular, in the KMS case, the corresponding maximal mixing operator system $\mS_\beta$ cannot contain any operator that is a countable sum over partial isometries $v_{mn}=\ket{E_m}\bra{E_n}$.  
    \end{enumerate}
\end{lemma}

\begin{proof} 
{\bf (1):} First consider the real-time evolutions $t=\Re(z)$ and $t'=\Re(z')$. We have
\begin{eqnarray}
    f^\psi_{a_\psi(t)b_\psi(t')}(s)={}_\psi \braket{1|a_L e^{i (-t+s+t'-i/2)\hat{K}} b_L|1}_\psi\ .
\end{eqnarray}
Therefore,
\begin{eqnarray}
    \lim_{s\to \infty}f^\psi_{a(t)b(t')}(s)=\lim_{s\to \infty} f^\psi_{ab}(s-t+t')\ .
\end{eqnarray}
{\bf (2):} For complex $z$ with $\Re(z)\geq 0$ we can write 
\begin{eqnarray}
a_\psi(z)\ket{1}_\beta=\int_{-\infty}^\infty \frac{\alpha dt}{\pi(t^2+z^2)}\:a_\psi(t)\ket{1}_\beta\ .  
\end{eqnarray}
Then, repeating the argument above we find
\begin{eqnarray}
    \lim_{s\to \infty}f^\psi_{a_\psi(z)b_\psi(z')}(s)=\lim_{s\to \infty}f_{ab}(s)\ .
\end{eqnarray}
{\bf (3):} Eigenoperators of $\hat{K}$ are defined as $v_\lambda$ that solve
\begin{eqnarray}
    (\hat{K}-\lambda)v_\lambda\ket{1}_\psi= 0\ .
\end{eqnarray}
Any such operator cannot be a mixing operator because its correlators are periodic in modular time. 
Any operator that is a countable sum over $v_\lambda$ cannot be in $\mS_\psi$ because its correlators are almost periodic in modular time and have no limits.
In the case of KMS state, this means that any operator of the form below cannot be mixing:
\begin{eqnarray}
    a=\sum_{mn} \alpha_{mn} \ket{E_m}\bra{E_n}\notin \mS_\beta\ .
\end{eqnarray}
\end{proof}
It is worth emphasizing that there can be many mixing operator systems. For instance, in the GFF example of time-band algebras in section \ref{sec:timeband}, we will see that there are mixing operator systems with support limited to any finite time interval. However, the maximal operator system has the key property that it is closed under time evolution analytically continued to the strip. In particular, if $a$ and $b$ are mixing operators then when 
\begin{eqnarray}
    f^\psi_{ab_\psi(i\beta/2)}(s)=f_{ab}(t+i\beta/2)<\infty
\end{eqnarray}
we can add $b_\psi(i\beta/2)$ to the list of mixing operators.

To obtain a sufficient condition for the mixing property we prove the following Lemma based on the Riemann-Lebesgue Lemma.\footnote{For similar statements see  \cite{festuccia2007arrow,lashkari2021modular}.}

\begin{lemma}[Measurability criterion]\label{RB}
 Consider a countable set of pairwise orthogonal bounded observables $\mX=\{\mO_1,\mO_2,\cdots \}$, $\bO_i=\mO_i-\psi(\mO_i)$ the one-point subtracted operators in state $\psi$, and $f^\psi_{ij}:=f^\psi_{\bO_i\bO_j}$. 
  Let the Fourier transform of the thermal correlators $\hat{f}_{ij}(\omega)$ be (Lebesgue)-measurable functions of frequency $\omega$ for all $i,j$. Then, the operator system $S_\psi^\mX$ generated by $\mX$ is mixing. 
   \end{lemma}

 \begin{proof} We split the proof into two parts. First, we prove $f^\psi_{ii}(t)$ clusters. Then, we prove that $f^\psi_{ij}(t)$ clusters. Finally, Lemma (\ref{mixingop}) finishes the proof. 
 
  {\bf (1):} Since $\hat{f}_{ii}(\omega)>0$ and it is measurable it integrates to 
  \begin{eqnarray}
   \int d\omega \hat{f}_{ii}(\omega)= f_{ii}(t=0)={}_\psi\braket{a_i|e^{-\pi\hat{K}}a_i}_\psi<\infty, 
  \end{eqnarray}
   where in the last step we have used the fact that $a$ is a bounded operator.
   Therefore, $\hat{f}_{ii}(\omega)\in L^1(\omega)$. Then, it follows from the Riemann-Lebesgue Lemma that in the long time limit the two-point function decays\footnote{The Riemann-Lebesgue Lemma says that an $L^1$-function cannot have infinite frequencies.}
  \begin{eqnarray}
      \lim_{t\to \infty} f_{ii}(t)=0
  \end{eqnarray}
  and we learn that $\mO_i$ clusters in time. Note that we need to remove the one-point function; otherwise, $\hat{f}_{\mO\mO}(\omega)$ is not measurable because its one-point function gives a Dirac delta function at zero frequency.


{\bf (2):} From the first part we know that $\hat{f}_{ii}$ and $\hat{f}_{jj}$ are in $L^1\cap L^2$. Since we assumed $\hat{f}_{ij}$ is Lebesgue measurable it follows from the Cauchy-Schwarz inequality (\ref{Cauchyineq}) that 
\begin{eqnarray}
    \int d\omega |\hat{f}_{ij}(\omega)|\leq \int d\omega |\hat{f}_{ii}(\omega)|^{1/2}|\hat{f}_{jj}(\omega)|^{1/2}\leq f_{ii}^{1/2}(t=0)f_{jj}^{1/2}(t=0)<\infty\ .
\end{eqnarray}
Therefore, $\hat{f}_{ij}(\omega)\in L^1$ and as a result $f_{ij}(t)$ clusters.
It follows from Lemma (\ref{mixingop}) that they form an operator system. This finishes the proof.
\end{proof}

Next, we prove that the correlator $f_{ab}(t)$ falls to zero exponentially fast as $t\to \pm \infty$ (i.e. the operators are exponentially mixing) if and only if $\hat{f}_{ab}(\omega)$ is analytic. 

  \begin{lemma}[Analyticity criterion]\label{smooth}
There exists a small $\epsilon>0$ such that for all $z\leq \epsilon$ we have
 \begin{eqnarray}\label{expfalloff}
     \int dt\: e^{2z |t|}|f_{ab}(t)|^2 <\infty
 \end{eqnarray}
if and only if $\hat{f}_{ab}(\omega)$ is analytic in the strip $|\Im(\omega)|\leq \ep$ and inside the strip we have: 1) $\hat{f}(\omega)\in L^2(\Re(\omega))$ and 2) the supremum over these $L^2$-norms over the strip remains bounded. 
  \end{lemma}
  \begin{proof} 
  See theorem IX.12 in Reeds and Simons, Volume 2 \cite{reed1975ii}. 
  \end{proof}
  
We saw that the LR correlator $f^\psi_{ab}(t)$ is an analytic function of the modular time $t$ for $|\Im(t)|< 1/2$. It follows from the theorem above that $|\hat{f}_{ab}(\omega)|$ must fall faster than $e^{- \omega/2}$. In the KMS state, the fall-off is faster than $e^{-\beta \omega/2}$.

The measurability criterion in Lemma \ref{RB} and the analyticity criterion above are examples of a class of {\it uncertainty principles} in functional analysis that relate the decay properties of functions and distributions to the smoothness of their Fourier transforms.
Informally, one can summarize the intuition coming from these theorems as:
\begin{itemize}
\item {\bf Plancherel's theorem:} If $f_{ab}(t)$ is in $L^2$ (i.e. it is square-integrable) its Fourier transform $\hat{f}_{ab}(\omega)$ is also in $L^2$.

\item {\bf Riemann-Lebesgue lemma:} If $\hat{f}_{ab}(\omega)$ is absolutely continuous (in our case Lebesgue measurable) then its Fourier transform $f_{ab}(t)$ is continuous and decays to zero.
    \item If $f_{ab}(t)$ decays faster than any power law in time (i.e. it is a Schwartz function of time), its Fourier transform is also a Schwartz function.
    \item If $f_{ab}(t)$ decays exponentially fast in time its Fourier transform is analytic and decays exponentially fast in real frequency space. 
    \item {\bf Paley-Wiener theorem:} If $f_{ab}(t)$ decays faster than any exponent the Fourier transform is an entire function of imaginary frequencies.
    \item {\bf Paley-Wiener theorem (alternate form):} If $f_{ab}(t)$ has compact support then its Fourier transform is an entire function of exponential type $C$, i.e. there exists a $C$ such that $|\hat{f}(\omega)|\leq e^{C|\omega|}$.
\end{itemize}

\subsection{Mixing operator systems and continuous spectrum}\label{mixingopexample}

So far, we have identified the measurability of $\hat{f}_{ab}(\omega)$ as a sufficient condition for $f_{ab}(t)$ to cluster in time. Here, we construct such operators and observe that the constraints coming from clustering in time on the matrix elements of mixing operators in the energy eigenbasis resemble the condition on simple operators in the Eigenstate Thermalization Hypothesis (ETH).

Let us analyze the measurability criterion of Lemma \ref{RB} more carefully. We saw that a discrete spectrum of $\Omega(E)$ does not allow mixing operators. In the large $N$ or the thermodynamic limit, if the spectrum of $\Omega(E)$ becomes continuous in some range $I=(\bar{E}-\delta,\bar{E}+\delta)$ we explicitly construct mixing operators. Denote the density of state at energy $E$ by $\Omega(E)$. Then,
\begin{eqnarray}
    &&\hat{f}_{ab}(\omega)=\, Z^{-1} \int dEdE'\:e^{-\beta(E+E')/2}\delta\left(\omega - (E-E')\right)\hat{f}_{ab}(E,E')\nn\\
    &&\hat{f}_{ab}(E,E')=\sum_{\alpha=1}^{\Omega(E)}\sum_{\beta=1}^{\Omega(E')} \, a^*(E_\alpha,E_\beta') b(E_\alpha,E_\beta')\,  \, 
\end{eqnarray}
where the sum $\alpha$ and $\beta$ run over microstates of a fixed energy.
A simplifying assumption to make is to consider $a$ and $b$ whose matrix elements become independent of the microstates in the thermodynamic limit and depend only on the energies. For such operators, we can write
\begin{eqnarray}\label{fomega2}
    \hat{f}_{ab}(\omega)&&=\, Z^{-1} \int dEdE'\: \Omega(E)\Omega(E') e^{-\beta(E+E')/2} \, a^*(E,E') b(E,E')\, \delta\left(\omega - (E-E')\right) \, ,\nn\\
        &&=\, Z^{-1} \int dE \:e^{-\beta E}\Omega(E-\omega/2)\Omega(E+\omega/2)\: a^*(E-\omega/2,E+\omega/2)b(E-\omega/2,E+\omega/2)\nn\\
        &&=\, \int d\nu_\omega(E) \: a^*(E-\omega/2,E+\omega/2)b(E-\omega/2,E+\omega/2)
    \end{eqnarray} 
where we have shifted the energy spectrum, and we have defined the discrete measure supported on the countable set of energy differences $\omega=E_m-E_n$:
\begin{eqnarray}\label{numeasure}
    d\nu_\omega(E)=dE \:(Z^{-1} e^{-\beta E})\Omega(E-\omega/2)\Omega(E+\omega/2)\ .
\end{eqnarray}
The measure $d\nu_\omega(E)$ captures the autocorrelation between the density of states at energy $E-\omega/2$ and $E+\omega/2$. It integrates to the Fourier transform of the spectral form factor:
\begin{eqnarray}
    \int d\nu_\omega(E)=Z(\beta)^{-1}\int dt \:e^{i t\omega}|Z(\beta/2+it)|^2\ .
\end{eqnarray}

As long as $d\nu_\omega$ (or equivalently the spectrum of $dP_\omega$) is discrete, the Fourier transform $\hat{f}_{ab}(\omega)$ is a tempered distribution that is a countable sum of Dirac delta functions. However, if $\Omega(E)$ is a measurable function of energy, then $d\nu_\omega(E)$ defined in (\ref{numeasure}) is absolutely continuous with respect to the Lebesgue measure.\footnote{The multiplication of measurable functions is also measurable.}
Then, it follows from (\ref{fomega2}) that further constraints on the matrix elements of $a$ in the energy eigenbasis can ensure that $\hat{f}_{ab}(\omega)$ is a measurable function of $\omega$. 
Based on this intuition, we prove the following lemma.

  \begin{lemma}\label{simpleop}
Consider a modular Hamiltonian with a spectrum that in some range $I=(\bar{E}-\delta, \bar{E}+\delta)$ is measurable, i.e. it is in $L^1(E)$, and a countable set of bounded orthonormal observables $\mathcal{X}=\{\bO_1,\bO_2,\cdots \}$ with vanishing one-point functions where the matrix elements of every $\bO_i$ in the modular energy eigenbasis have the following form
  \begin{eqnarray}\label{simplecond}
\bO_i(E,\alpha;E',\beta) =
  \begin{cases}
   g_i(E-E',E+E')      & \quad \text{if $E\in I$ and $E'\in I$}\\
    0  & \quad \text{Otherwise}
  \end{cases} \, ,
  \end{eqnarray}
where $\alpha$ and $\beta$ run over microstates of energy $E$ and $E'$, respectively, and $g_i(E-E',E+E')$ is a bounded measurable function of both variables, i.e. $g_i\in L^\infty(E)$. Then, the set $\mathcal{X}$ generates a mixing operator system $\mS^\mX$ with the property that for any $a,b \in \mS^\mX$ we have
   \begin{enumerate}
       \item The correlator $f_{ab}(z)$ is an entire function of complex time $z$ that grows at most exponentially fast at large $|z|$, i.e. there exists a constant $C$ such that for all $z\in \mathbb{C}$ we have $|f_{ab}(z)|\leq C e^{\delta |z|}$.

       \item The correlator clusters in complex time $\Im(z)\geq 0$, i.e. $\lim_{|z|\to \infty} f_{ab}(z)=0$.
       
   \end{enumerate}
   
  \end{lemma}

  \begin{proof} Since $\Omega(E)$ is a measurable function of $E$ the measure $d\nu_\omega(E)$ in (\ref{numeasure}) is measurable. Since both $g_i$ and $g_j$ are measurable functions of both variables, their product is also measurable, and hence the integrand in (\ref{fomega2}) is a measurable function of two variables $E$ and $\omega$. Finally, we observe that the integral over $E$ in (\ref{fomega2}) is bounded because $a$ and $b$ are bounded operators
  \begin{eqnarray}
      |\hat{f}_{ab}(\omega)|\leq |a||b|\int d\nu_\omega(E)\leq |a||b|\ .
  \end{eqnarray}
  As a result, $\hat{f}_{ab}(\omega)$ is measurable and Lemma \ref{RB} implies that the correlator $f_{ab}(t)$ clusters. Then, it follows from the lemma (\ref{mixingop}) that the set $\mX$ generates a mixing operator system $S^\mX$. 

  Since $|\hat{f}_{ab}(\omega)|$ is bounded, it is in $L^2(I)$. Then the statement {\bf (1)} is a consequence of the Paley–Wiener theorem for compactly supported functions. It says that the Fourier transform of a function that belongs to $L^2(I)$ for some interval $I$ of width $2\delta$ is an entire function of exponential type $\delta$. 
  
  {\bf (2):} This follows from the fact that $\hat{f}_{ab}(\omega)$ is a measurable function of $\mathbb{R}_+$ by the variation of the Riemann-Lebesgue lemma for a Laplace transform. 
  \end{proof}

The second statement of the above lemma implies that for operators $\bf{a}_i$ with the measurability criterion in (\ref{simplecond}) the correlator $f^\psi_{aa}(t)$ is an entire function of $t$, and in particular, for all real $\tau$ we have $\|a_\psi(i\tau)\ket{1}_\beta\|<\infty$. When $\Omega(E)$ is measurable, such operators belong to a $*$-subalgebra that is called the {\it Tomita subalgebra} of observables. Acting on the ``vacuum" they create vectors that are entire in complex modular time.

Note that a function of two variables that is continuous in each variable separately is a measurable function of both. Formulated in terms of continuity,
the condition (\ref{simplecond}) for the matrix elements of mixing observables $\mO$ in the energy eigenbasis in Lemma \ref{simpleop} resembles the Eigenstate Thermalization Hypothesis (ETH). This resemblance is not accidental. When the spectrum of Hamiltonian is discrete (type I algebra), the obstruction to the clustering of $f_{ab}(t)$ comes from the existence of normalizable energy eigenvectors. The existence of energy eigenstates implies periodicity in time and Poincar\'e recurrences. A piece of coal burning in a closed quantum system (type I algebra) never truly thermalizes because the diagonal matrix elements $\braket{E|a|E}$ in a microcanonical ensemble do not evolve over time. ETH was invented to solve the problem of {\it thermal information loss} by focusing attention on a set of {\it simple operators} that cluster in time at the thermodynamic limit. The resolution proposed by ETH is that in a narrow-width microcanonical ensemble, the diagonal matrix elements of simple operators {\it} are smooth functions of only the overall energy $\bar{E}$ and independent of the microscopic detail of the microstate. We postpone a more detailed discussion of this connection to future work.

\subsection{Mixing algebras are type III$_1$ factors}

In this subsection, we consider the case where the mixing property is preserved under multiplication. Starting with a set of mixing operators $\mX$, we are assuming that, by addition and multiplication, we obtain an algebra of mixing operators in $\psi$ that we refer to as a {\it mixing subalgebra of $\psi$ generated by $\mX$} and denote it by $\mM^\mX_\psi$. 
As we will see, this occurs naturally in any Gaussian (quasi-free) state. This includes free field theories and large $N$ theories in the limit of Generalized Free Fields ($N\to \infty$). The Isserlis Lemma and Lemma \ref{higherpointmixing} below clarify how in Gaussian states, Wick's theorem allows the multiplication of mixing operators to remain mixing. We postpone a discussion of the algebra of mixing operators for a general chaotic system in the thermodynamic limit of infinite volume to upcoming work. 

Consider a set of random variables $\{\varphi(t)\}_{t\in X}$ in a multivariate normal distribution (state) with a one-point and connected two-point function (covariance matrix)
\begin{eqnarray}\label{covariance}
   && {}_\psi\braket{1|\varphi(t)|1}_\psi=\mu(t)\nn\\
   &&\Sigma_\psi(t,t')={}_\psi\braket{1|(\varphi(t)-\mu(t))(\varphi(t')-\mu(t'))|1}_\psi\ .
\end{eqnarray}
The characteristic function of this distribution is the expectation value of the random variable 
$\varphi(f)=\braket{\varphi|f}=\sum_{t\in X}f(t) \varphi(t)$: 
\begin{eqnarray}\label{CharGauss}
    {}_\psi\braket{1|e^{i \varphi(f)}|1}_\psi=e^{-\frac{1}{2}\braket{f|\Sigma_\psi f}+i\braket{\mu|f}}\ .
\end{eqnarray}
We define zero-mean random variables by removing the one-point functions
\begin{eqnarray}
    \varPhi(t)=\varphi(t)-{}_\psi\braket{1|\varphi(t)|1}_\psi\ .
\end{eqnarray}
Then, the Isserlis Lemma below establishes that the higher point functions factor in terms of various contractions of the two-point functions.

\begin{lemma}[Isserlis Lemma]\label{isser}
Consider a set of random variables $\{\varPhi(t)\}_{t\in X}$ in a multi-variate Gaussian state with zero mean. All the odd-point functions of $\varPhi(t)$ vanish, and all the even-point functions are given by the Wick contractions:
\begin{eqnarray}
    \psi(\varPhi(t_1)\cdots \varPhi(t_{2n}))=\sum_{p\in pairings}\prod_{(i,j)\in p}\psi(\varPhi(t_i)\varPhi(t_j))
\end{eqnarray}
where the sum is over all pairings of the $2n$ points.
\end{lemma}
\begin{proof}
See \cite{isserlis1918formula} for proof.    
\end{proof}

\begin{lemma}\label{higherpointmixing}
   Consider a set of random variables $\{\varPhi(t)\}_{t\in \mathbb{R}}$ in a multi-variate Gaussian state with zero mean. Assume that the two-point function such that for all $t_1,t_2\in \mathbb{R}$ we have 
   \begin{eqnarray}\label{2ptvanish}
   &&\psi(\varPhi(t_1)\varPhi(t_2))=\psi(\varPhi(t_1-t_2)\varPhi(0))\nn\\
   &&    \lim_{s\to \infty}\psi(\varPhi(t_1+s)\varPhi(t_2))=0\ .
   \end{eqnarray}
  Then, for any $a(t_1,\cdots t_n)=\varPhi(t_1)\cdots \varPhi(t_n)$ we have
   \begin{eqnarray}\label{higherpt}
       \lim_{s\to \infty}\psi(a(t_1+s, \cdots t_k+s)a(t_{k+1}, \cdots, t_{2n}))=\psi(a(t_1,\cdots t_k))\psi(a(t_{k+1},\cdots t_{2n}))\ .
   \end{eqnarray}
\end{lemma}
\begin{proof}
    It follows from Lemma \ref{isser} that the 4-point function is
\begin{eqnarray}
\lim_{s\to \infty}\psi(\varPhi(t_1+s)\varPhi(t_2+s)\varPhi(t_3)\varPhi(t_4))&&=2\lim_{s\to \infty}\psi(\varPhi(s)\varPhi(0))^2+\psi(\varPhi(t_1)\varPhi(t_2))\psi(\varPhi(t_3)\varPhi(t_4))\nn\\
&&=\psi(\varPhi(t_1)\varPhi(t_2))\psi(\varPhi(t_3)\varPhi(t_4))\ .
\end{eqnarray}
Therefore, the 4-point function clusters. More generally, consider any $2n$-point function split as $k$ points (group 1) and $2n-k$ points (group 2) separated by a large $s$. From Lemma \ref{isser} we know that to compute (\ref{higherpt}) we have to sum over all pairings. There are two types of such pairings. Those that include no cross pairings between groups (1) and (2). They contribute to the right-hand side of (\ref{higherpt}). And those that have at least one pairing between the two. All such terms vanish by (\ref{2ptvanish}). As a result, all $2n$-point functions also cluster.
\end{proof}

By definition, if a set of mixing observables form an algebra it is a $*$-algebra. Since we are interested in C$^*$-algebras, in this work, we only consider bounded mixing operators.\footnote{In the examples we will discuss, the distribution valued operator $\varphi(t)$ can be viewed as an unbounded mixing operator that generates a $*$-algebra of unbounded mixing operators.} First, we show that if a set of bounded observables $\mX=\{\mO_1,\cdots \}$ generates an algebra of mixing operators, we can close this algebra in the norm topology or the weak topology preserving the mixing property to get C$^*$ and a von Neumann algebra of mixing operators that we denote respectively by, $\mB^\mX_\psi$ and $\mM^\mX_\psi$. Then, we present a general operator algebra result, Lemma \ref{simpleIII1}, that establishes if the set of all mixing operators, in some thermodynamic or large $N$ limit, forms a von Neumann algebra (the maximal mixing subalgebra) it is a type III$_1$ factor with a trivial centralizer. In other words, the maximal mixing von Neumann subalgebra of $\psi$ (i.e. $\mM_\psi\subset \mA$) has trivial overlap with the centralizer $\mA^\psi$ of $\psi$: $\mM_\psi\cap \mA^\psi=\lambda 1$.

\begin{lemma}[Mixing subalgebras]\label{mixingvN}
Consider $\psi$ a normal state of a von Neumann algebra $\mA$. If a countable set of observables $\mX=\{\bO_1,\bO_2,\cdots \}$ with one-point functions removed, i.e. $\psi(\bO_i)=0$, generates a $\psi$-mixing subalgebra of $\mA$, it can be completed to a $\psi$-mixing C$^*$-algebra $\mB^\mX_\psi$ and a $\psi$-mixing von Neumann algebra $\mM^\mX_\psi$.
Furthermore, the set of all mixing operators in $\mA$ is a von Neumann subalgebra $\mM_\psi\subset \mA$. 
\end{lemma}
\begin{proof} By construction, the mixing subalgebra is a $*$-algebra. The spectral radius of each operator $a$ provides a norm that satisfies the C$^*$-property $\|a\|_\infty=\|a^\dagger\|_\infty$; see appendix \ref{app:op}. For normal states $\psi$, two sequences of mixing operators $\{a_n\},\{b_m\}$ generated by $\mX$ and their strong operator limits $\lim_{n\to\infty} \|a_n-a\|=0$ and $\lim_{m\to \infty}\|b_m-b\|=0$ we have 
\begin{eqnarray}
     \psi(a b(t))=\psi(\lim_{n,m} a_n b_m(t))= \lim_{n,m}\psi( a_n b_m(t))=\lim_{n,m} \psi(a_n) \psi(b_m(t))=\psi(a)\psi(b)\ .
\end{eqnarray}
We can take the limits outside because the state is normal.
The limit operators $a$ and $b$ are also mixing. Therefore, closing $\mA^\mX$ in norm-topology gives a C$^*$-algebra of $\psi$-mixing operators $\mB^\mX_\psi$ generated by $\mX$. Repeating the argument above for weak-operator limit says that the double commutant of the C$^*$-algebra above gives a $\psi$-mixing von Neumann algebra $\mM^\mX_\psi$.
There is a maximal $\psi$-mixing von Neumann algebra $\mM_\psi\subset \mA$ corresponding to all bounded mixing operators in $\mA$. 
\end{proof}

Consider the maximal $\psi$-mixing von Neumann subalgebra $\mM_\psi\subset \mA$. The restriction of a normal state on $\mA$ to this subalgebra is a normal state on $\mM_\psi$. Therefore, we can define the modular flow of the maximal $\psi$-mixing subalgebra. For all $\psi$-mixing operators $a\in \mM_\psi$ we have the flow:
\begin{eqnarray}
    a_{\psi_\mM}(t):=\Delta_{\mM;\psi}^{it} a\Delta_{\mM;\psi}^{-it}\ .
\end{eqnarray}
A classic result of operator algebra (the Lemma below) implies that, on $\psi$-mixing operators, the modular flow of $\mM_\psi$ is the same as the modular flow of $\mA$. 

\begin{lemma}\label{Takesaki}
    Consider the inclusion of von Neumann algebras $1\in \mM\subset \mA$ and a normal faithful state $\psi$ on $\mA$ whose modular flow preserves $\mM$, i.e. $\Delta_{\mA}^{it}\mM\Delta_{\mA}^{-it}\in \mM$ for all $t\in \mathbb{R}$. Then, for all $a\in\mM$ we have
    $\Delta_\mM^{it}a\Delta_{\mM}^{-it}=\Delta_{\mA}^{it}a\Delta_{\mA}^{-it}$.  
\end{lemma}

\begin{proof} The Lemma follows from the following uniqueness result. The modular flow of $\mM$ is the unique group of $*$-automorphisms of $\mM$ that satisfies the KMS condition (see \cite{stratila2020modular} page 18, or \cite{struatilua2019lectures} 10.17):
\begin{eqnarray}
   \forall a,b\in \mM:\qquad \psi(a b_{\psi_\mM}(t+i))=\psi(b_{\psi_\mM}(t) a)\ .
\end{eqnarray}
\end{proof}

In Lemma \ref{maximalmixingsystem}, we proved that the maximal mixing subalgebra $\mM_\psi\subset \mA$ has the key property that it is preserved under the modular flow of $\mA$. As a result, we have the following Corollary:

\begin{corollary}\label{maximalmixingflow}
    Consider a von Neumann algebra $\mA$, a faithful normal state $\psi$, and the maximal von Neumann algebra of $\psi$-mixing operators $\mM_\psi$. The action of the modular flow of $\mA$ restrict to mixing operators is the same as that of the modular flow of $\mM_\psi$.
\end{corollary}

Consider a von Neumann algebra $\mA$ and a faithful normal state $\psi$.
 The modular flow for some fixed $\tau\in \mathbb{R}$ is an called an {\it inner automorphism} of the algebra if there exists a unitary operator $u_\tau\in \mA$ such that 
\begin{eqnarray}
    \forall a\in \mA: \alpha_\tau(a)=u_\tau a u_\tau^\dagger\ .
\end{eqnarray}
The set of such $u_\tau$ generates a subalgebra of $\mA$ that we denote by $\mA_{inn}$, and as we will see it is inside the centralizer $\mA^\psi$. These definitions depend on the state $\psi$. However, as we review in appendix \ref{app:op}, the set of $\tau\in \mathbb{R}$ for which the modular flow $\Delta_\psi^{i\tau}$ is inner is a subgroup of $\mathbb{R}$, and an algebraic invariant (independent of the state of $\psi$):
\begin{eqnarray}
    \mT=\{\tau\in \mathbb{R}: \alpha^\tau_\psi\: \text{is inner}\}\ .
\end{eqnarray}
The connection to clustering comes from the following Lemma \ref{trivialcentralizer} and Theorem \ref{clustering-thm}. 
 \begin{lemma}[Trivial centralizer]\label{trivialcentralizer}
     Consider a non-trivial von Neumann algebra $\mA$ and a faithful normal state $\psi$.\footnote{We call an algebra trivial if all operators are proportional to identity. In other words, the Abelian algebra of complex numbers is trivial.} If the centralizer of $\psi$ is trivial then the modular flow of all states cannot be an inner flow for all modular times, and $\mA$ must be a type III factor.
 \end{lemma}
 \begin{proof}
By definition, the center of the algebra is inside the centralizer of all states. Since the centralizer of $\psi$ is trivial, we deduce that the center of $\mA$ is trivial, and hence, $\mA$ is a factor. 

The trivial centralizer of $\psi$ also rules out the possibility that the modular flow of $\psi$ is an inner flow for all modular times. For the sake of contradiction, we assume that the modular flow of $\psi$ is an inner flow for all modular times. That is, $\forall t \in \mathbb{R}$ and $\forall a \in \mA$, $\alpha_{\psi}^{\tau}(a) = u_{\tau}au_{\tau}^{\dagger}$ where $u_{\tau} \in \mA$. Then it follows that $u_{\tau}$ is in the centralizer of $\psi$ since for an arbitrary $a\in \mA$ we have 
\begin{eqnarray}
  \psi(au_\tau)=\psi(\alpha^{\tau}_\psi(a u_\tau))=\psi(u_\tau a u_\tau u_\tau^\dagger)=\psi(u_\tau a)\ . 
\end{eqnarray}
Thus, the trivial centralizer of $\psi$ implies $u_{\tau}$ is a phase, and hence, the modular flow is trivial meaning $\alpha_{\psi}^{\tau}(a) = a$ for all $\tau \in \mathbb{R}$ and for all $a \in \mA$. This means that the whole algebra $\mA$ is the centralizer of $\psi$ (see equation \eqref{eq-centralizer-invariant-flow}), which is a contradiction. Thus, the trivial centralizer of $\psi$ implies that the modular flow of $\psi$ cannot be an inner flow for all modular time $\tau$. 

The set of modular times for which the modular flow of $\psi$ is an inner flow is independent of $\psi$ and is a property of the algebra $\mA$; see Appendix \ref{app:op} for a review. Thus, we prove that the modular flow of all states cannot be an inner flow for all modular time. 

Finally, we note that the algebra is type I or type II if the modular flow is an inner flow for all modular times; see Lemma \ref{typenotiii} in Appendix \ref{app:op}. The basic idea behind this lemma is that if and only if the modular flow of $\psi$ is always inner, then we can define a semidefinite trace, which only exists in type I or type II algebras. Thus, algebra $\mA$ cannot be type I or type II and, therefore, is a type III factor. This finishes the proof.
\end{proof}
 
     

\begin{theorem}[Mixing von Neumann algebra is a type III$_{1}$ factor]\label{clustering-thm}
    Consider a non-trivial von Neumann algebra $\mM$ and a normal faithful state $\psi$. The algebra $\mM$ is $\psi$-mixing only if it is a type III$_1$ factor and the state $\psi$ has a trivial centralizer.
\end{theorem}
\begin{proof}
Assume an operator $c\in \mM^\psi$ is in the centralizer of the state $\psi$. It follows from Theorem 3.6 of \cite{pedersen1973radon} that any $c\in \mM^\psi$ is invariant under the modular flow
\begin{eqnarray}
    \forall t\in \mathbb{R}:\qquad \alpha^t_\psi(c)=c\ .
\end{eqnarray}
As a result, the two-point function 
\begin{eqnarray}
    \psi(a \alpha^t_\psi(c))=\psi(a c)
\end{eqnarray}
clusters in modular time only if 
\begin{eqnarray}
    \forall a\in \mM:\qquad \psi(a(c-\psi(c)))=0\ .
\end{eqnarray}
If we choose $a^\dagger={\bf c}=c-\psi(c)$ then we get 
$\psi({\bf c}^\dagger {\bf c})=0$. Since $\psi$ is faithful we find 
\begin{eqnarray}
    c=\psi(c)\ .
\end{eqnarray}
This implies that the centralizer is trivial. It follows from Lemma \ref{trivialcentralizer} that the algebra $\mM$ is a type III factor. 
 Moreover, it follows from a well-known result of Connes (Lemma \ref{factorcentralizer} in appendix \ref{app:op}) that since the centralizer is trivial, the spectrum of $\Delta_\psi$ is the same for all states (and the same as the Arveson spectrum in (\ref{arveson})):
 \begin{eqnarray}
     \mS(\mM)=\text{Spec}(\Delta_\psi)\ .
 \end{eqnarray}
 Then, from the classification of type III algebras in Definition \ref{deftypeiii} it follows that 
 \begin{enumerate}[(a)]
     \item $\mM$ is type III$_0$ if $\text{Spec}(\Delta_\psi)=\{0,1\}$.
     \item $\mM$ is type III$_\lambda$ if $\text{Spec}(\Delta_\psi)=\{0\cup \lambda^n\}$ with $n\in \mathbb{Z}$.
     \item $\mM$ is type III$_1$ if $\text{Spec}(\Delta_\psi)=\mathbb{R}_+$.
 \end{enumerate}
 We rule out cases (a) and (b) by contradiction.
 In cases (a) and (b) the modular Hamiltonian has a discrete spectrum, therefore a modular correlator takes the form
 \begin{eqnarray}
     f^\psi_{aa}(t)=\sum_{m\in \text{Spec}(\Delta_\psi)} m^{1/2+it} {}_\psi\braket{a|dP_m|a}_\psi\ .
 \end{eqnarray}
 Since the sum is discrete, this is an almost-periodic function of time and cannot cluster which is a contradiction. Therefore, a $\psi$-mixing algebra must be a type III$_1$ factor with a trivial centralizer.
\end{proof}
As a corollary, we find that in $\psi$-mixing von Neumann algebra, the spectrum of $\Delta_\psi$ is always $\mathbb{R}_+$.
\begin{corollary}\label{specstate}
    A von Neumann algebra $\mM$ is $\psi$-mixing only if the spectrum of the modular flow of $\psi$ and any other state in the folium is $\mathbb{R}_+$:
    \begin{eqnarray}
        \forall \ket{\chi}\in\mH_\psi  \qquad \text{Spec}(\Delta_\chi)=\mathbb{R}_+\ .
    \end{eqnarray}
\end{corollary}

Another interesting consequence is that if the algebra is $\psi$-mixing it is $\chi$-mixing for almost all other $\ket{\chi}\in \mH_\psi$:
\begin{lemma}
    If a von Neumann algebra $\mM$ is $\psi$-mixing then it is $\chi$-mixing for a dense set of vectors $\ket{\chi}\in \mH_\psi$. 
\end{lemma}
\begin{proof}
    The proof follows from the result of \cite{connes1978homogeneity} that proved for a type III$_1$ factor that the set of vectors 
    \begin{eqnarray}
        \ket{\chi}=UU'\ket{1}_\psi,
    \end{eqnarray}
    with $U\in \mM$ and $U'\in \mM'$ is dense in $\mH_\psi$. Then, it follows that for all $a\in \mM$
    \begin{eqnarray}
    \alpha_\chi^t(a)=U\alpha_\psi^t(U^{\dagger}aU) U^\dagger\ ,
    \end{eqnarray}
    and for $a, b \in \mM$,
    \begin{align}
        \bra{\chi} \alpha_{\chi}^{t}(a) b \ket{\chi} 
        = \, & {}_{\psi}\bra{1} \alpha_{\psi}^{t}(U^{\dagger} a U) \, U^{\dagger} b U\ket{1}_{\psi} \ . 
    \end{align}
    Since $U^\dagger aU\in \mM$ and $U^\dagger b U\in \mM$ and $\mM$ is $\psi$-mixing, we find
    \begin{align}
        \lim_{t\to\infty} \bra{\chi} \alpha_{\chi}^{t}(a) b \ket{\chi} =& \,  \lim_{t\to\infty}  {}_{\psi}\bra{1} \alpha_{\psi}^{t}(U^{\dagger} a U) \, U^{\dagger} b U\ket{1}_{\psi} \, ,\nonumber\\
        =& \, {}_{\psi}\bra{1} U^{\dagger} a U\ket{1}_{\psi} \, {}_{\psi}\bra{1} U^{\dagger} b U\ket{1}_{\psi} \, ,\nonumber\\
        =&  \bra{\chi} a \ket{\chi} \, \bra{\chi} b \ket{\chi} \, .
    \end{align}
\end{proof}
Note that the above result implies that we have $\chi$-mixing for all almost states $\chi$. However, there are many states in $\mH_\psi$ with non-trivial centralizers. The existence of a centralizer is an obstruction to mixing. 

Theorem \ref{clustering-thm} combined with Lemma \ref{Takesaki} implies the following lemma: 
  \begin{lemma}\label{simpleIII1}
  The von Neumann subalgebra of all $\psi$-mixing operators $\mM_\psi\subset \mA$, if nontrivial, is a type III$_{1}$ factor and the restriction of $\psi$ to $\mM_\psi$ has a trivial centralizer in $\mM_\psi$.
  \end{lemma}

  \begin{proof} It follows from Corollary \ref{maximalmixingflow} that
  \begin{eqnarray}
      \forall a\in \mM_\psi:\qquad a(t):=\Delta_{\mM_\psi}^{it}a\Delta_{\mM_\psi}^{-it}=\Delta_\mA^{it}a\Delta_\mA^{-it}\ .
  \end{eqnarray}
  If all operators in $\mM_\psi$ are $\psi$-mixing with respect to the modular flow of $\mA$ then they are $\psi$-mixing with respect to the modular flow of $\mM_\psi$ as well, and by Theorem \ref{clustering-thm} we find that $\mM_\psi$ is a type III$_1$ von  Neumann algebra with a trivial centralizer. 
  \end{proof}

In a KMS state, the modular flow is the time evolution. Therefore, we have the following corollary: 
\begin{corollary}\label{KSMtypeiii}
If the set of all mixing operators in a KMS state forms an algebra, it is a type III$_1$ von Neumann factor $\mM\subset \mA$ and none of the conserved charges of $\mA$ can be affiliated with $\mM$.\footnote{An operator $a\in B(H)$ is affiliated with $\mA$ if it commutes with the commutant $\mA'$. When $\mA$ is a von Neumann factor, if $g$ is a bounded function and $a$ is affiliated with $\mA$ then $g(a)\in\mA$.}
\end{corollary}

 \section{Mixing algebras of GFF and holography}\label{GFFsec}

We saw that if the maximal mixing subalgebra is non-trivial, it must be a type III$_1$ von Neumann factor. An important example for us is the $N\to \infty$ limit of 
$SU(N)$ $\mathcal{N}=4$ SYM above the Hawking-Page phase transition.\footnote{For completeness, in appendix \ref{app:SYMGFF}, we review the boundary arguments by Festuccia and Liu that suggest that the two-point function of single trace primaries in interacting $\mathcal{N}=4$ SYM above the Hawking-Page phase transitions should cluster.} In this case, every single-trace conformal primary $\mO$ with the one-point function removed generates an algebra of generalized free fields (GFF) that we denote by $\mY_\mO$, following the notation of \cite{leutheusser2021emergent}. Since the correlators in a theory of GFF are Gaussian (the state is quasi-free), we only need to provide the two-point function $f_{\mO\mO}$  to define the von Neumann algebra. 
We show that if the generator of the GFF algebra, $\mO$, is mixing, they form an algebra of observables that is mixing, and by Theorem \ref{clustering-thm}, the algebra $\mY_\mO$ is a type III$_1$ factor. Our result settles and generalizes a conjecture of Leutheusser and Liu in \cite{leutheusser2021emergent}. Note that, as was discussed in \cite{witten2022gravity}, the single-trace operators $H/N$ commutes with all the other single-trace operators. Therefore, in what follows, we choose to not include it in the algebra of the single-trace observables. In the language of \cite{chandrasekaran2022large} we restrict the generating set of our GFF to the single-trace operators that are non-central in the large $N$ limit.

We construct the algebra $\mY_\mO$ explicitly and show that there are C$^*$-subalgebras associated with time-intervals $I_t=(t,\infty)$ at any temperature (time-band algebras), as was argued by \cite{leutheusser2021emergent}. We observe that there is an ambiguity in defining the time-band algebra. To resolve this ambiguity, we need to make choices. As we show, for some choices the time-band algebras are not von Neumann algebras. However, the choice motivated by holography is special in that it associates GFF von Neumann algebras with every causally complete region of the bulk.

Generalized free fields in $d$-spacetime dimensions can be described as \cite{dutsch2003generalized}
\begin{eqnarray}
    \varphi_\rho(x^\mu)=\int d^{d}k\: \Theta(k^0)\rho(k_\mu k^\mu) (a^\dagger(k^\mu)e^{i k_\mu x^\mu}+a(k^\mu)e^{-i k_\mu x^\mu})\ .
\end{eqnarray}
Consider GFF on a compact manifold, e.g. a sphere. Expanding the field in spherical harmonics, each mode, in effect, is a model of GFF in $0+1$-dimensions \cite{festuccia2007arrow} (see appendix \ref{app:SYMGFF}). To simplify our discussion of GFF without sacrificing generality we focus on the $0+1$-dimensional GFF, which as we will describe in appendix \ref{app:GHO} can be interpreted as a collection of simple harmonic oscillators with different frequencies $\omega$ (masses). The generalization to higher dimensions is straightforward. 

\subsection{Algebras of GFF}\label{sec:GFFalgebra}


Generalized free fields in $0+1$-dimensions (the zero mode of GFF on a compact spatial manifold) correspond to an effective field
\begin{eqnarray}\label{GFF01}
    \varphi_\rho(t)&&=\int d\omega \Theta(\omega)\sqrt{\rho(\omega^2)}(e^{i \omega t} a_\omega^\dagger+ e^{-i\omega t}a_\omega)
\end{eqnarray}
where the {\it spectral density} function $\rho(\omega^2)$ is a general (discrete or continuous) positive function of frequencies $\omega$, and $\Theta(\omega)$ is the Heaviside step function that ensures the positivity of mass.\footnote{See appendix \ref{app:GHO} for a discussion of the origin of the expansion above, and more detail on the construction of the algebra in $0+1$-dimensions.}
For future convenience, we also define the choice of GFF where the measure is $d\mu(\omega)=d\omega \Theta(\omega)$:    
\begin{eqnarray}\label{Lebesqchoice}
        \varphi_1(t)=\int d\omega \Theta(\omega)(e^{i\omega t}a_\omega^\dagger+e^{-i\omega t}a_\omega)\ .
\end{eqnarray}
This theory includes all massive fields equally. To quantize a theory of GFF and construct the observable algebra we first need to build the one-particle Hilbert space. The reader who is not interested in the detailed mathematical construction of the operator algebra of GFF can skip to Theorem \ref{GHOclsuter}.

{\bf One-particle Hilbert space:} Consider the Abelian $*$-algebra of complex Schwartz functions of time (smooth bounded functions $f$ that decay away faster than any power law at $t\to \pm\infty$). There are two advantages to choosing Schwartz functions in quantum field theory. First, the smoothness allows us to do calculus on the ``phase space".\footnote{The $*$-algebra of Schwartz functions is in the intersection of all $L^p$ spaces for $p\in (1,\infty]$. They can be viewed as the $*$-subalgebra of bounded measurable functions  where the action of time-evolution on them can be extended to the entire complex $t$-plane. In fact, the algebra of Schwartz functions is the Tomita $*$-subalgebra of $L^\infty(\mathbb{R}^{d-1,1})$. We thank Yidong Chen for explaining the second point to us.} However, as opposed to analytic functions, Schwartz functions can have finite support which allows us to make sense of local algebras of observables. In the case of $0+1$-dimensions, these local algebras become the algebra of time bands. Second, the Fourier transform sends Schwartz functions to Schwartz functions, therefore they enjoy the same smoothness properties in frequency space as well.\footnote{There is nothing sacred about this choice. We could have started with any $*$-subalgebra of Schwartz functions as well. We will see that the freedom in making such choices is tied to an ambiguity in the algebra of the GFF. We come back to this in our discussion of the GFF algebras of holography. Moreover, the process of building the one-particle Hilbert space involves choosing a measure $d\mu$ on time, discarding all functions that are in the kernel of $d\mu$ as null states and enlarging the remaining set of Schwartz functions to complete it to  $L^2(\mathbb{R}_t,d\mu)$.}

The set of these functions forms a complex vector space $\ket{f}\in\mK$.
 A state of this $*$-algebra is a positive distribution measure (discrete or continuous) over frequencies \footnote{More generally, we can consider unnormalizable states (weights) such as the Lebesgue measure $d\omega$ on the algebra as was the case in (\ref{Lebesqchoice}).}
\begin{eqnarray}
    \mu(f)=\int d\mu(\omega) \: \hat{f}(\omega)\ .
\end{eqnarray}
The positivity of mass implies that the measure $d\mu(\omega)=d\omega \Theta(\omega)\rho(\omega^2)$ is supported only on positive frequencies. 
A state $d\mu$ provides us with a Hermitian form on the vector space $\mK$,\footnote{It is Hermitian because $\braket{f|g}_\mu=\braket{g|f}_\mu^*$.}
\begin{eqnarray}\label{innerprod}
    \braket{f|g}_\mu=\int d\mu(\omega)\hat{f}^*(\omega)\hat{g}(\omega)\ .
\end{eqnarray}
This Hermitian form satisfies all properties of an inner product $\mK$ except that it is degenerate. For example, any function $\hat{f}_0(\omega)$ with support contained only in negative frequencies $\mathbb{R}_\omega^-$ is within the kernel of this Hermitian form: $\braket{f_0|f_0}_\mu=0$. Such null vectors form a linear subspace. We quotient out by the set of null vectors $I$ by identifying $\ket{f}\sim \ket{f+I}$. The state $\mu$ reduced to the quotient space is faithful. Therefore, it gives an inner product. Completing the quotient space in this inner product gives the one-particle Hilbert space $\mH_\mu$.

It is instructive to make a comparison with conventional quantum field theory. For massive free bosons in $d$ space-time dimensions, the state we choose in the one-particle Hilbert space is $d\mu(k^\mu)=d\omega \Theta(\omega)\delta(k_\mu k^\mu+m^2)$. The propagator solves 
\begin{eqnarray}
    (\Box_{(x)}+m^2)G(x,y)=\delta(x-y)
\end{eqnarray}
and the solutions to the equations of motion are of the form $\int dy f(y)G(x,y)$ for each Schwartz function $f(y)$. The null vectors are ``off-shell" and correspond to the spacetime functions that are in the image of the Klein-Gordon equation. Quotienting the set of Schwartz functions with those in the image of the Klein-Gordon equation (the equations of motion) gives the phase space, i.e. the space of solutions to the equations of motion. The restriction of $d\mu$ to the quotient set $\mK/I$ is an inner product of the phase space. The one-particle Hilbert space is the Cauchy completion of $\mK/I$ in the norm induced by this inner product. 

In the absence of equations of motion, in GFF, taking the quotient of $\mK$ by the null vectors of the bilinear in (\ref{innerprod}) (the choice of measure $\rho(\omega^2)$) defines the analog of the phase space. The null vectors of the state $d\mu$ are ``off-shell". Since $\rho(\omega^2)$ sets the causal propagator, we can think of $1/\rho(\omega^2)$ as nonlocal equations of motion (in momentum space). Given a Schwartz function of spacetime $f$ we can think of the GFF field as 
\begin{eqnarray}
    \varphi_\rho(f)=\varphi_1(\sqrt{\rho(-\Box)}f)
\end{eqnarray}
where $\varphi_1$ is the choice in (\ref{Lebesqchoice}).

In a phase space, we have a closed non-degenerate two-form (the symplectic two-form). Similarly, as we will see below, in the GFF, the role of the symplectic form is played by the so-called symplectic bilinear (an anti-symmetric, non-degenerate bilinear that satisfies the Jacobi identity). In quantization, the symplectic bilinear decides the commutators.

The vectors $\ket{f}\in \mK$ carry a  positive-energy unitary representation of time-translations 
 \begin{eqnarray}
    &&e^{-i Hs}f=f(t+s):=f_s(t)\nn\\
    &&e^{-i Hs}\hat{f}(\omega)=e^{-i \omega s}\hat{f}(\omega)\ .
 \end{eqnarray}
 The time-evolved two-point function and its Fourier transform are
 \begin{eqnarray}\label{2pttime}
     &&G_{gf}(t)=\braket{g|e^{-i Ht }f}_\mu=\int d\mu(\omega)\hat{g}^*(\omega)\hat{f}(\omega) e^{-i \omega t} \ .
 \end{eqnarray}
 The imaginary part of the inner product is anti-symmetric and non-degenerate.\footnote{There exists no $\ket{f}$ such that $\Im\braket{f|g}=0$ for all $g$ because $i\braket{f|f}\neq 0$.} It corresponds to a choice of commutator (symplectic bilinear) \footnote{Whereas the norm the complex inner product above induces $\braket{f|f}_\mu$ fixing a choice of a quasi-free state (see appendix \ref{app:GHO}).}:
 \begin{eqnarray}
     [\varphi_\rho(f),\varphi_\rho(g)]&=&\int d\omega \Theta(\omega)\rho(\omega^2)(\hat{f}^*(\omega)\hat{g}(\omega)-\hat{f}(\omega)\hat{g}^*(\omega))\label{kernelcomm}
 \end{eqnarray}
 where the kernel is the Fourier transform of the causal propagator (the retarded propagator minus the advanced propagator):
 \begin{eqnarray}
    G_C(t):=2i\Im\braket{\delta(t)|\delta(0)}=\int d\omega \rho(\omega^2)\Theta(\omega)(e^{-i \omega t}-e^{i\omega t})\ .
\end{eqnarray}

If we choose $\rho(\omega^2)$ to be measurable, since $\hat{f}(\omega)$ and $\hat{g}(\omega)$ are bounded\footnote{The Schwartz functions are inside all $L^p(\omega)$.} the integrand of (\ref{kernelcomm}) is in $L^1(\mathbb{R}_\omega)$. Evolving one of the vectors in time, we have
\begin{eqnarray}
    [\varphi_\rho(f_t),\varphi_\rho(g)]=\int d\omega \:e^{i\omega t} \Theta(\omega)\rho(\omega^2)(\hat{f}^*(\omega)\hat{g}(\omega)-\hat{f}(-\omega)\hat{g}^*(-\omega))\ .
\end{eqnarray}
It follows from the Riemann-Lebesgue Lemma that this commutator decays at large $t$:
\begin{eqnarray}
    \lim_{t\to \infty}\|[\varphi_\rho(f_t),\varphi_\rho(g)]\|_\infty= 0\ .
\end{eqnarray}
This implies that for any bounded function of $\varphi(f_t)$ and $\varphi(g)$ the commutator vanishes at large time separation which is the so-called asymptotic Abelianness property. As we will see in Theorem \ref{GFFtypeiii1} below, asymptotic Abelianness fixes the von Neumann algebra to a type III$_1$ factor.


To make the discussion less abstract, we work out two simple examples. 
\begin{enumerate}
    \item {\bf All positive energy modes:} This case is the GFF field $\varphi_1$ in (\ref{Lebesqchoice}).
    There are no analogs of equations of motion, and after taking the quotient by negative energy modes, we get a norm. The closure with respect to this norm gives $L^2(\mathbb{R}_t,d\omega \Theta(\omega))$ which is the space of Hardy functions $H^2$: the space of square-integrable functions of time that have only positive frequencies.\footnote{They can also be characterized as the space of square-integrable holomorphic functions on the upper half-plane.}
    
\item {\bf Conformal GFF:} As the second example, consider conformal GFF (0+1)-dimensions. The two-point function decays as a power law with power $2\Delta$. Its K\"all\'en-Lehmann representation corresponds to the spectral density
\cite{dutsch2003generalized}:
\begin{eqnarray}
    \rho(\omega^2 )\sim\omega^{2\nu}, 
\end{eqnarray}
where $\nu=\Delta-1/2$.\footnote{More generally, in GFF in $d$-spacetime dimensions the parameter $\nu=\Delta-d/2$.} In this case, we denote the GFF field with $\varphi_\Delta(t)$:
\begin{eqnarray}\label{GFFfield}
    \varphi_\Delta(t)=\int d\omega\Theta(\omega)\omega^\nu(e^{i\omega t}a_\omega^\dagger+e^{-i\omega t}a_\omega)\ .
\end{eqnarray}
Taking the closure with respect to the inner product in (\ref{innerprod}) enlarges the set of Schwartz functions to those that satisfy
\begin{eqnarray}
    \int d\omega \omega^{2\nu} |f(\omega)|^2<\infty\ .
\end{eqnarray}
The space of these functions is called the Sobolev space $\mathbb{H}^{(-\nu),2}(\mathbb{R}_t,\mathbb{C})$. The one-particle Hilbert space is $\mH_\mu=L^2(\mathbb{R}_t,d\omega \Theta(\omega)\:\omega^{2\nu})$. 


\end{enumerate}

{\bf Fock space:} The vector space $\mK$ with a choice of symplectic bilinear is a {\it symplectic vector space}. Exponentiating this vector space gives us the group of Weyl unitaries.
For bosonic degrees of freedom, we define the unitary coherent operators $W(f)=e^{i\varphi(f)}$.\footnote{The generalization to fermionic fields is straightforward.} 
They satisfy the Weyl algebra
\begin{eqnarray}
    &&W(f)W(g)=e^{-\frac{1}{2}[\varphi(f),\varphi(g)]} W(f+g)=e^{-i\Im \braket{f|g}}W(f+g)\ .
\end{eqnarray}

We define the Fock space $\mF=\oplus_q \mH_\mu^{\otimes q,sym}$, where $\ket{\Omega}\in \mH_0\in \mF$ is the unique vector (up to an overall phase) that is invariant under time translations. The Fock space inherits a unitary representation of time translations from the one-particle Hilbert space. The choice of complex inner product in (\ref{innerprod}) fixes a quasi-free state 
\begin{eqnarray}
&&\braket{\Omega|W(f)|\Omega}=e^{-\frac{1}{2}\braket{f,f}_\mu}\ .
\end{eqnarray}
This is to be compared with (\ref{CharGauss}) as the characteristic equation for multi-variate Gaussian random variables. The analog of the covariance matrix in (\ref{covariance}) is $\rho(\omega^2)\Theta(\omega)$. The choice of state only concerns the real part of the inner product $\braket{f|f}_\mu$. The KMS state for the GFF is \cite{lindner2013perturbative}
\begin{eqnarray}\label{KMSGFF}
    \braket{f|f}_\mu=\int \frac{d\omega\Theta(\omega)\rho(\omega^2)}{2\cosh(\beta \omega/2)}|\hat{f}(\omega)|^2\ .
\end{eqnarray}
This choice satisfies the KMS condition in frequency space we discussed in (\ref{KMSfreq}). 

Under time evolution, Weyl operators evolve according to
\begin{eqnarray}
&&e^{-i Ht}W(f)e^{i H t}=W(f_t)\ .
\end{eqnarray}
Taking the double commutant of the Weyl group in the Fock space $\mF$ gives the von Neumann algebra of the GFF.
All polynomials of $\varphi(f)$ are affiliated with this algebra. The higher point functions of $\varphi(x)$ can be found using the Wick contractions of Lemma \ref{isser}. The odd correlation functions vanish for any state and the even point function can be written as a sum over two-point functions. As we argued in Lemma \ref{higherpointmixing}, it follows from Wick contractions that the clustering of the two-point function implies the clustering of all higher-point functions.

Now, we are ready to prove the following theorem:
\begin{theorem}\label{GHOclsuter}
  In a theory of GFF with spectral density $\rho(\omega^2)$ and an arbitrary quasi-free state $\psi$ represented in its GNS Hilbert space as the vector $\ket{1}_\psi$, the following statements are equivalent:
   \begin{enumerate}
   \item For all $f,g$ the commutator clusters 
   \begin{eqnarray}
   \lim_{|t|\to \infty}|{}_\psi\braket{1|[\varphi(f_t),\varphi(g)]|1}_\psi|=0\ .    
   \end{eqnarray}
       \item 2k-point Strong Asymptotic Abelianness, i.e., all the  Out-of-Time Ordered Correlators (OTOCs) cluster:
       \begin{eqnarray}\label{commutclus}
         \forall\: 1\leq k\in\mathbb{N}:\qquad  \lim_{|t|\to \infty}{}_\psi\braket{1||[\varphi(g),\varphi(f_t)]|^{2k}|1}_\psi=0\ .
       \end{eqnarray}
       \item Norm Asymptotic Abelianness:
       \begin{eqnarray}\label{normasymptotAbelianness}
       \lim_{|t|\to \infty}\|[\varphi(f_t),\varphi(g)]\|_\infty=0   \ .
       \end{eqnarray}
   \end{enumerate}
\end{theorem}
\begin{proof}
{\bf (1$\Leftrightarrow $3)} In GFF the commutator of the fundamental field is proportional to identity:
\begin{eqnarray}
    [\varphi(f_t),\varphi(g)]={}_\psi\braket{1|[\varphi(f_t),g(t)]|1}_\psi\ .
\end{eqnarray}
Therefore, statements 1 and 3 are equivalent. 

{\bf (3$\to$ 2):} This follows from the inequality
\begin{eqnarray}
    |{}_\psi\braket{1|[\varphi(f_t),\varphi(g)]^{2k}|1}_\psi|\leq \|[\varphi(f_t),\varphi(g)]^{2k}\|\leq \|[\varphi(f_t),\varphi(g)]\|^{2k}
\end{eqnarray}
where $k\geq 1$.

{\bf (2$\to$ 1):} Take $k=1$. The statement follows from the inequality $|{}_\psi\braket{1|X|1}_\psi|^2\leq {}_\psi\braket{1|X^\dagger X|1}_\psi$ for $X=[\varphi(f_t),\varphi(g)]$. 
\end{proof}
Finally, we notice that, by the Riemann-Lebesgue Lemma, if $\rho(\omega^2)$ is a measurable function of $\omega$ then the assumption of the theorem above holds.

\begin{theorem}\label{GFFtypeiii1}
    In a theory of GFF with spectral density $\rho(\omega^2)$ and an arbitrary quasi-free state $\psi$ represented in its GNS Hilbert space as the vector $\ket{1}_\psi$,
    if for all $f,g$ the two-point function of the fundamental field clusters in modular time (strong mixing)
       \begin{eqnarray}
        \lim_{t\to \infty}{}_\psi\braket{1|\varphi(f_t)\varphi(g)|1}_{\psi}=0\ ,
    \end{eqnarray}
    then
    \begin{enumerate}
        \item We have Norm Asymptotic Abelianness, i.e., 
        (\ref{normasymptotAbelianness}) holds.
       \item For all $f$ we have $\lim_{t\to \infty}W(f_t)={}_\psi\braket{1|W(f)|1}_\psi$.
      \item The von Neumann algebra of observables of this GFF is mixing, and hence a type III$_1$ factor. 
      \end{enumerate}
\end{theorem}

\begin{proof} 
{\bf (1)} It follows from the inequality below
\begin{eqnarray}
    \lim_{t\to \infty}|{}_\psi\braket{1|[\varphi(f_t),\varphi(g)]|1}_\psi|\leq 2\lim_{t\to \infty}|{}_\psi\braket{1|\varphi(f_t)\varphi(g)|1}_\psi|=0\ .
\end{eqnarray}
Then, it follows from Theorem \ref{GHOclsuter} that the GFF algebra satisfies asymptotic Abelianness. 

{\bf (2)} Consider the matrix elements of the time-evolved $W(g_t)$:
\begin{eqnarray}
    {}_\psi\braket{W(f_1)|W(g_t)|W(f_2)}_\psi=e^{-i\Im\braket{g_t|f_1+f_2}_\mu}e^{-\Re\braket{g_t|f_2-f_1}_\mu}{}_\psi\braket{1|W(g)|1}_\psi{}_\psi\braket{W(f_1)|W(f_2)}_\psi\nn\ .
\end{eqnarray}
Since the two-point function and the commutator cluster we have
\begin{eqnarray}
    &&\lim_{t\to \infty}\braket{g_t|f_2\pm f_1}_\mu=0\nn\\
    &&\lim_{t\to \infty}{}_\psi\braket{W(f_1)|(W(g_t)-\braket{W(g)})|W(f_2)}_\psi=0\ .
\end{eqnarray}
Since in the von Neumann algebra, we define the limit using the weak closures, we have the operator statement
\begin{eqnarray}
    \lim_{t\to \infty}(W(g_t)-{}_\psi\braket{1|W(g)|1}_\psi)=0\ .
\end{eqnarray} 
{\bf (3)} Clearly, any operator that is a finite sum $a=\sum_i \alpha_i W(g_i)$ clusters. If we have a sequence of such operators $a_n$, since the algebra is weakly closed the limiting operators also cluster, therefore $\mA$ is a mixing von Neumann algebra. It follows from Theorem \ref{clustering-thm} that it is a type III$_1$ factor.
\end{proof}

Finally, we find the following Lemma:
 \begin{lemma}\label{measurableresult}
 In a theory of GFF in a KMS state if the spectral density $\rho(\omega^2)$ is a measurable function of real frequencies then
  \begin{enumerate}
         \item The observable algebra $\mA$ has Norm Asymptotic Abelianness.
        \item For all $f$ we have $\lim_{t\to \infty}(W(f_t)-\braket{W(f)}_\beta)=0$.
       \item The observable algebra is mixing, and hence a type III$_1$ factor with a trivial centralizer. 
       \end{enumerate}
 \end{lemma}

 \begin{proof} If we choose $\rho(\omega^2)$ to be measurable since $f(\omega)$ and $g(\omega)$ are bounded, the kernel of the KMS state in (\ref{KMSGFF}) is in $L^1(\omega)$. 
 It follows from the Riemann-Lebesgue lemma that the two-point function clusters, and the proof follows from Theorem \ref{GFFtypeiii1}. 
 \end{proof}

The result above settles the conjecture of Leutheusser and Liu in \cite{leutheusser2021emergent}. They conjectured that the GFF algebra in a KMS state is type III$_1$ if and only if the spectral density is a smooth function that is supported everywhere on real frequencies $\omega$. 
The assumption of measurable $\rho(\omega^2)$ is a lot weaker than smoothness, and  the assumption of entire support in $\mathbb{R}$ is not needed. As we see below, there are type III$_1$  GFF algebras with spectral density $\rho(\omega^2)$ supported on a finite interval in frequency space. However, it is plausible that a smooth and non-vanishing spectral density is required for the type III$_1$ algebra to be hyperfinite. 

In a KMS state, the modular Hamiltonian in the Fock space is the same as the time evolution. For positive $\omega$ we have the eigenoperators 
\begin{eqnarray}
    &&e^{iHt}a_\omega e^{-i Ht}=e^{-i \omega t} a_{\omega} \, ,\nn\\
    &&e^{iHt}a^\dagger_\omega e^{-i Ht}=e^{i \omega t} a_{\omega}^{\dagger} \, ,
\end{eqnarray}
in the one-particle Hilbert space. There are more eigenoperators in the multi-particle sectors. More generally, defining $a_{-\omega}:=a^\dagger_\omega$ we find that any operator of the form
\begin{eqnarray}\label{mutliparspec}
    e^{i Ht}(a_{\omega_1}^{m_1}a_{\omega_2}^{m_2}\cdots a_{\omega_n}^{m_n})e^{-i Ht}=\lb  \sum_{i=1}^n m_i \omega_i\rb (a_{\omega_1}^{m_1}a_{\omega_2}^{m_2}\cdots a_{\omega_n}^{m_n})
\end{eqnarray}
is an eigenoperator of time evolution.
As we include all multi-partite states, even with a spectral density supported at two points (two modes $\omega_1,\omega_2$ that are linearly independent over the set of rationals) the spectrum of the modular operator is already, generically, dense in $\mathbb{R}$. If $\rho(\omega^2)$ is the characteristic function of an interval $I$ in the frequency space ($\rho(\omega^2)=1$ inside interval $I$ and zero outside), once we include the eigenoperators in the multiparticle sectors, it follows from (\ref{mutliparspec}) that the spectrum of the modular operator is the entire $\mathbb{R}_+$, and since the spectral density is measurable, it follows from Lemma \ref{measurableresult} that the operators are mixing. Then, Lemma \ref{trivialcentralizer} implies that the centralizer of $\psi$ is trivial. Finally, a well-known result of Connes (Lemma \ref{factorcentralizer} in appendix \ref{app:op}) implies that the spectrum of all states is $\mathbb{R}_+$, hence the algebra is type III$_1$.

The converse of the Leutheusser-Liu conjecture cannot hold either. To see this, assume that we have a type III$_1$ algebra with a non-trivial centralizer. If their conjecture were true that $\rho(\omega^2)$ would have to be measurable. Then, Lemma \ref{measurableresult} implies that the algebra is mixing, and therefore has a trivial centralizer which constitutes a contradiction. Therefore, there exist GFF algebras that are type III$_1$ factors but do not even have a measurable spectral density.


 \subsection{Time-band algebras of GFF}\label{sec:timeband}

We are now ready to construct the time-band algebras $\mB_I$.
The set of Schwartz functions supported on a time interval $I=(t_1,t_2)$ forms a closed linear subspace $\mS_I\subset \mK$. Weyl operators $W(f)$ with $\ket{f}\in \mS_I$ are closed under multiplication and the $\dagger$ operation of the Fock space: $W(f)^*=W(-f)$. The formal sums $\sum_i c_i W(f_i)$ generate a $*$-algebra. As we review in appendix \ref{app:op}, there exists a natural C$^*$-norm, and closing the $*$-algebra in this norm gives a C$^*$-algebra of the time interval  $\mB_I$ that is independent of any choice of state.

There exists an analog of the Reeh-Schlieder theorem for these C$^*$-algebras in time:

 \begin{lemma}[Reeh-Schlieder in time]\label{Reehschliederpure}
     Consider a theory of GFF with $\mB_{(-\infty,\infty)}$ a type I algebra irreducibly represented on $\mH$. The C$^*$-algebra $\mB_I$ is associated with a finite time band $I$. Consider $\ket{\Psi}\in \mH$ a pure state such that $e^{i Ht}\ket{\Psi}$ is an entire vector. Then, the closure of $\overline{\mB_I\ket{\Psi}}=\mH$. 
 \end{lemma}
 \begin{proof} We simply mimic the Reeh-Schlieder proof in \cite{witten2018aps}.
 Consider any vector $\ket{\chi}\in \mH$ and matrix elements 
 \begin{eqnarray}\label{overlap}
     c_\chi(t_1,\cdots ,t_n)=\braket{\chi|\varphi(t_1)\cdots \varphi(t_n)|\Psi}\ .
 \end{eqnarray}
 We prove that if $\ket{\chi}\in \mH$ is perpendicular to all $\varphi(t_1)\cdots \varphi(t_n)\ket{\Psi}$ with all $t_i\in I$ and $n$ an arbitrary integer then $\ket{\chi}$ is perpendicular to all $\varphi(t_1)\cdots \varphi(t_n)\ket{\Psi}$ with $t_i\in (-\infty,\infty)$. Therefore, by the definition of $\mH$ it must be that $\ket{\chi}=0$. To establish this, we choose $t_{n} \in I$ and note that $c_\chi(t_1,\cdots , t_n + u) = \braket{\chi|\varphi(t_1)\cdots e^{iHu}\varphi(t_n)e^{-iHu}|\Psi}$ is analytic in the upper half-plane $\Im(u)> 0$ because $H$ is bounded from below and $e^{iHu}\ket{\Psi}$ is assumed to be entire. Moreover, we note that $c_\chi(t_1,\cdots , t_n+u)$ is continuous on real $u$ axis and that $c_\chi(t_1,\cdots , t_n+u) = 0$ for $u \in I_\epsilon = [-\epsilon,\epsilon]$ since $t_{n}+u \in I$ for small enough $u$. As shown in \cite{witten2018aps}, these properties are enough to prove that $c_\chi(t_1,\cdots , t_n+u)$ is analytic for $ u \in I_\epsilon$.
The idea of \cite{witten2018aps} is that for any $u$ such that $\Im(u)\ge 0$, we can write the Cauchy integral
 \begin{eqnarray}
     c_\chi(t_1,\cdots, t_n+u)=\oint_\gamma du' \frac{c_\chi(t_1,\cdots, t_n+u')}{u'-u}\ ,
 \end{eqnarray}
 where the contour $\gamma$ is in the upper half-plane and it encloses the point $u$. Since $c_\chi$ vanishes on the interval $u\in I_\epsilon$, we can choose $\gamma$ in the upper half-plane with a piece on $I_\epsilon$ (the real axis), and can remove the contribution from the piece on $I_\epsilon$ from the Cauchy integral. This implies that the Cauchy integral remains analytic even when $u \in I_\epsilon$. Now since $c_\chi$ vanishes for $u\in I_\epsilon$ and is analytic for $u\in I_\epsilon$, it must vanish on the whole real $u$ axis. 
 Repeating this argument for each $t_n$, one by one, we can take the operators out of $I$, and we find that $c_\chi(t_1,\cdots ,t_n)$ in (\ref{overlap}) vanishes for any $t_{i} \in (-\infty,\infty)$
 , therefore $\ket{\chi}=0$.
 \end{proof}

 The theorem above is relevant for the pure state of $\mathcal{N}=4$ SYM in the $N\to \infty$ limit. The set of vectors $\ket{\Psi}$ with entire $e^{iHt}\ket{\Psi}$ is dense in the Hilbert space $\mH$. There are two ways to see this. First, define $P_E$ the projection to the subspace of Hamiltonian with energies lower than a cut-off $E$. Then, any vector $\ket{\Psi}\in P_E\mH$ leads to entire vectors $e^{i Ht}\ket{\Psi}$. Second, any vector smoothed over time with a Gaussian 
 \begin{eqnarray}
     \ket{\tilde{\Psi}}\sim\int_{-\infty}^\infty dt\, e^{-\gamma t^2} e^{iHt} \ket{\Psi}\ .
 \end{eqnarray}
 leads to the entire vector $e^{iHt}\ket{\tilde{\Psi}}$.

 It is also worthwhile to note that in the proof of the above theorem, we did not need $c_\chi$ to be analytic on the whole upper half-plane. In fact, if $c_\chi$ was only analytic on a strip $0\le \Im(u) \le u_*$, then the argument from \cite{witten2018aps} presented above would still be valid. This allows us to present the following corollary.
 
\begin{corollary}\label{ReehSchliederKMS}
    Consider a theory of GFF in a KMS state represented by the thermofield double vector $\ket{1}_\beta$. If $\mB_I$ is a time band algebra, then $\overline{\mB_I\ket{1}_\beta}=\overline{\mB_{(-\infty,\infty)}\ket{1}_\beta}$.
\end{corollary}
\begin{proof}
    We repeat the argument in Lemma \ref{Reehschliederpure} replacing the Hamiltonian with the modular Hamiltonian of the KMS state $\hat{K}$ obtaining the result.\footnote{For a discussion of the Reeh-Schlieder property in the KMS state of a conventional QFT see \cite{jakel2000reeh}.} Note that $e^{i\hat{K}t}\ket{1}_\beta=\ket{1}_\beta$ is entire in $t$. 
\end{proof}


 In local QFT, to any open set of spacetime $\mathcal{D}$, we associate a C$^*$-algebra $\mB_{\mathcal{D}}$. 
 It follows from the timelike tube theorem in \cite{strohmaier2023timelike} that the double commutant of the C$^*$-algebra of $\mathcal{D}$ is the von Neumann algebra of the time-like envelope of $\mathcal{D}$ (the causal completion of $\mathcal{D}$) \begin{eqnarray}
    \mB_{\mathcal{D}}''=\mA_{\mE(\mathcal{D})}\ .
\end{eqnarray}
Following \cite{strohmaier2023timelike}, we denote the time-like envelope of $\mathcal{D}$ by $\mE(\mathcal{D})$. The von Neumann algebras of local QFT are factors. Moreover, the algebra of topologically trivial regions in flat spacetime satisfies {\it Haag's duality}: 
\begin{eqnarray}
    (\mA_{\mE(\mathcal{D})})'=\mA_{\mE(\mathcal{D})'}
\end{eqnarray}
where the notation $\mathcal{D}'$ means the causal complement of $\mathcal{D}$ \cite{haag2012local}.\footnote{Note that for general C$^*$-algebra, there is no reason to expect Haag's duality.}

As we will see in the next subsection, the GFF C$^*$-algebras of time intervals, i.e. $\mB_I$, are not always von Neumann algebras. By von Neumann's bi-commutant theorem, they are von Neumann algebras if and only if they are equal to their bicommutant in the Fock space; see appendix \ref{app:op}. We will see examples of constructions in GFF algebras in holography that are von Neumann algebras. 
However, as opposed to local QFT, the algebras of time interval need not satisfy a timelike analog of Haag's duality. If $I=(s,s+T)$ is some time interval, we define its complement as $I'=(-\infty,s)\cup (s+T,\infty)$. We say that the time interval algebra $\mB_I$ satisfies {\it Haag's duality in time} if $(\mB_I)'=\mB_{I'}$. The Lemma \ref{doublecommut} shows that only in very special situations, one can expect an analog of Haag's duality for certain time interval algebras.

\begin{lemma}[Haag's duality in time]\label{doublecommut}
    The time interval algebra $\mB_I$ for a finite time interval $I=(s,s+T)$ satisfies timelike Haag's duality, $(\mB_I)'=\mB_{I'}$ only if $\rho(\omega^2)$ is an entire function of $\omega$ bounded above by $|\rho(\omega^2)|< C e^{|\omega|T}$. 
\end{lemma}









\begin{proof} Assume Haag's duality in time for time interval algebra $\mB_I$ of GFF. Then, by definition, for all functions $f$ support on $I$ and $g'$ supported on $I'$ we have
\begin{eqnarray}
 [\varphi(f),\varphi(g')]=0\ .
\end{eqnarray}
Consider two functions $f,g$ supported on $I$. Then, the commutator $[\varphi(f_t),\varphi(g)]=G_{fg}(t)$ is a bounded continuous function of time that has compact support on $t\in (-T,T)$.\footnote{For any pair of operators in the algebra the (LL)-thermal correlator $\braket{a(t)b}_\beta$ is a continuous function that is the boundary value of a  function that is holomorphic a strip $\Im(t)\in [0,\beta]$. Therefore, the commutator is also continuous.} As a result, $G_{fg}(t)$ is in $L^2(-T,T)$. It follows from the Paley-Wiener theorem that we mentioned in section \ref{sec-info-loss} that its Fourier transform $\rho(\omega^2)$ is an entire function of real frequency $\omega$ that is of exponential type $T$. This finishes the proof.
\end{proof}

\subsection{GFF algebras in holography}

In holography, we are interested in conformal GFF: $d\mu(\omega)=d\omega \Theta(\omega)\omega^{2\nu}$. The discussions of the previous section imply that the von Neumann algebra of GFF in the KMS state $\mA=\mB''_{(-\infty,\infty)}$ above the Hawking-Page phase transition is a type III$_1$ factor. However, the time band algebras we defined in the previous section need not be von Neumann algebras. 
This raises two questions in holography. What is the time-band C$^*$-algebra $\mB_I$ we defined above in the bulk and what is the boundary analog of the von Neumann algebra of causally complete regions of the bulk?

We now go back and revisit our choice of the space of functions that led to the one-particle Hilbert space. 
As long as the spectral density is a function (not a sum of Dirac delta functions) 
we can absorb the choice of spectral density into the choice of functions by $f(\omega)\to \sqrt{\rho(\omega^2)}f(\omega)$.
There are two equivalent pictures: 
\begin{enumerate}
    \item Fixing the algebra of functions (for instance to be Schwartz functions) and changing the spectral density.
    \item Fixing the spectral density, for instance, $\rho(\omega^2)=1$ and choosing the corresponding GFF field $\varphi_1(t)$, and changing the algebra of functions.
\end{enumerate}
As we close the vector space of functions under the inner product set by the spectral density $\rho(\omega^2)$ both pictures become equivalent and give the one-particle Hilbert space $\mH_\mu=L^2(\mathbb{R}_t,d\omega \Theta(\omega)\rho(\omega^2))$. Therefore, we can view the choice of measure $d\mu(\omega)$ and the choice of the algebra of functions, interchangeably.

\begin{figure}[t]
    \centering
    \includegraphics[width=0.6\linewidth]{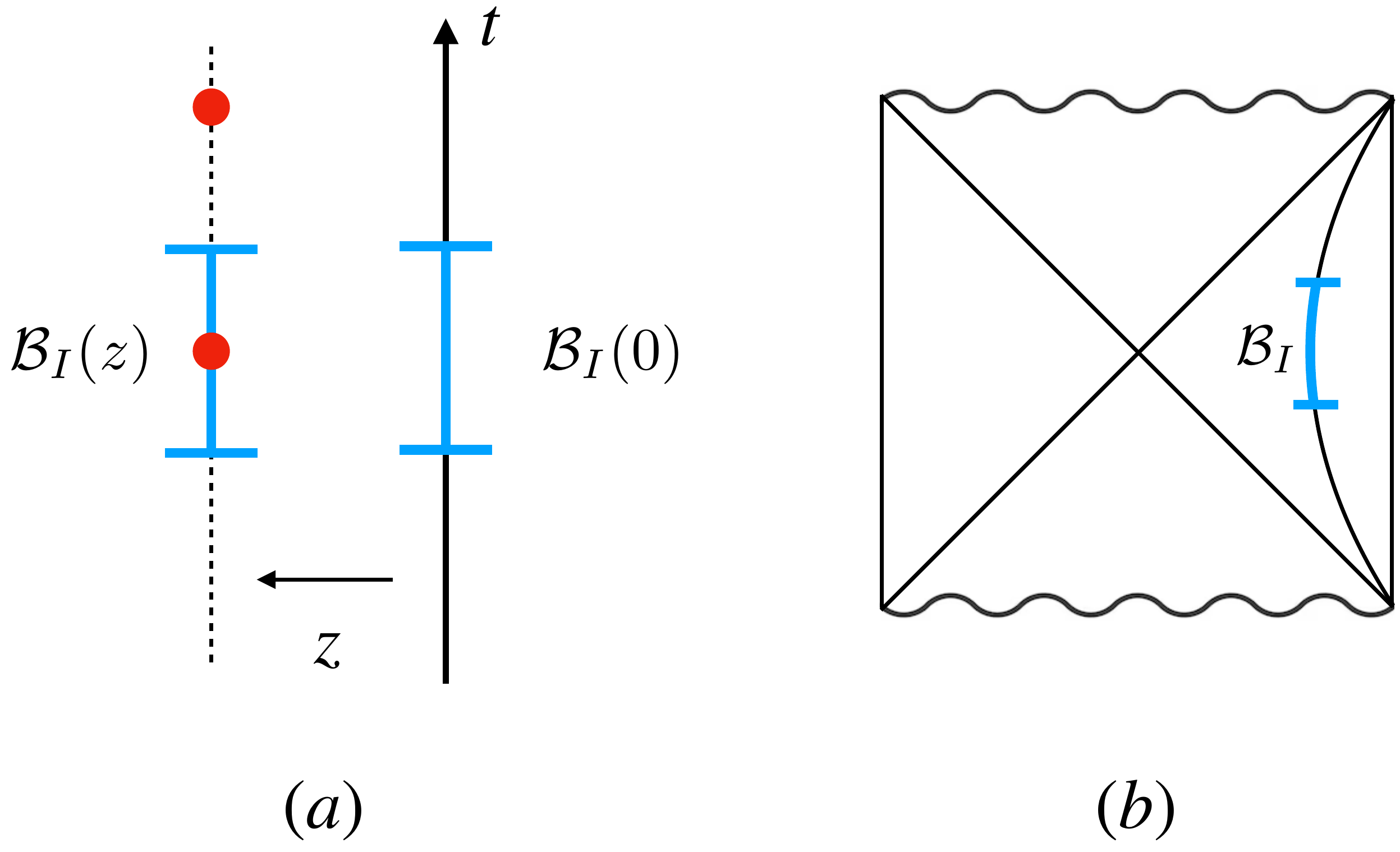}
    \caption{ \small{(a) In the Poincar\'e-patch of AdS, we can associate C$^*$-algebras $\mB_I(0)$ to a time-band $I$ defined at any constant $z$ in the bulk. In general, operator insertions on the red dots do not commute. (b) Similarly, using the HKLL reconstruction for an eternal black hole, we can associate various algebras to time intervals on time-like surfaces in the bulk.}}
    \label{fig:conformalGFF}
\end{figure}

{\bf Bulk subalgebras of conformal GFF:} Let us introduce an auxiliary positive parameter $z$, and consider GFF fields in (\ref{GFF01}) with the spectral density
\begin{eqnarray}\label{GFFbulkz}
    \sqrt{\rho_z(\omega^2)}=\frac{z^{1/2}}{\sqrt{2}}J_\nu(z\omega)
\end{eqnarray}
where $J_\nu$ are the Bessel functions of the first kind. This is a deformation of the conformal GFF because in the limit of very small $z$ we have the asymptotic expansion
\begin{eqnarray}
    &&J_\nu(z\omega)\approx\frac{2^{-\nu}}{\Gamma(\nu+1)}(z\omega)^\nu (1+ O(z\omega)^2)\nn\\
    &&\sqrt{\rho_z(\omega^2)}\approx\lb \frac{z^\Delta 2^{-(\nu+1/2)}}{\Gamma(\nu+1)}\rb\omega^\nu +O(z\omega)^2
\end{eqnarray}
which is the same as the conformal GFF up to a constant field multiplying the spectral density. At fixed $z$, our choice of the Hermitian bilinear is
\begin{eqnarray}
    \braket{f|g}_z&&=\frac{z}{2}\int d\omega \Theta(\omega) |J_\nu(z\omega)|^2 \hat{f}^*(\omega)\hat{g}(\omega)\ .
\end{eqnarray}
Define the GFF field
\begin{eqnarray}
    \varphi_z(t)=\int d\omega \Theta(\omega)z^{1/2}J_\nu(z\omega)(e^{i\omega t}a^\dagger_\omega+e^{-i\omega t}a_\omega)=\varphi_1(f_z)\ .
\end{eqnarray}
There are two ways to interpret the above expression. We can view it as modifying the spectral density of the field to (\ref{GFFbulkz}), or as $\varphi_1(f_z)$ for a special function $f_z(\omega)=\sqrt{2\rho_z(\omega^2)}$. 
The first point of view defines ``bulk fields" at fixed $z$ (deeper in the bulk in the Poincar\'e patch; see figure \ref{fig:conformalGFF}:
\begin{eqnarray}
    \varphi_z(\omega)=\frac{z^{1/2}}{\sqrt{2}}\omega^{-\nu}J_\nu(z\omega)\varphi_\Delta(\omega)\ .
\end{eqnarray}
where $\varphi_\Delta(t)$ is the conformal GFF defined in (\ref{GFFfield}). The above equation is the bulk HKLL reconstruction map in vacuum AdS in the Poincar\'e patch \cite{hamilton2006holographic}.
In the second point of view, we take the field to be the one that includes all masses equally, i.e. $\varphi_1(t)$ in (\ref{Lebesqchoice}), and replace Schwartz functions with a particular choice of $*$-subalgebra of functions in $L^2(\mathbb{R}_t,d\omega \Theta(\omega)\rho(\omega^2))$. For every fixed value of $z$ and a time interval on the boundary, we obtain a time-interval C$^*$-algebra $\mB_I(z)$; see figure \ref{fig:conformalGFF}. 


\begin{figure}[t]
    \centering
    \includegraphics[width=0.6\linewidth]{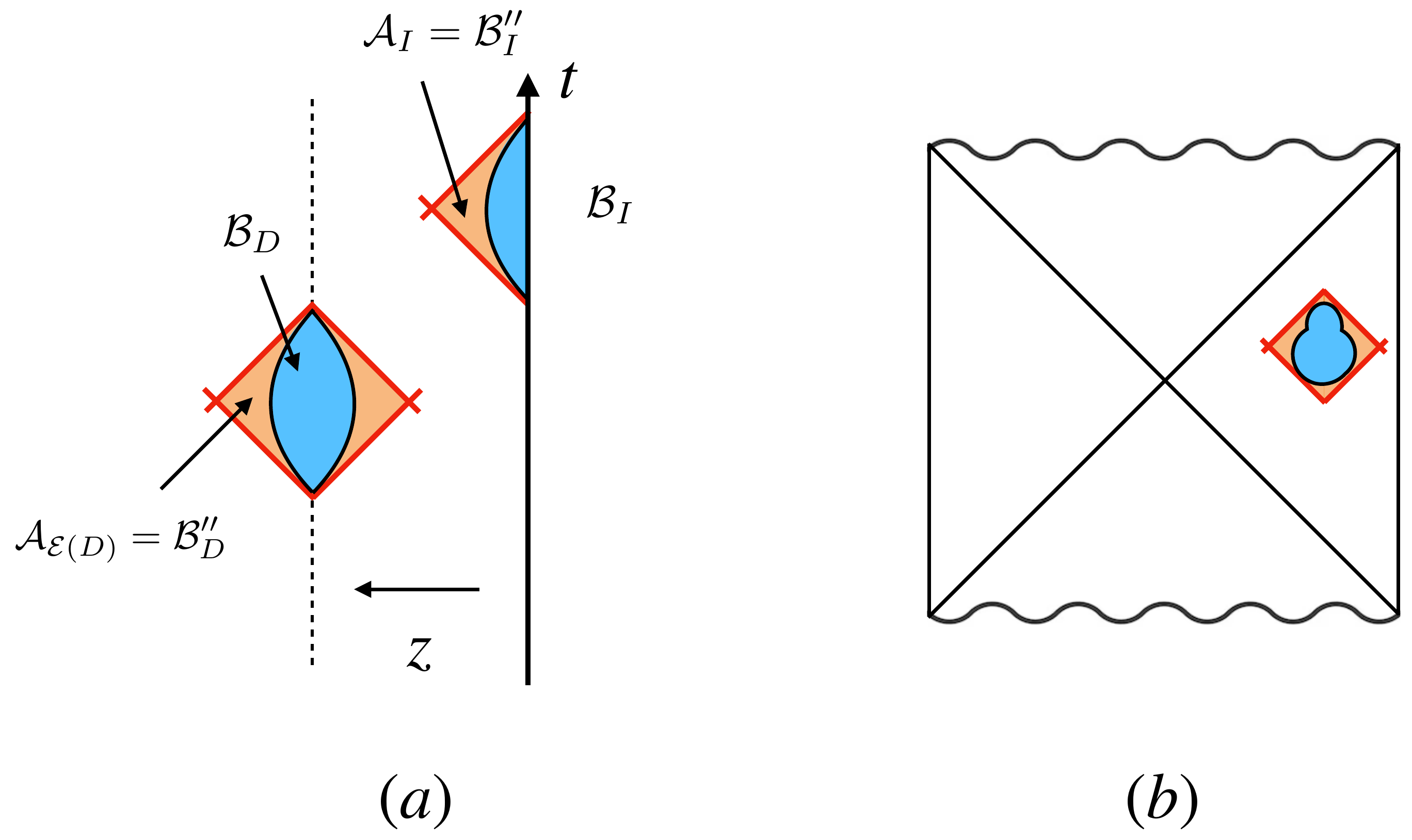}
    \caption{ \small{(a) In the Poincar\'e patch of AdS, we can associate C$^*$-algebras with any open set in the bulk, e.g. blue regions. By the timelike tube theorem, the double-commutant of these algebras are the von Neumann algebras of causally-complete regions in the bulk, i.e., $\mB''_{\mathcal{D}}=\mA_{\mE(\mathcal{D})}$ and $\mB''_{I}=\mA_{I}$. (b)} Similarly, we can define such C$^*$ and von Neumann algebras in other geometries such as an eternal black hole.}
    \label{fig:bulkGFF}
\end{figure}

In holography, we consider the GFF theory that includes $\varphi_z(\omega)$ for all $z\in (0,\infty)$. Since we can add fields with different values of $z$, we have to enlarge our set of functions from $\hat{f}(\omega)$ to Schwartz functions $\hat{f}(z,\omega)$ that solve the ``bulk" equations of motion:
\begin{eqnarray}
    \lb -z^2(\p_z^2-\p_t^2)+\Delta(\Delta-1)\rb f(z,t)=0\ .
\end{eqnarray}
The solutions to these equations are precisely $\sqrt{\rho_z(\omega^2)}$ that we chose in (\ref{GFFbulkz}). 
This is equivalent to choosing the space of Schwartz functions of $t,z$ on vacuum AdS$_2$ geometry in the Poincar\'e patch and choosing the measure
\begin{eqnarray}
    \braket{f|g}=\int dz d\omega \rho_z(\omega^2)\hat{f}^*(z,\omega)\hat{g}(z,\omega)\ .
\end{eqnarray}
The imaginary part of the inner product above gives the bulk advanced commutator of two ``bulk" fields:
\begin{eqnarray}
   2i\Im\braket{\delta(z,t)|\delta(z',t')}= [\varphi(z,t),\varphi(z,t')]=G_{adv}
\end{eqnarray}
which vanishes when $(z,t)$ and $(z',t')$ are spacelike separated.\footnote{To show this explicitly, it is helpful to use the Hankel's identity $\int_0^\infty tdt  J_\nu(zt)J_\nu(z't)=z^{-1}\delta(z-z')$ \cite{dutsch2003generalized}.} 
This allows for the possibility of associating von Neumann algebras to boundary time intervals by taking the $*$-subalgebra of functions with support on the {\it timelike envelope of the time interval $I$} \cite{strohmaier2023timelike}.

More generally, for any region of the bulk $\mathcal{D}$ we can associate a $*$-subalgebra of GFF $\mB_{\mathcal{D}}$ by restriction to the $*$-algebra of Schwartz functions support on $\mathcal{D}$; see figure \ref{fig:bulkGFF}. It follows from the timelike tube theorem in \cite{strohmaier2023timelike} that the double commutant of the C$^*$-algebra of $\mathcal{D}$ is the von Neumann algebra of the time-like envelope of $\mathcal{D}$:
\begin{eqnarray}
    \mB_{\mathcal{D}}''=\mA_{\mE(\mathcal{D})}\ .
\end{eqnarray}

The discussion above generalizes to any arbitrary bulk geometry. In holography, for $\mathcal{M}$ a $d+1$-dimensional bulk geometry with the metric $g_{\mu\nu}$ and any open set $\mathcal{D}$ in this geometry and a choice of bulk free fields, we have a C$^*$-algebra of GFF on the boundary. If the bulk region is causally complete, i.e., $\mathcal{D}=\mE(\mathcal{D})$, the corresponding GFF algebra is a von Neumann algebra. 

As an example consider a scalar field of mass $M$ in the bulk, and the equations of motion 
\begin{eqnarray}
    (\Box_g+M^2)\varphi(X)=0\ ,
\end{eqnarray}
where $X$ are bulk coordinates, and $x$ are boundary coordinates. 
For any bulk Schwartz function $g$ supported on a bulk open set $\mathcal{D}$ we have the GFF operator
\begin{eqnarray}
    &&\varphi(g)=\int_{\mathcal{D}} d^{d+1}X\: g(X) \varphi(X)\nn\\
   && \varphi(X)=\int dx \: f(X|x) \mO(x) 
\end{eqnarray}
where we have used the HKLL bulk reconstruction map \cite{hamilton2006holographic} to write the bulk fields in terms of the boundary GFF. The function $f(X|x)$ is the formal inverse of the bulk-boundary correlation function 
\begin{eqnarray}
    &&f(X|x)=\int dy \braket{\varphi(X)\varphi_\Delta(y)}K(y,x)\nn\\
    &&\int dy K(x,y) \braket{\varphi_\Delta(y)\varphi_\Delta(z)}=\delta(x-z)\ .
\end{eqnarray}
See \cite{hamilton2006local,papadodimas2013infalling} for explicit expressions in the eternal black hole geometry.

\subsection{Blackholes in microcanonical ensemble}

Consider an eternal AdS black hole in a microcanonical ensemble of energy $E_0=O(N^2)$ and width $O(1)$. This means that the TFD state is replaced with 
\begin{eqnarray}
    \ket{\overline{TFD}}=e^{-S(E_0)}\sum_i e^{-\beta(E_i-E_0)/2} f(E_i-E_0) \ket{E_i,E_i}
\end{eqnarray}
where $S(E_0)$ is the microcanonical entropy, and $f(E-E_0)$ is any smooth and invertible function of energy that is $N$-independent and normalizes to \cite{chandrasekaran2022large}
\begin{eqnarray}
    \int dE |f(E-E_0)|^2=1\ .
\end{eqnarray}
It was argued in \cite{chandrasekaran2022large} that the algebra one finds by acting with the one-point removed single trace operators on $\ket{\overline{TFD}}$ is independent of the function $f(E-E_0)$ and isomorphic to the large $N$ algebra in the canonical ensemble (the TFD state). Clustering in modular time would establish that the algebra is type III$_1$. The modular Hamiltonian in this state is 
\begin{eqnarray}
   \hat{K}=\int dEdE' \:f(E-E_0)f^{-1}(E'-E_0)\ket{EE'}\bra{EE'}\ .
\end{eqnarray}
As we argue in appendix \ref{app:typeiii}, the fact that the Hamiltonian is not affiliated with the algebra $\mA_R$ already implies that the algebra is type III.\footnote{Note that $\hat{K}\ket{\overline{TFD}}$ cannot be prepared by the action of any one-point removed single trace operator on $\ket{\overline{TFD}}$.}
Comparing to the expression in (\ref{fomega2}) we find that the function $f(E-E_0)$ plays the role of $\Omega(E)$. As we saw in Lemma \ref{simpleop}, as long as $f(E-E_0)$ is measurable and matrix elements of one-point removed single trace operators are only a function of energy and not the microstate, they form a mixing algebra in the $N\to \infty$ limit. We argued for this using the eigenstate thermalization hypothesis (ETH):
\begin{eqnarray}
    \braket{E;\alpha|\mO|E';\beta}=\delta(E-E')\delta_{\alpha\beta}\mO(E)+e^{-S(E)}R_{\alpha\beta}f(E-E';E+E')\ .
\end{eqnarray}
This form is preserved under the multiplication of an $O(1)$ number of simple operators. In the strict $N\to \infty$ limit $S(E)\to \infty$, the simple operators are independent of energy microstates. 
We postpone further exploration of the connection to ETH to upcoming work.




\section{Summary and discussions}\label{sec:conclude}

In summary, we first argued that Maldacena's information loss discussion in eternal black holes implies that the observable algebra of the boundary theory at $N\to\infty$ cannot be type I. This originates from the fact that type I algebras do not allow the existence of any mixing operators. 
We found necessary and sufficient conditions for the existence of mixing operators. The information loss problem has a counterpart for an out-of-equilibrium state if one replaces time evolution with modular time evolution. It is interesting to use this generalization to study the algebra of mixing operators in a pure state dual to a one-sided black hole.

In type I algebras, the time-evolution operator is inside the algebra, and the spectrum of its generator (modular Hamiltonian) is discrete. We showed that when the spectrum of Hamiltonian is continuous there always exist operator systems of mixing observables. 
The maximal mixing operator system has two key properties:
1) It is closed under time evolution, and 2) It cannot contain the time evolution unitary $e^{i Ht}$ for any value of $t$. In the thermodynamic limit of infinite entropy (infinite volume or infinite $N$), it is expected that the set of mixing operators closes under multiplication and forms an algebra. Using the two key properties above we proved that the resulting algebra is a type III$_1$ von Neumann factor. 
%
Applying our result to a theory of generalized free fields at finite temperature settled the conjecture of Leutheusser and Liu  by proving when the spectral density $\rho(\omega^2)$ is a measurable function of $\omega$ (e.g. above the Hawking-Page phase transition) the observable algebra is a type III$_1$ factor.

The result that fixed the algebra type of GFF relied on Wick's theorem. A CFT at finite temperature in infinite volume is not Gaussian. Nonetheless, since $\hat{f}_{\mO\mO}$ for all primaries is a measurable function of frequencies, then the global observable algebra is type III$_1$ factor. This conclusion generalized to the microcanonical ensemble or arbitrary out-of-equilibrium state relates the algebra type to the clustering of conformal primaries in modular time.

We explicitly constructed the GFF algebra associated with any time interval and discussed the ambiguities in the choice of such algebras. In holography, using the HKLL reconstruction, in every geometry $g_{\mu\nu}$, we found a set of choices that lead to GFF C$^*$-algebra corresponding to any open set $\mathcal{D}$ in the bulk spacetime.
If the open set $\mathcal{D}$ is causally complete the resulting algebra is equal to its double commutant and hence a von Neumann algebra. It is worth noting that for any choice of bulk regions that have time-intervals $I_t:=(t,\infty)$ as the boundary we have a net of C$^*$-inclusions: $\mB_{I_t}\subset \mB_{I_s}$ for all $t>s$. The boundary time evolution corresponds to the restriction map in this net. In \cite{uhlmann1977relative}, using the interpolation theory of non-commutative $L_p$ spaces, Uhlmann defined the relative entropy of any two states for C$^*$-algebra $I_t$. This immediately implies a second law of thermodynamics. We postpone the study of Uhlmann's relative entropy for these time interval algebras to upcoming work. 


Finally, we make a comment about the uniqueness of the emergent type III$_1$ algebras of mixing operators. It is known that the {\it hyperfinite} type III$_1$ factor is unique. Here, we argued that the algebras of mixing operators are type III$_1$ factors, but we did not show they are hyperfinite. The proof of hyperfiniteness often relies on nuclearity-type assumptions \cite{buchholz1987universal} which restrict the growth of the modular density of states at high modular frequencies. It would be interesting to explore this connection in the context of a variation of the Leutheusser-Liu conjecture for hyperfinite type III$_1$ algebras.

Overall, we hope to have provided enough evidence that studying the emergence of III$_1$ algebras in physics can provide clues to the emergence of a local spacetime in holography (and potentially the emergence of hydrodynamics) which go well beyond a mathematical curiosity.

\begin{acknowledgments}
NL would like to thank the Institute for Advance
Study for their hospitality and the NSF grant PHY1911298. We thank Venkatesa Chandrasekaran, Yidong Chen, Thomas Faulkner, Marius Junge, Jonah Kudler-Flam, and Nicholas LaRacuente. NL is very grateful to
the DOE that supported this work through grant DE- SC0007884 and the QuantiSED Fermilab consortium.
\end{acknowledgments}

\appendix

\section{Basics of operator algebras}\label{app:op}

\subsection{Operator systems and $*$-algebras}

Consider a separable Hilbert space $\mH$ and the algebra of bounded operators on it $B(\mH)$. A $*$-operation is an abstraction of Hermitian-conjugation $a\to a^\dagger$.\footnote{Formally, it is an involution $a\to a^*$ with the properties
\begin{eqnarray}
    (a+b)^*\to a^*+b^*, \: (ab)^*=b^* a^*, \: 1=1^*, \: (a^*)^*=a\ .
\end{eqnarray}}
We say a set $\mX$ is $*$-closed if it for all $a\in \mX$ we have $a^\dagger\in \mX$. We will use the notation $*$ and $\dagger$ interchangeably.

In physics, we are often interested in subspaces of algebras of operators generated by a collection of observables (self-adjoint operators $\mO_i=\mO_i^\dagger$) which are automatically $*$-closed and include the identity operator. For example, say an experimentalist can measure any observable from some set $\mX=\{1,\mO_1, \mO_2\cdots \}$.\footnote{Of course, she always has access to the identity operator.} Realistically, she has $k$-units of resources (time, energy, or grant money) and each measurement consumes a unit of resource. In practice, she has access to an operator system generated by these observables $\mS_k$ that is the complex polynomials of degree $k$ generated by $\mX$:
\begin{eqnarray}
    \mS_k=\left\{\sum_{i_1\cdots i_m} c_{i_1\cdots i_m}\mO_1^{i_1}\cdots \mO_m^{i_m}| \mO_i\in \mX, i_1+\cdots i_m=k \right\}\ .
\end{eqnarray}

\begin{definition}[Operator System]\label{operatorsys}
    An operator system $\mS$ is a closed linear subspace of $B(\mH)$ that contains the identity operator and is $*$-closed.\footnote{$*$-closed means that if $a\in \mS$ so is $a^\dagger\in \mS$.} 
\end{definition}

\begin{definition}[Irreducibility]\label{irrop}
    An operator system $\mS\subset B(\mH)$ is called irreducible if $\mS$ transforms no proper subspaces of $B(\mH)$ to itself.
\end{definition}
Irreducible operator systems play an important role in this work. It follows from Schur's Lemma that 
\begin{enumerate}
    \item $\mS$ is irreducible if and only if $\mS'$ is trivial (proportional identity).
    \item $\mS$ is irreducible if and only if $\mS''=B(\mH)$.
\end{enumerate}
It is natural to consider $\mO_i$ to be independent measurements so that they correspond to an orthogonal basis for $\mS_k$. Clearly, $\mS_{k'}\subset \mS_k$ for all $k'\leq k$. In the limit $k\to \infty$, we obtain $\mS_\infty$ which is a linear subspace closed under multiplication. It is a $*$-algebra generated by $\mX$.

\begin{definition}[$*$-algebra]\label{Cstar}
    A subalgebra of $B(\mH)$ is called a $*$-algebra $\mB$ if it contains the identity operators, and it is $*$-closed.
\end{definition}

 Another example important in this work is large $N$ gauge theory, and the operator system $\mS_k$ of up to $k$-trace operators generated by single-trace operators $\mO_m=\text{tr}(\Phi^m)$. 
 In the limit $N\to \infty$, the operators $\mO_m$ become orthogonal and they form a $*$-algebra.

Consider an arbitrary density matrix $\rho$ on $B(\mH)$. The restriction of $\rho$ on a $*$-algebra $\mB$ is a {\it normal} state on it.\footnote{Every operator $a\in \mB$ can be viewed as a linear map from density matrices $\rho$ to complex numbers. The reason for the adjective ``normal" will become clear.} However, if we think of states as positive linear functionals from $\mB$ to complex numbers $\psi:\mB\to \mathbb{C}$ there can be a lot more states than those described by density matrices $\rho$ of $B(\mH)$.
If the density matrix has full rank in $B(\mH)$ (the state is faithful\footnote{A state is $\psi:\mA\to \mathbb{C}$ faithful if $\psi(a^\dagger a)=0$ for $a\in \mA$ implies $a=0$. If the state is a density matrix on $\mA$ then this is equivalent to the condition that the density matrix is full rank.}) it defines two inner products on $\mS$,\footnote{A general normal state $\rho$ leads two positive semi-definite Hermitian forms on $\mB$, and their corresponding semi-norms.}
\begin{eqnarray}\label{seminorms}
    &&\braket{a|b}_{\rho,L}= \text{tr}(\rho a^\dagger b)\nn\\
    &&\braket{a|b}_{\rho,R}= \text{tr}(\rho b^\dagger a)\ .
\end{eqnarray}
Using the interpolation theory of positive Hermitian forms on $*$-algebras we interpolate between the two inner products above  \cite{uhlmann1977relative}:
\begin{eqnarray}\label{interp}
    \braket{a|b}_{\rho,p}=\text{tr}(\rho^{1-p}a^\dagger \rho^{p}b)
\end{eqnarray}
where $p$ is a complex number in the strip $\Re(p)=(0,1)$. 
When $p$ is real the bilinear above is the Kosaki $(\rho,\rho,1/p)$-norm \cite{furuya2023monotonic}. Closing the algebra in any of these norms gives a {\it Banach algebra}, i.e. a norm-complete $*$-algebra.

If $\rho$ is the thermal state $\rho_\beta=e^{-\beta H}/Z$ then the interpolated bilinear form in (\ref{interp}) is precisely the analytically continued thermal two-point function we were studying in (\ref{analext}):
\begin{eqnarray}
    \braket{a|b}_{\rho_\beta,it}=f_{ab}(t-i \beta/2)\ .
\end{eqnarray}
In fact, the form above is a generalization of $f_{ab}(t)$ to a general density matrix $\rho$ where the Hamiltonian time-evolution is replaced by evolution with the modular Hamiltonian $H_\rho=-\log\rho$, i.e., modular evolution.

The modular time evolved two-point function  $\braket{a|b}_{\beta, it}$ contains all the information we need about the algebra of operators in $\mB$. In the example of our experimentalist, the $*$-algebra $\mB$ was generated by a set of observables $\mO_i$ that we will refer to as {\it fundamental fields}. Since an arbitrary $a$ is a polynomial in $\mO_i$, then the first order polynomial corresponds to the two-point functions of the fundamental fields $f_{\mO_i\mO_i}(t)$, and the bilinear form restricted to $\mS_k$ contains the information of up to $2k$-point functions of the fundamental fields.

\subsection{C$^*$-algebras and von Neumann algebras in $B(\mH)$}



For every $*$-algebra represented in the Hilbert space $B(\mH)$ and any  density matrix $\rho\in B(\mH)$ we can generalize the norms in (\ref{seminorms}) to the $L_{p,\rho}$ norms for $p\in [1,\infty]$ \footnote{We warn the reader that these norms are different from the $L_{p,\rho}$-norms defined for von Neumann algebras by Araki and Masuda in \cite{araki1982positive,furuya2023monotonic}.}:
\begin{eqnarray}\label{pnorms}
    \|a\|^{p}_{p,\rho}=\text{tr}\lb \rho(a^\dagger a)^{p/2}\rb\ .
\end{eqnarray}

The collection of $\|a\|_{p,\rho}$ for every set of $\rho$ induces a locally convex topology. When $\rho$ are all pure states of $\mH$:
the $p=1$ gives the weak topology, and $p=2$ gives strong topology. When $\rho$ are all density matrices of $\mH$: the $p=1$ gives the $\sigma$-weak topology and $p=2$ gives the $\sigma$-strong topology.\footnote{The $\sigma$-strong and $\sigma$-weak topologies are sometimes called the ultrastrong and the ultraweak topologies, respectively. Since every pure state $\ket{\Psi}$ can be viewed as a density matrix $\ket{\Psi}\bra{\Psi}^{1/2}=\ket{\Psi}\bra{\Psi}$ the $\sigma$-weak topology is finer than the weak topology.} When $p=\infty$ the induced topology is independent of $\rho$. It is called the norm topology.
For instance, we say a sequence $a_n$ converges to $a$ in the weak topology if $\lim_n\|a_n-a\|_{1,\rho}=0$ for all density matrices $\rho$.
\begin{definition}
A $*$-algebra that is closed (includes all its limit points) in the norm topology is called a {\it C$^*$-algebra}. 
\end{definition}
\begin{definition}
A $*$-algebra that is closed in the weak topology is called a von Neumann algebra.
\end{definition}
Consider a full rank density matrix $\rho\in B(\mH)$:
\begin{eqnarray}
    \rho=\sum_m p_m \ket{m}\bra{m}\ .
\end{eqnarray}
The analog of the thermofield double is the canonical purification of $\rho$ in the duplicated Hilbert space $\mH\otimes \mH'$:
\begin{eqnarray}
    \ket{1}_\rho=\sum_m \sqrt{p_m}\ket{m}\ket{m}\ .
\end{eqnarray}
To every operator $a\in \mB$ we can associate a vector in $\tilde{\mH}=\mH\otimes \mH$ with the inner product
\begin{eqnarray}
    &&a\to \ket{a}_\rho:=(a\otimes 1)\ket{1}_\rho\in \tilde{\mH}\nn\\
    &&\braket{a|b}_\rho=\text{tr}(\rho a^\dagger b)=\braket{a|b}_\rho
\end{eqnarray}
so that $\|a\|_{1,\rho}$ norm is the one-norm in $\tilde{\mH}$, and $\|a\|_{\infty,\rho}$ is the operator norm in $\tilde{\mH}$.
The $L_{p,\rho}$-norm is the Hilbert space norm 
\begin{eqnarray}
    \|a\|_{p,\rho}=|\ket{a^{p/2}}_\rho|^{1/p}\ .
\end{eqnarray}
We define the Banach space $\mB_{p,\rho}$ as the closed subspace of vectors $\ket{a^p}_\rho$. 
For $1\leq p\leq p'$ we have the inclusion order
\begin{eqnarray}
    &&\|a\|_{p,\rho}\leq \|a\|_{p',\rho}\ .
\end{eqnarray}
For a sequence of operators $a_n\in \mB$ if $a\in \mB_{p',\rho}$, i.e. $\lim_n \|a_n-a\|_{p',\rho}\to 0$, it is also in $\mB_{p,\rho}$ but not the other way around. Therefore,
\begin{eqnarray}
    &&\mB_{p',\rho}\subset \mB_{p, \rho}\ .
\end{eqnarray}
The Banach algebra $\mB_{\infty,\rho}$ is closed in the norm topology, therefore, it is a C$^*$-algebra, and the Banach algebra $\mB_{1,\rho}=\mA$ is closed in $\sigma$-weak topology, therefore it is a von Neumann algebra. 

Closing $\mB$ in the $\sigma$-weak topology means that the set of all trace class operators in $\tilde{\mH}$ is the predual of the von Neumann algebra $\mB_{1,\rho}$; see (\ref{pnorms}).
This predual is called the set of {\it normal states} of the von Neumann algebra $\mA_{1,\rho}$.
Every operator $a$ in the von Neumann algebra $\mB_{1,\rho}$ can be viewed as a linear map from the set of density matrices to complex numbers:
 \begin{eqnarray}
     a:\rho \to \text{tr}(\rho a)\ .
 \end{eqnarray}
 And for every map $f:\rho\to \mathbb{C}$ there exists an operator $a_f\in \mB_{\infty,\rho}$ such that $f(\rho)=\rho(a_f)$. 
 Furthermore, there are enough normal states (density matrices) to distinguish all operators of $\mA$, i.e. if $\text{tr}(\rho a)=0$ for all $\rho\in B(\mH)$ then $a=0$.

 We saw that the von Neumann algebra $\mB_{1,\rho}$ is larger than the C$^*$-algebra $\mB_{\infty,\rho}$. The von Neumann bicommutant theorem clarifies what are the operators that we have added to $\mB_{\infty,\rho}$ to get a von Neumann algebra.


 \begin{theorem}[Bicommutant Theorem]\label{bicom}
     If $\mA\subset B(\mH)$ is $*$-algebra with $\mA''=\mA$, then 
     \begin{enumerate}
         \item $\mA$ is $\sigma$-weakly closed.
         \item $\mA$ is strongly closed.
     \end{enumerate}
 \end{theorem}
\begin{proof}
    The proof is standard and we will not repeat it here.
\end{proof}

The theory of von Neumann algebras can be viewed as a non-commutative measure theory. Perhaps, the biggest surprise in going from commutative to non-commutative is that non-commutative von Neumann algebras have a canonical time evolution ($*$-automorphism) called modular evolution. Von Neumann algebras are {\it dynamical systems}.

The set of all states of $\mA$ forms a convex cone. The extremal rays of this cone are called {\it pure} states. Given $\psi$ a normal faithful state of a von Neumann algebra, we perform the GNS construction detailed below to obtain {\it the standard form} $\{\pi_\psi(a),\mH_\psi,\ket{1}_\psi\}$. 

The GNS construction is precisely the quantization procedure we used to define the one-particle Hilbert space in section \ref{sec:GFFalgebra}. We start with a state $\psi$ of a C$^*$-algebra. The kernel of the state is the set of $a\in \mA$ such that $\psi(a^\dagger a)=0$. They form a subspace of $\mA$ called the left ideal of $\mA$ that we denote by $I$. We identify operators in $\mA$ by the equivalence relation $\ket{a} \sim \ket{a+I}$. The state $\psi$ restricted to the quotient $\mA/I$ is faithful, therefore it yields a norm and inner product on $\mA$. The Cauchy-completion of $\mA/I$ in this norm is the GNS Hilbert space $\mH_\psi$. We define the representation $\pi_\psi(a)$ by its action on the vectors $\ket{[b]}_\psi\in \mH_\psi$:
\begin{eqnarray}
    \pi_\psi(a)\ket{[b]}_\psi=\ket{[ab]}_\psi
\end{eqnarray}
where $\ket{[a]}$ is the equivalence class of $a$ in $\mH_\psi$. There exists a unit norm cyclic vector in $\mH_\psi$ that corresponds to the identity operator in $\mA$. We denote it by $\ket{[1]}_\psi$. For simplicity, we often assume $\psi$ is faithful so that we can denote $\ket{[a]}_\psi$ by $\ket{a}_\psi$.

Modular theory starts with the definition of the Tomita operator for a von Neumann algebra and a faithful normal state $\psi$ represented in the standard form:
\begin{eqnarray}
S_\psi \pi_\psi(a)\ket{1}_\psi=\pi_\psi(a^\dagger)\ket{1}_\psi\ .
\end{eqnarray}
Taking the closure of this operator, its polar decomposition gives $S_\psi=J_\psi \Delta_\psi^{1/2}$ where the positive operator $\Delta_\psi$ is called the modular operator, and the anti-unitary $J_\psi$ is called the modular conjugation. The operator $\hat{K}=-\log\Delta_\psi$ is called the {\it modular Hamiltonian}. It generates a unitary flow that preserves the algebra called {\it the modular flow}
\begin{eqnarray}
    \alpha_\psi^t(a)=\Delta_\psi^{it}a\Delta_\psi^{-it}\in \mA\ .
\end{eqnarray}
Modular flow is a $*$-automorphism of $\mA$, i.e., it is a bijection $\alpha:\mA\to \mA$ that respect the algebra operations (multiplication, addition, and the $*$-operation).

Every $*$-automorphism of a von Neumann algebra $\alpha(a)$ represented in the standard form can be implemented by a unitary rotation $U_\alpha:\mH_\psi\to \mH_\psi$
\begin{eqnarray}
    \pi_\psi(\alpha(a))=U_\alpha \pi_\psi(a) U_\alpha\ .
\end{eqnarray}
Furthermore, we can choose $U_\alpha$ such that it commutes with the modular conjugation.\footnote{For a proof see Theorem 2.2.4 in section V of \cite{haag2012local}.} The set of unitaries $u\in \mA$ generate $*$-automorphisms that are called {\it inner} automorphisms of $\mA$. An automorphism that is not inner is called {\it outer}.

\subsection{W$^*$-dynamical system}

A $W^*$-{\it dynamical system} is a triple $\{\mA,G, \alpha\}$ where $\mA$ is a von Neumann algebra, $G$ is a locally compact group that acts on $\mA$ as (strongly continuous) automorphism of $\mA$. In other words, for any $g\in G: \alpha_g:\mA\to \mA$ such that
\begin{enumerate}
    \item If $e\in G$ is identity element then $\forall a\in \mA: \alpha_e(a)=a$
    \item Group multiplication is preserved, i.e.
    $\forall a\in\mA, g_1,g_2\in G: \alpha_{g_1}\alpha_{g_2}=\alpha_{g_1g_2}$.
    \item  The map $g\to \alpha_g(a)$ is continuous in the norm of operator $a$.
\end{enumerate}
In this work, we will often refer to a $W^*$-dynamical system simply as a dynamical system. 
A covariant representation of a dynamical system $(\mH,\pi, U_g)$ is a GNS Hilbert space $(\mH,\pi)$ with a group that acts as strongly continuous unitaries $U_g$.
An important example of a covariantly represented dynamical system is the modular flow: $(\mH_\psi, \pi_\psi, \Delta_\psi^{it})$. 


\subsection{Classification of von Neumann algebras}

We start with a few basic definitions. An algebra is called a {\it factor} if and only if it has a trivial center. Although, in a factor $\mA$, there exists no operator $z\in \mA$ that commutes with all $a\in \mA$, given a state $\psi$ it is possible that $\psi([a,z])=0$ for all $a\in \mA$. The sets of such operators form the centralizer {\it} of $\psi$ which we denote by $z\in \mA^\psi$. The center of a von Neumann algebra is inside its centralizer. Clearly, the centralizer depends on the state; however, as we will see in Lemma \ref{factorcentralizer} when the centralizer is a factor, the spectrum of modular flow is the same among all states, and hence independent of the state (it becomes an algebraic invariant).

A {\it trace} is an unnormalized state (weight) that satisfies
\begin{eqnarray}
    \text{tr}(a^\dagger a)=\text{tr}(a a^\dagger)\ .
\end{eqnarray}
A projection $e\in \mA$ is a self-adjoint operator that squares to itself, i.e. $e=e^\dagger$ and $e^2=e$.
Murray and von Neumann classified all von Neumann factors into three classes based on the values $d=\{\text{tr}(e)| e\in \mA, e=e^\dagger, e^2=e\}$ for all projections $e\in \mA$ and some $x>0$:
\begin{enumerate}
    \item Type I$_d$ if $d=\{0,x, 2x, \cdots , dx\}$.
    \item Type I$_\infty$ if $d=\{0,x,2x,\cdots \}$.
    \item Type II$_1$ if $d=[0,x]$.
    \item Type II$_\infty$ if $d=[0,\infty]$.
    \item Type III if $d=\{0,\infty\}$.
\end{enumerate}

A finer classification of these types is provided using modular theory. Consider the modular operator $\Delta_\psi^{it}$ for some state $\psi$. If for some $\tau$ this automorphism of the algebra is inner it can be factored as $\Delta_\psi^{i\tau}=u_\tau u'_\tau$ where $u'_\tau=J_\psi u J_\psi$. In fact, for that $\tau$ the modular automorphism corresponding to any other vector  $\ket{\chi}_\psi\in \mH_\psi$ is also inner because as Connes has proved we can write
\begin{eqnarray}\label{cocyclerel}
\Delta_\psi^{it}\Delta_\chi^{-it}=v_{\psi|\chi}(t)v'_{\psi|\chi}(t)
\end{eqnarray}
for unitaries $v_{\psi|\chi}\in\mathcal{A}$ and $v'_{\psi|\chi}\in \mathcal{A}'$ that are called the cocycles (See \cite{lashkari2021modular} for a review of the proof).\footnote{They are called cocycle in the sense of non-Abelian group cohomology.} Therefore, the modular flow $\Delta_\chi^{it}$ is also inner at this $\tau$:
\begin{eqnarray}
\Delta_\chi^{i\tau}=v_{\chi|\psi}(\tau)u_\tau v'_{\chi|\psi}(\tau) u'_\tau\ .
\end{eqnarray}
As a result, the set of $\tau\in \mathcal{T}$ for which $\Delta_\psi^{i\tau}=u_\tau u_\tau'$ is inner is independent of $\psi$, and a property of the algebra. The set $\mathcal{T}$ is a subgroup of $\mathbb{R}$. There are three options: 
\begin{enumerate}\label{possibilities}
    \item $\mathcal{T}=\mathbb{R}$ (We will see that this occurs if and only if the algebra is type I or type II)
    \item $\mathcal{T}=n \tau$ with $n\in\mathbb{Z}$ 
    \item $\mathcal{T}=\{0\}$ (We will see this occurs if the algebra is type III$_1$.)
\end{enumerate}

\begin{lemma}\label{typenotiii}
    The modular flow is inner for all times (possibility 1 above) if and only if the algebra is type I or type II.\footnote{Even though all automorphisms of type I factor are inner, all automorphisms of type II are not inner \cite{witten2022gravity}. Our statement is only for the modular flow, which is a special example of automorphism.}
\end{lemma}
\begin{proof}
In type I and II algebras, there exists a semifinite trace, and for any normal state $\psi$ we can define a density matrix $\rho_\psi$ affiliated with the algebra. The modular flow factors as
\begin{equation}
    \Delta_\psi^{it}=\rho_\psi^{it}J_\psi \rho_\psi^{-it}J_\psi\ .
\end{equation}
It follows that modular flow is inner for all times:
\begin{eqnarray}
    \alpha_\psi^t(a)=\rho_\psi^{it}a\rho_\psi^{-it}\ .
\end{eqnarray}
To prove the converse, we assume that the modular flow is inner at all times. Since the modular flow is inner and strongly continuous by Stone's theorem, it can be generated by the operator $\log\rho_\psi$ which is affiliated with the algebra. In other words, we can factor the modular operator as $\Delta_\psi=\rho_\psi J_\psi \rho_\psi J_\psi$ where $\rho_\psi$ is a self-adjoint operator affiliated with the algebra. Therefore, for analytic $a\in \mA$ and $\alpha,\beta\in (0,1)$ we have the relations (see \cite{haag2012local} the discussion around V.2.62)
\begin{eqnarray}
    &&\psi(\rho_\psi^\alpha a)=\psi(\rho_\psi^{\alpha-\beta} a \rho_\psi^{\beta})\nn\\
    &&\psi(\rho_\psi^{-1}a b)=\psi(b a \rho_\psi^{-1})\ .
\end{eqnarray}
As a result, we can define a semifinite trace as
\begin{eqnarray}
    \text{tr}(a):=\psi(\rho_\psi^{-1/2}a\rho_\psi^{-1/2})\ .
\end{eqnarray}
The algebra is either type I or type II.
\end{proof}
Therefore, the possibilities $\mT=n\tau$ and $\mT=\{0\}$ above only occur in type III algebras. Another algebraic invariant of a von Neumann algebra is the intersection of the spectrum of $\Delta_\psi$ over all of the normal states of $\psi$ (Arveson spectrum):
\begin{eqnarray}\label{arveson}
    &&\mS(\mA)=\cap_\psi\text{Spec}(\Delta_\psi)\ .
\end{eqnarray}
Connes used this invariant to classify type III algebra into three classes:
\begin{definition}\label{deftypeiii}
    A type III von Neumann factor is defined to be 
    \begin{enumerate}[(a)]
    \item Type III$_0$ if $\mS(\mA)=\{0,1\}$.
    \item Type III$_\lambda$ with $\lambda\in (0,1)$ if $\mS(\mA)=\{0\cup \lambda^n\}$ for all $n\in \mathbb{Z}$.
    \item Type III$_1$ if $\mS(\mA)=\mathbb{R}_+$.
\end{enumerate}
\end{definition}

The set $\mS(\mA)$ is difficult to compute. The insight of Connes was to realize that the Arveson spectrum can be related to the other algebraic invariant $\mT$ we defined above \cite{connes1973classification}. He proved that
\begin{eqnarray}\label{connesineq}
&&\ln\lb \mathcal{S}(\mathcal{A})\backslash 0\rb\subset\Gamma(\mA)    
\end{eqnarray}
where 
 \begin{eqnarray}
     \Gamma=\{\lambda\in \mathbb{R}: \:e^{i\lambda t}=1\quad \text{for all} \:t\in\mT\}\ .
  \end{eqnarray}

  \begin{lemma}
      If the algebra is type III$_1$ then the modular flow is never inner, i.e. $\mT=\{0\}$.
  \end{lemma}
  \begin{proof}
      In type III$_1$ algebras by definition the Arveson spectrum is $\mathbb{R}_+$ therefore it follows from (\ref{connesineq}) that $\Gamma=\mathbb{R}$, and as a result $\mT=\{0\}$.
  \end{proof}

Another key theorem proved by Connes (below) is particularly relevant for mixing. 
\begin{lemma}\label{factorcentralizer}
    Consider a von Neumann algebra $\mA$. If the centralizer of a normal faithful state is a factor, then the spectrum of $\Delta_\psi$ is the same for all states $\psi$ and equal to the Arveson spectrum, i.e.,
    \begin{eqnarray}
        \mS(\mA)=\text{Spec}(\Delta_\psi) \, .
    \end{eqnarray}
\end{lemma}
\begin{proof}
    For proof see Corollary 3.2.7 of \cite{connes1973classification}.\footnote{We thank Yidong Chen for pointing out this Corollary to us.}
\end{proof}

\begin{corollary}\label{corrfinal}
    Consider a von Neumann algebra $\mA$ and a normal faithful state $\psi$. If the centralizer $\mA^\psi$ is trivial and $\text{Spec}(\Delta_\psi)=\mathbb{R}_+$ then the algebra is type III$_1$.
\end{corollary}

\section{Notation}\label{app:notation}

For the reader's convenience, the notation used in this paper is summarized in the table here. 

\begin{table}[H]
    \centering
    \begin{tabular}{c|c}
        $\ket{1}_{\beta}$ & Thermofield double with inverse temperature $\beta$\\
        $\mathcal{A}$ & von Neumann algebra \\
        $\mathcal{M}$ & Mixing von Neumann algebra \\
        $\mB$ & $*$-algebra, Banach algebra, C$^*$-algebra\\
        $\psi$ & An arbitrary state\\
        $\mH_\psi$ & The GNS Hilbert space of $\mA$ and state $\psi$\\
        $\ket{1}_{\psi}$ & Vector representative of a normal state $\psi$ in the GNS Hilbert space $\mH_\psi$\\
         $\mathcal{X}$ & Set of observables (self-adjoint operators)\\
        $\mathcal{S}_{\psi}^{\mathcal{X}}$ & Operator system of $\psi$-mixing operators generated by $\mathcal{X}$\\
        $\mathcal{S}_{\psi}$ & Operator system of all $\psi$-mixing operators\\
        $\mathcal{B}_{\psi}^{\mathcal{X}}$ & $C^{*}$-algebra of $\psi$-mixing operators generated by $\mathcal{X}$ \\
        $\mathcal{M}_{\psi}^{\mathcal{X}}$ & von Neumann algebra of $\psi$-mixing operators generated by $\mathcal{X}$ \\
        $\mathcal{M}_{\psi}$ & von Neumann algebra 
        of all $\psi$-mixing operators \\
        $\mathcal{A}^{\psi}$ & Centralizer of state $\psi$\\
        $\mS(\mA)$ & Arveson spectrum of von Neumann algebra $\mA$
    \end{tabular}
    \label{tab:my_label}
\end{table}

\section{Type III algebras in the thermodynamic limit}\label{app:typeiii}

In a system with $N$ degrees of freedom, the Hilbert space grows exponentially in $N$, therefore we often think of the dimensionality of the Hilbert space $d\sim e^N$ or the exponential of entropy.
To set the notation, we write the spectral decomposition of the  dimensionless Hamiltonian as 
\begin{eqnarray}
&& H=\int dP_E \:E\:\rho(E)\nn\\
&&\rho(E)=\sum_\alpha \delta(E-E_\alpha)
\end{eqnarray}
where time is measured in the units of inverse temperature, and $dP_E$ is a projection-valued measure. 
The function $\tilde{\rho}(E)=\rho(E)/d$ is the probability density of eigenstates
$\int dE \:\tilde{\rho}(E)=1$.
In finite quantum systems, the spectrum of $H$ is discrete and the spectral projection is 
\begin{eqnarray}
dP_E=\sum_\alpha\delta(E-E_\alpha)\sum_{\gamma_\alpha=1}^{\Omega(E_\alpha)}\ket{E_\alpha;\gamma_\alpha}\bra{E_\alpha;\gamma_\alpha}    
\end{eqnarray}
 where we have assumed that the eigenvectors of energy $E_\alpha$ are $\Omega(E_\alpha)$-fold degenerate, and $\gamma$ labels degeneracies. The matrix element of an operator $\mO$ in the energy eigenbasis is
\begin{eqnarray}
&&\mO(E',E):=\sum_{\alpha,\beta}\delta(E-E_\alpha)\delta(E'-E_\beta)\braket{E_\alpha|\mO|E_\beta}=\text{tr}(dP_{E}\mO dP_{E'})\ .
\end{eqnarray}

Consider two simple harmonic oscillators with non-commensurate frequencies $\omega_1$ and $\omega_2$. Then, $m_1 \omega_1+m_2\omega_2$ with integer $m_1,m_2\in \mathbb{Z}$ is dense in the set of real numbers.\footnote{For any real number $r$ we can choose $m_1/m_2$ to approximate the real number $r/(m_1\omega_1)-1$ arbitrarily well. This is called the Diophantine approximation.} Therefore, energy differences are dense in real numbers. 
However, the time evolution of $O(E',E;t)$ can be generated using an operator in the algebra; namely $e^{i (H_1+H_2)t}$. This constitutes an example of a state of a type I algebra with a spectrum that is dense in $\mathbb{R}_+$. Of course, the diagonal elements $\mO(E,E)$ are time-independent, and hence conserved under time evolution. In addition to zero modes (conserved charges), this algebra also has recurrences. For instance, any operator that is supported only on the first or the second harmonic oscillator shows Poincar\'e recurrences because 
\begin{eqnarray}
    \mO(E'_1; E_1;t)=e^{i (n_1'-n_1)\alpha t} \mO(E'_1,E_1)
\end{eqnarray}
where $E_1$ is the energy of the first harmonic oscillator.

In the thermodynamic limit, we have countably infinite copies of finite quantum systems (or harmonic oscillators). Consider the KMS state of finite energy density. The operator $e^{-i\sum_j  H_j t}$ is no longer in the algebra because its generator $\sum_j H_j$ (total energy) has a divergent expectation value. The average energy $\lim_{n\to \infty} \frac{1}{n}\sum_{j=1}^nH_j$ is a valid observable in the algebra; however, it cannot generate the time evolution of $\mO(E',E;t)$. In the following lemma, we show that when $H$ is not affiliated with the algebra the algebra is necessarily type III.

\begin{lemma}\label{KMStypeIII}
In a KMS state, the algebra is type III if and only if the Hamiltonian is not affiliated with the algebra. 
\end{lemma}
\begin{proof} A von Neumann algebra is not type III if and only if time evolution (modular flow of the KMS state) is inner for all $t\in \mathbb{R}$ (see appendix \ref{app:op}). We prove this Lemma by contradiction. First, assume that the algebra is type III but the Hamiltonian is affiliated with the algebra, then $e^{i Ht}$ is inner for all $t$ which is a contradiction. Conversely, assume that $H$ is not affiliated with the algebra and the algebra is type not III. The modular flow in the KMS state is strongly continuous, therefore by Stone's theorem 
 it can be written as $e^{i Ht}$. Since the algebra is not type III, the flow by $e^{i Ht}$ is inner, hence its generator is affiliated with the algebra which is a contradiction. 
\end{proof}

Even in the thermodynamic limit, if we choose the system such that all eigenvalues are commensurate then the energy gaps (the eigenvalues of the modular operator) do not close, and we are left with Poincar\'e recurrences. This is the case in type III$_\lambda$ factors. In Araki's construction of type III$_\lambda$ factors they correspond to the thermodynamic limit of countably infinite qubits with eigenvalue $\{p,1-p\}$ with $\lambda=(1-p)/p$. Then, the energy gap spectrum is multiples of some smallest gap $\ep=-\log\lambda$. Generically, in an interacting theory, the gaps close as $E_n-E_m=O(e^{-N})$, and the function above becomes continuous. Such algebras are type III$_1$ and show no Poincar\'e recurrences.

\section{Special cases of the two-point correlator}\label{app:special2pt}

In this appendix, we comment on three special cases of the LR-correlator in (\ref{fabdef}): 
\begin{align}
    f_{ab}(t) & := \mbox{}_{\beta}\bra{1} a_L^\dagger  e^{it\hat{K}} b^{T}_R \ket{1}_{\beta} \ .
\end{align}

 \begin{enumerate}
    \item The case $b=a$: The Fourier transform of $f_{aa}(t)$ has a positive kernel:
     \begin{eqnarray}
          \hat{f}_{aa}(\lambda) \, = \, \sum_{m,n} e^{-\beta(E_m+E_n)/2} \, |a_{mn}|^2 \, \delta\left(\lambda - E_{m} + E_{n}\right) \, .
      \end{eqnarray}
     In other words, the function $f_{aa}(t)$ is positive definite. We have included a short review of positive definite functions below. At $t=0$ the positivity of the function follows from the fact that 
     \begin{eqnarray}
         f_{aa}(0)={}_\beta\braket{e^{-\beta \hat{K}/4} a |e^{-\beta \hat{K}/4}a}_\beta\ .
     \end{eqnarray}

     \item The case $h=h^\dagger=a=b$: Then, $f_{hh}(t)=f_{hh}(-t)=f_{hh}(t)^*$ is a real positive definite function, where we have used the fact that $[h_L(t),h_R^T]=0$. In other words, we have $\hat{f}_{hh}(E_m-E_n)=\hat{f}_{hh}(E_n-E_m)$, and we can write 
     \begin{eqnarray}\label{almostperiodeq}
         f(t)=\sum_{mn} \hat{f}_{aa}(t) \cos(\beta(E_m+E_n)/2)
     \end{eqnarray}
  to make the function manifestly real. 
    
    \item The case $a_+=X^\dagger X$ and $b_+=Y^\dagger Y$ both positive operators: Then, the function $f_{a_+b_+}(t)$ is positive: 
     \begin{eqnarray}
         &&f_{a_+b_+}(t)=\text{tr}\lb Y^\dagger e^{-\beta H/2} X^\dagger(t) X(t) e^{-\beta H/2}Y\rb\geq 0\nn\\
         &&\hat{f}_{a_+b_+}(E_m-E_n)=|(Xe^{-\beta H/2}Y)_{mn}|^2\geq 0\ .
     \end{eqnarray}
     Choosing $a_+=b_+$ we obtain a positive positive-definite function.
     Such functions form a convex cone. For more information on these functions see \cite{jaming2009extremal}. 
     
 \end{enumerate}

{\bf Positive-definite functions:}

A complex function $f(t)$ of the real parameter $t$ is positive-definite if for any set of values $\{t_1,\cdots t_n\}$ the matrix $n\times n$ $\mathcal{F}_{ij}=f(t_i-t_j)$ is positive semi-definite. Since $\mathcal{F}$ must be Hermitian, a positive-definite function must satisfy the constraint $f(t)=f(-t)^*$. Moreover, specific cases of $n=\{1,2\}$ imply that $f(0)\geq |f_{aa}(t)|\geq 0$. They form a cone and, according to the Bochner theorem, the exponentials $e^{i t \lambda}$ are the only extremal rays of it. Therefore, positive definite functions are in one-to-one correspondence with functions with positive Fourier transform. It is known that even bounded functions that are convex on $\mathbb{R}^+$ are positive definite. 

We also consider positive-definite functions that are positive $f(t)>0$. Such functions also form a cone, and its extremal rays include Gaussians \cite{jaming2009extremal}. Other examples of positive positive definite functions include $(1+c|t|)^{-\alpha}$ with $\alpha\geq 0$ and $c>0$, and some Hermite polynomials. It is also known that for a positive positive-definite function $\omega(t)$ and bounded measure $\nu(t)$ on $\mathbb{R}^+$ the following function is also positive and positive-definite:
\begin{eqnarray}
    F(t):=\int_0^\infty \omega(t/\lambda) d\nu(\lambda)\ .
\end{eqnarray}

The set of positive-definite functions on a line corresponds to the set of functions that can be analytically continued to the upper half-plane. The set of positive-definite functions that are in $L^p$ is called the Hardy space $H^p$.

\section{Thermal two-point functions from gravity}\label{thermalFourier}

Consider the vacuum AdS$_{d+1}$ with unit AdS scale, i.e. $l_{AdS}=1$, and a compact boundary $\mathbb{R}_t\times\mathcal{S}^{d-1}$. The metric is
\begin{eqnarray}
    ds^2=-(r^2+1)dt^2+\frac{dr^2}{r^2+1}+r^2d\Omega^2\ .
\end{eqnarray}
In thermal AdS with no angular momentum (below the Hawking-Page phase transition) we can compute correlators by identifying the Euclidean time $t_E\sim t_E+\beta$ in the above geometry. Above the Hawking-Page phase transition, we have the AdS black holes
\begin{eqnarray}
    &&ds^2=-f(r,r_+) dt^2+\frac{dr^2}{f(r,r_+)}+r^2d\Omega^2\nn\\
    &&f(r,r_+)=r^2\lb 1-\frac{r_+^d+(r_+^{d-2}-r^{d-2})}{r^d}\rb
\end{eqnarray}
with the Euclidean time identification $t_E\sim t_E+\frac{4\pi r_+}{dr_+^2+(d-2)}$. AdS black holes with super-AdS-scale horizon radius $r_+> 1$ are semi-classically stable.\footnote{The inner horizon vanishes because we have no angular momentum.}

Choosing $d=2$ in the formula above gives the thermal AdS$_3$ and BTZ black holes with no angular momentum \cite{fitzpatrick2016information}:
\begin{eqnarray}
    ds^2=-(r^2-r_+^2)dt^2+\frac{dr^2}{r^2-r_+^2}+r^2d\phi^2
\end{eqnarray}
with the inverse temperature $\beta=2\pi/r_+$. Deficit angles correspond to the analytic continuation of the geometry to $r_+^2<0$. The Hawking-Page phase transition occurs at $\beta=2\pi$.

Since BTZ black holes are quotients of vacuum AdS$_3$ we can use the method of images to compute the correlators.
For conformal primaries $\mO$ of dimension $(h_-,h_+)$ above the Hawking-Page phase transition we find \cite{keski1999bulk}:
\begin{eqnarray}
    &&f_{\mO\mO}(\varphi,t)=\braket{\mO_L(\varphi,t)\mO_R(0,0)}_\beta\sim\sum_{n\in\mathbb{Z}} g(u_++2\pi n,h_+)g(u_--2\pi n,h_-)\nn\\
    &&g(u,h)=\sinh^{-2h}(u),\qquad u_\pm =\frac{\pi}{\beta}(t\pm\varphi)\ .
\end{eqnarray}
Then, in the Fourier space we have
\begin{eqnarray}
    \hat{g}(\omega)&&=C(h)\Gamma(h-i\omega/2)\Gamma(h+i\omega/2)e^{-\omega/2}\nn\\
    &&C(h)=\frac{4^\Delta e^{-ih}\sin(2h\pi)}{(2\pi)^{3/2}}\ .
\end{eqnarray}
We introduce Fourier frequencies $\omega_\pm=(\omega\pm k)\beta/\pi$ conjugate to $u_\pm$ to write
\begin{eqnarray}
      \hat{f}_{\mO\mO}(k,\omega)&&=\hat{g}(\omega_+)\hat{g}(\omega_-)\sum_{n\in\mathbb{Z}}e^{4\pi i n k}\nn\\
    &&\sim \prod_{\pm} e^{-\omega_\pm/2}\Gamma(h_\pm-i \omega_\pm)\Gamma(h_\pm+i \omega_\pm)
\end{eqnarray}
which is a smooth function of $\omega$ and $k$.

The mode $n=0$ corresponds to the planar BTZ. The two-point function in this geometry is dual to the thermal CFT$_2$ correlator in infinite volume, which is universally fixed by conformal symmetry. Therefore, in this case, the Fourier transform $\hat{f}_{\mO\mO}$ for any Virasoro primary $\mO$ is also a smooth function of $\omega$.

\section{Large $N$ SYM and GFF}\label{app:SYMGFF}

Consider $SU(N)$ $\mathcal{N}=4$ Super Yang Mills (SYM)  on $\mathcal{M}_d\times \mathbb{R}_t$ with the action normalized as\footnote{The discussion here generalizes to any compact spatial manifolds.}:
\begin{eqnarray}\label{HamilN4}
    S=\frac{-N}{2g_{YM}^2}\int dt\int_{\mathcal{M}_d} \sqrt{h}\: \text{tr}\lb \frac{1}{2}F_{\mu\nu}F^{\mu\nu}+(D_t\Phi_\alpha)^2+\frac{(d-1)\mathcal{R}}{4d}(\Phi_\alpha)^2+V \rb,
\end{eqnarray}
where $V$ represents the cubic and quartic terms, and $h$ and $\mathcal{R}$ are the metric and the corresponding Ricci scalar on the manifold $\mathcal{M}$.
With this normalization choice, the connected thermal $n$ point functions of single trace operators have values of order $N^{2-n}$. In the strict $N\to \infty$, subtracting the thermal (or vacuum) one-point function from any single trace operator $\mathbb{O}$ makes it well defined. That is because $\mathbb{O}=\mO-\braket{\mO}_\beta$ has finite fluctuations \cite{witten2022gravity}.
The state is a Gaussian of zero mean whose variance is set by the two-point function.
The higher-point functions are fixed by Wick contractions. The single trace field operators such as $\mathbb{O}$ form generalized free fields, because the correlators satisfy Wick's theorem, and the two-point function in the K\"all\'en-Lehmann representation is given by
\begin{eqnarray}\label{2ptGFF}
    G(t)=\int_0^\infty dm^2 \:\rho(m^2) G(t;m^2) 
\end{eqnarray}
where $G(t;m^2)$ is the vacuum two-point function on the cylinder for free fields of mass $m^2$. 

To simplify the presentation (without loss of generality), let us focus on a single $N\times N$ matrix real scalar field $\Phi$ and $\mathcal{M}=S^3$.\footnote{For quantization on $S^3$ see \cite{kim2003plane}. The Ricci scalar is $\mathcal{R}=6/R^2$ with $R$ the radius of $S^3$.}
The contribution of the scalar matrix $\Phi$ to the action in (\ref{HamilN4}) in terms of the mode expansion in $Y_{lm_1m_2}(x)$ the spherical harmonics of $S^3$ take the form
\begin{eqnarray}
    S_\Phi=\frac{\pi^2R^3 N}{g_{YM}^2}\sum_{l=0}^\infty\sum_{m_1,m_2=1}^{l+1}\int dt\:\text{tr}\lb (\p_t\Phi_{lm_1m_2})^2-\frac{(l+1)^2}{R^2}(\Phi_{lm_1m_2})^2+V(\Phi)\rb 
\end{eqnarray}
where the matrix degrees of freedom $\Phi_{lm_1m_2}$ have mass $m_l=(l+1)/R$.\footnote{The isometry group of $S^3$ is $SO(4)\simeq SU(2)_L\otimes SU(2)_R$. The scalars $\Phi_k$ transform in the irreducible $(l_L+1,m_L)\otimes (l_R+1,m_R)$ representations of the $SU(2)$s with mass $(l+1)/R$ where $l$ is a positive integer, and $m_1$ and $m_2$ run over the range $(-l, -1+1,\cdots, l-1,l)$.} Since the matrix field $\Phi$ is real we have $\Phi_{l,m_1,m_2}(t)=(\Phi_{l,-m_1,-m_2}(t))^\dagger$. We use the collective momentum notation $k=(l,m_1,m_2)$, and $-k=(l,-m_1,-m_2)$. 

Let us further focus on a single $k$ mode: a single $N\times N$ matrix quantum mechanics.
We quantize the theory at finite temperature covariantly by defining a tower of single-trace operators \footnote{The covariant quantization has the advantage that it also works for theories that violate the time-slice axiom like Generalized Free Fields.}
\begin{eqnarray}
&&\phi^{(n)}_k(t)=\text{tr}(\Phi^n_k(t))-\braket{\text{tr}(\Phi^n_k(t))}_\beta,\nn\\
&&\phi^{(n)}_k(\omega)=\int dt\:e^{i t\omega}\phi^{(n)}_k(t),
\end{eqnarray}
for positive integers $n$. In the Hilbert space of the matrix model, $\phi_k^{(n)}$ creates a single excitation. 
Acting with this operator $q$ times creates a $q$-particle excitation that in $N\to\infty$ belongs to $\mH_q$ of the Fock space $\mF$. The two-point function 
\begin{eqnarray}
    \braket{\phi^{(n)}_k(t)\phi^{(n)}_k(0)}=\braket{\delta(t)|\delta(0)}_\mu
\end{eqnarray}
corresponds to a choice of spectral function in the single-particle Hilbert space, and decides the type of the algebras at zero and finite temperatures.

\subsection{Two-point functions from the boundary}

In this subsection, we collect and review statements in the literature that attempt to justify the bulk two-point function from the boundary perspective. Most of what follows summarizes the work of Festuccia and Liu in \cite{festuccia2006excursions,festuccia2007arrow}.

{\bf Free theory:} Let us start with a theory of $D$ simple harmonic oscillators in the adjoint representation of $U(N)$. The degrees of freedom are $N\times N$ matrices $(M_\alpha)_{ij}$ and the Hamiltonian is 
\begin{eqnarray}
    H=\frac{N}{2}\sum_{\alpha=1}^D\text{tr}\lb (\p_t M_\alpha)^2+\omega_\alpha^2 M_\alpha^2\rb
\end{eqnarray}
and all $\omega_\alpha>0$.
Considering this as a gauge theory means that all physical states are $U(N)$ singlet; therefore, they can be expanded in terms of multi-trace operators. For instance, this theory arises as the free SYM on a sphere expanded in spherical harmonics. In that case, all $\omega_\alpha$ are positive integer (or half-integer) multiples of $1/R$ (the inverse radius of the sphere). Similarly, energy eigenvalues are positive half-integers with large degeneracies decided by partitions of integers. At high energies, i.e. $ER=O(N^2)$, the degeneracy is of order $e^{O(N^2)}$, while at low energies, i.e. $ER=O(1)$, the degeneracy is of order one.
In the $N\to \infty$ limit, the energy gaps do not change and Maldacena's argument implies that no correlator can cluster in time. The algebra of the observable remains type I in this limit. However, if we choose $\omega_\alpha$ such that $k$ of them are incommensurable, at high energies, we expect energy gaps of the order $O(N^{-2(k-1)})$ but still exponentially large degeneracies \cite{festuccia2007arrow}. In the limit $N\to \infty$ the energy gaps at high energies close as the power law in $N$, while the gaps at low energies remain. 
In Lemma \ref{simpleop}, we saw that if the matrix elements of the operator $a$ in the energy eigenbasis are smooth functions of only the energies $E-E'$ and $E+E'$ of the operator clusters. The algebra is no longer Type I, but the single-trace operators do not cluster.
Since the theory is free, the matrix elements of single-trace operators are nonzero only if $E=E'\pm 1$. 
Similarly, the multi-trace operators 
 \begin{eqnarray}
     \mO=\prod_k\text{tr}(M_{\alpha_{k,1}}\cdots M_{\alpha_{k,n_k}})
 \end{eqnarray}
 do not cluster either because even if we truncate them to the high energy sector, their matrix elements $\mO(E,\alpha;E+\delta,\beta)$ are non-zero only for a discrete set of $\delta E$  \cite{festuccia2007arrow}. Each mode adds or subtracts a unit of $\omega_\alpha$ for some $\alpha$, and since the total number of $M_\alpha$ terms in $\mathcal{O}$ is order $O(N^0)$ we will obtain a countable sum of delta functions
 \begin{eqnarray}
     \hat{f}_{aa}(\delta E)=\sum_k |f_k|^2 \delta(\delta E-\omega_k)\ .
 \end{eqnarray}
 The observables $\mO$ cannot probe arbitrarily small energy differences. In a free theory, one can use Wick contractions to explicitly compute the thermal two-point function, and see explicitly that for $\mO$ the function $\hat{f}_{\mO\mO}$ is a countable sum over Dirac delta functions, and hence non-measurable \cite{festuccia2007arrow}.
 To obtain observables that cluster we want $\hat{f}_{aa}(\delta E)$ to be measurable functions of $\delta E$. 
 We will see that as we turn interactions on, more operators, including $\mO$ cluster in time.

{\bf Interacting theory:} Next, we add an interaction term to the Hamiltonian
\begin{eqnarray}
    H\to H+N V(M_\alpha,\lambda)\ .
\end{eqnarray}
In SYM, $\lambda=g_{YM}^2 N$ is the t'Hooft coupling, and we can schematically expand $V$ as
\begin{eqnarray}
    V=\sqrt{\lambda}V_3(M_\alpha)+\lambda V_4(\lambda)
\end{eqnarray}
where $V_3$ and $V_4$ are cubic and quartic in $M_\alpha$, respectively. Festuccia and Liu computed Feynman diagrams and showed that turning on a small $\lambda$ is non-perturbative at high temperatures. They showed that the large $t$ limit does not commute with $\lambda\to 0$, resulting in a radius of convergence for the perturbation series of the high-temperature thermal two-point function that shrinks as $1/t$. They also qualitatively argued that $\hat{f}_{\mO\mO}(\delta E)$ becomes continuous and no-where vanishing at small $\lambda$ based on the following expectations:
\begin{itemize}
    \item A generic small perturbation lifts all degeneracies resulting in level spacings that are $e^{-O(N^2)}$.
    \item An interacting energy eigenstate $\ket{i}$ expanded in the free energy eigenbasis is expected to be supported on an energy shell of width $O(N)$ with $e^{O(N^2)}$ states.
\end{itemize}

This matches the continuous spectrum one finds in gravity.
Therefore, in the high-temperature phase (deconfined phase) in the limit $N\to \infty$ the energy gaps vanish and $\hat{f}_{\mO\mO}$ is expected to become a continuous function (almost everywhere) and hence measurable. Interestingly, in the $N\to\infty$ limit, the set of single trace operators generates an algebra. In particular, a single trace operator $\mO$ of conformal dimension $\Delta$ acting on the thermal field double generates an algebra of conformal free fields that we call $\mathcal{Y}_\mO$.

\section{GFF with discrete spectrum}\label{app:GHO}

 In this appendix, we start with a model that is a collection of harmonic oscillators as a model of GFF in $0+1$ dimensions and show that it is equivalent to a Gaussian quantum field in time, i.e. $\varphi(t)$.

Consider a set of $k$ simple harmonic oscillators of different masses $m$.\footnote{The generalization to the case where some modes have the same mass is straightforward.} The Hilbert space is $\mH=\otimes_m \mH_m$ and the algebra is $\mA=\otimes_m \mA_m$ and isomorphic to those of a single harmonic oscillator (type I$_\infty$ von Neumann algebra). 
The creation $a_m^\dagger$ and annihilation operator $a_m$ can be packaged together as the field $\varphi_m(t)=\frac{1}{\sqrt{2m}} (a_m^\dagger e^{imt}+a_m e^{-imt})$. The algebra $\mA$ is generated by exponentiating bounded function of $\varphi(f)$ where we have defined a {\it fundamental field} in time $\varphi(t)$
\begin{eqnarray}
    &&\varphi(t):=\sum_{m>0} \frac{1}{\sqrt{2m}}(e^{i mt} a_m^\dagger+ e^{-i m t}a_m),\qquad [a_m,a_{m'}^\dagger]=\delta_{mm'}\nn\\
    &&\varphi(f)=\int dt \varphi(t) f(t)=\sum_{m>0}\frac{1}{\sqrt{2m}} (\hat{f}_m a_m^\dagger+ \hat{f}_{-m} a_m)
    \end{eqnarray}
and $f$ is any smooth real function of time. 

For simplicity, let us focus on the first two levels of each harmonic oscillator as a qubit, i.e. $\ket{0}_m$ and $\ket{1}_m\sim a_m^\dagger\ket{0}_m$. Sequences of binaries $\ket{0100\cdots 0}$ of length $k$ form a basis for this sector of the Hilbert space. We call this sector the {\it one-particle Hilbert space}. When the number of harmonic oscillators goes to infinity, i.e. $k\to\infty$, the one-particle Hilbert space is not separable. To define a separable Hilbert space, we need to reduce the number of allowed observables. From the Araki-Woods construction of tensor powers \cite{araki1963representations,witten2018aps}, we know how to obtain the von Neumann algebra that can be represented in the separable Hilbert space. We choose a state $\psi$, for instance the density matrix $\rho_\psi$:
\begin{eqnarray}
    \rho_\psi=\otimes_m \rho_m, \qquad \rho_m=p_m \ket{0}\bra{0}_m+(1-p_m)\ket{1}\bra{1}_m\ .
\end{eqnarray}
Then, we consider the $*$-algebra corresponding to all operators that are finite tensor products $a=\otimes_{m=1} a_m$ and close it under the weak norm topology. This means that for every sequence $\{a_m\}$ such that $\|a_m \ket{\rho_\psi^{1/2}}\|<\infty$ we add to the $*$-algebra an operator $a=\lim_m a_m$ with the corresponding limit points. 

To perform the Araki-Woods construction in our case, we need to restrict the functions $f$ to those that create normalizable states:
\begin{eqnarray}\label{condsum}
 \braket{\rho_\psi^{1/2}||\varphi(f)|^2\rho_\psi^{1/2}}=\sum_{m>0}\frac{1}{2m}|f_m|^2<\infty\ .
\end{eqnarray}
An alternative way of performing the Araki-Woods construction is to redefine the fundamental field to the generalized free field: 
\begin{eqnarray}
    &&\varphi_1(t)=\sum_{m>0}(e^{im t}a_m^\dagger+e^{-i mt}a_m)\nn\\
     &&\varphi_1(\tilde{f})=\sum_{m>0} \frac{1}{\sqrt{2m}}(\hat{f}_ma_m^\dagger+ \hat{f}_{-m}a_m)
\end{eqnarray}
where we have absorbed $\frac{1}{\sqrt{2m}}$ in the definition of the function $\tilde{f}_m=\hat{f}_m/\sqrt{2m}$.
Consider the set of vectors $\ket{f}$ equipped with the Hermitian bilinear 
\begin{eqnarray}\label{innerproduct}
    &&\braket{f|g}_\psi=\braket{\rho_\psi^{1/2}|\varphi(f)\varphi(g)|\rho_\psi^{1/2}}=\sum_{m>0}\frac{1}{2m}\lb p_m\hat{f}_m^* \hat{g}_m+(1-p_m)\hat{f}_{-m}^*\hat{g}_{-m}\rb\nn\\
    &&2\Re\braket{f|g}_\psi=\sum_m\frac{1}{2|m|}\hat{f}^*_m g_m,\qquad 2i\Im\braket{f|g}_\psi=\braket{f|g}_\psi-\braket{g|f}_\psi\ .
\end{eqnarray}
Taking the quotient by the kernel of this bilinear gives us an inner product. Closing the vector space in the inner product defines the one-particle Hilbert space $\mH_\psi$. The imaginary part of the inner product defines the commutator (a symplectic bilinear on vectors $\ket{f}$) whereas the norm it induces corresponds to a choice of state (a symmetric bilinear on $\ket{f}$). Closing the linear span of the vectors $\ket{f}$ with this choice of the inner product gives the one-particle Hilbert space $\mH_\psi$.

 Next, we construct the Fock space $\mF_\psi=\oplus_n \mH_\psi^{\otimes n,sym}$. The second quantization map 
\begin{eqnarray}
&&W(f):=e^{i \varphi(f)}
\end{eqnarray}
sends  vectors to  $\ket{f}\to W(f)$:
we have a natural $*$-operation $W(f)^\dagger=W(-f)$. They satisfy the Weyl algebra
\begin{eqnarray}
    &&W(f)W(g)=e^{-\frac{1}{2}[\varphi(f),\varphi(g)]}W(f+g)\nn\\
    &&[\varphi(f),\varphi(g)]=2i\Im\braket{f|g}_\psi\ .
\end{eqnarray}
The set of formal path-integrals $\int Df\: c(f) W(f)$ with complex $c(f)$ form a $*$-algebra.

Every bounded function of $\varphi(f)$ is a bounded operator and has a spectrum. Its spectral radius defines a C$^*$-norm (i.e., a norm that satisfies $\|\ba\|=\|\ba^\dagger\|$). Taking the closure in this norm gives a C$^*$-algebra.\footnote{This C$^*$-norm can be constructed using interpolation theory of norms $\psi((a^\dagger a)^p)^{1/p}$ or by taking a supremum over the class of states compatible with the symplectic bilinear $\sigma$ \cite{binz2004construction}.} Taking the double commutant of the Weyl $*$-algebra or its corresponding C$^*$-algebra in $B(\mF_\psi)$ gives a von Neumann algebra $\mA$.

An equivalent approach to the procedure described above is to separate the choice of a commutator from the choice of state. We start with the real vector space of $\ket{f}$ with a symplectic bilinear $\sigma(f,g)$. The choice of quasi-free states then corresponds to a symmetric bilinear $\mu_\psi$ on the vector space of $\ket{f}$:
\begin{eqnarray}
    &&\psi(W(f))=e^{-\mu_\psi(f,f)},
\end{eqnarray}
such that
\begin{eqnarray}
  &&\mu_\psi(f,f)\geq 0\nn\\
  &&|\sigma(f,g)|^2\leq \mu_\psi(f,f)\mu_\psi(g,g)\ .
\end{eqnarray}
The second requirement (compatibility with $\sigma$) comes from the Cauchy-Schwarz inequality
\begin{eqnarray}
    |\psi(W(f)W(g))|\leq |\psi(W(f))||\psi(W(g))|\ .
\end{eqnarray}
The choice of the commutator $\sigma$ and the state $\mu$ can be packaged together as the inner product in (\ref{innerproduct}) on a complexified vector space $\mathcal{K}$
such that 
\begin{eqnarray}
    &&\sigma(f,g)=2i\Im \braket{f|g}_\psi=\psi([\varphi(f),\varphi(g)])\nn\\
    &&\mu_\psi(f,f)=\braket{f|f}_\psi\ .
\end{eqnarray}

The Gaussian states (quasi-free states) are fixed by the two-point function of the fundamental field $\varphi$. For instance, the typical choice of Gaussian states that correspond to the zero-temperature vacuum is
\begin{eqnarray}\label{statesexample}
    &&\psi_0(\varphi(0)\varphi(t))=\sum_{m>0}\frac{p_m}{2m} e^{im t}\ .
\end{eqnarray}

\bibliographystyle{ieeetr}
\bibliography{main}
\end{document}